\documentclass[12pt]{article}
\usepackage{graphicx,psfrag,epsf}
\usepackage{enumerate}
\usepackage{natbib}
\usepackage{url} 
\usepackage{amsmath,amssymb}
\usepackage{dsfont}
\usepackage{epstopdf}
\usepackage[hang]{subfigure}
\usepackage{booktabs}
\usepackage{textcomp}
\usepackage{ulem}
\usepackage{color}
\usepackage{threeparttable}
\usepackage{comment}
\usepackage{float}
\usepackage{indentfirst}
\usepackage{algorithm,algorithmic}
\usepackage{tikz-cd}
\usepackage{tikz}
\usepackage{ntheorem}
\usepackage{newfloat} 
\usepackage{authblk}
\usepackage{enumitem}
\usepackage{pifont}
\newcommand{\cmark}{\ding{51}}

\DeclareFloatingEnvironment[fileext=loa, name=example]{example1}
\DeclareFloatingEnvironment[fileext=lob, name=Example]{example2}
\DeclareFloatingEnvironment[fileext=loc, name=Example]{example3}
\DeclareFloatingEnvironment[fileext=lod, name=Example]{example4}

\newcommand{\blind}{1}

\newcounter{regime}

\newcounter{mystep}

\newtheorem{coro}{Corollary}
\newtheorem{theorem}{Theorem}

\newtheorem{lemma}{Lemma}
\newtheorem{definition}{Definition}
\newtheorem{prop}{Proposition}

\newtheorem{condition}{Condition}
\newtheorem{remark}{Remark}

\newtheorem*{proof}{\it Proof.}

\def\K{K_N}

\def\prodk{\prod_{k=1}^{\K}}

\def\T{{\mathrm{\scriptscriptstyle T} }}
\def\Nk{N_{[k]}}
\def\nk{n_{[k]}}
\def\nkt{n_{[k]1}}
\def\nkc{n_{[k]0}}
\def\Pik{\Pi_{[k]}}
\def\pik{\pi_{[k]}}
\def\fk{f_{[k]}}
\def\ekt{e_{[k]1}}
\def\ekc{e_{[k]0}}
\def\sumk{\sum_{k=1}^{\K}}
\def\tauk{\tau_{[k]}}
\def\taukhat{\hat \tau_{[k]}}

\def\sumik{\sum_{i \in [k]}}
\def\sumikS{\sum_{i \in [k]\cap \mathcal S}}

\def\cov{\textnormal{cov}}
\def\barW{\bar W}
\def\barWS{\bar W_{\mathcal S}}
\def\barWSk{\bar W_{[k]\mathcal S}}
\def\barXkt{\bar{X}_{[k]1}}
\def\barXkc{\bar{X}_{[k]0}}
\def\barWkt{\bar{W}_{[k]1}}
\def\barWkc{\bar{W}_{[k]0}}
\def\barXt{\bar{X}_{1}}
\def\barXc{\bar{X}_{0}}
\def\barYt{\bar{Y}_{1}}
\def\barYc{\bar{Y}_{0}}
\def\barRt{\bar{R}_{1}}
\def\barRc{\bar{R}_{0}}
\def\MS{\mbox{M}_{S}}
\def\MT{\mbox{M}_{T}}

\def\barYkt{\bar{Y}_{[k]1}}
\def\barYkc{\bar{Y}_{[k]0}}
\def\pr{\textnormal{pr}}
\def\cd{\xrightarrow{d}}
\def\var{\textnormal{var}}
\def\skt{s_{[k]1}}
\def\skc{s_{[k]0}}
\def\skXt{s_{[k]X,1}}
\def\skXc{s_{[k]X,0}}
\def\skWt{s_{[k]W,1}}
\def\skWc{s_{[k]W,0}}
\def\op{o_p(1)}
\def\tauadj{\hat \tau_{\textnormal{adj}}}
\def\opt{\textnormal{opt}}
\def\tauunadj{\hat \tau}
\def\asyequ{\mathrel{\ \dot{\sim}\ }}

\DeclareMathOperator*{\argmin}{arg\,min}

\def\Vcr{V_{\tau\tau,\textnormal{C}}}
\def\Vsr{V_{\tau\tau}}
\def\TV{\mathrm{TV}}
\def\ssrsrr{\widetilde{\text{SRSRR}}}
\def\srsrr{{\text{SRSRR}}}
\def\tauora{\hat\tau_{\mathrm{ora}}}

\newcounter{case}
\newcommand{\case}[1]{\refstepcounter{case}(\textbf{Case \thecase}\label{#1})}

\usepackage{hyperref}
\hypersetup{
    colorlinks=true,
    urlcolor=blue,
    citecolor=blue
}

\begin{document}

\def\spacingset#1{\renewcommand{\baselinestretch}
{#1}\small\normalsize} \spacingset{1}

\if1\blind
{
  \title{\bf Balancing Covariates in Survey Experiments}

\author[1]{Pengfei Tian}
\author[2]{Jiyang Ren}
\author[3]{Yingying Ma\thanks{Correspondence should be addressed to Yingying Ma (\texttt{\href{mailto:pengdingpku@berkeley.edu}{mayingying@buaa.edu.cn}})}}

\affil[1]{\small Qiuzhen College, Tsinghua University}
\affil[2]{\small Department of Statistics and Data Science,
    Tsinghua University}
\affil[3]{\small School of Economics and Management, Beihang University}

    \date{}
  \maketitle
}\fi
	
\if0\blind
{
	\title{\bf BALANCING COVARIATES IN SURVEY EXPERIMENTS}
	\maketitle
}\fi

\bigskip

\begin{abstract}
The survey experiment is widely used in economics and social sciences to evaluate the effects of treatments or programs. In a standard population-based survey experiment, the experimenter randomly draws experimental units from a target population of interest and then randomly assigns the sampled units to treatment or control conditions to explore the treatment effect of an intervention. Simple random sampling and treatment assignment can balance covariates on average. However, covariate imbalance often exists in finite samples. To address the imbalance issue, we study a stratified approach to balance covariates in a survey experiment. A stratified rejective sampling and rerandomization design is further proposed to enhance the covariate balance. We develop a design-based asymptotic theory for the widely used stratified difference-in-means estimator of the average treatment effect under the proposed design. In particular, we show that it is consistent and asymptotically a convolution of a normal distribution and two truncated normal distributions. This limiting distribution is more concentrated at the true average treatment effect than that under the existing experimental designs. Moreover, we propose a covariate adjustment method in the analysis stage, which can further improve the estimation efficiency. Numerical studies demonstrate the validity and improved efficiency of the proposed method.
\end{abstract}


\newpage

\spacingset{1.9} 

\section{Introduction}
The survey experiment, also referred to as the experiment embedded within a survey, has been viewed as the gold standard for estimating the treatment effect for a target population of interest \citep{mutz2011population}. Following \cite{mutz2011population}, the survey experiment in this article is the abbreviation for a population-based survey experiment. It contains two stages: ``survey" means the sampling stage, i.e., the experimenter randomly samples units from the target population, and ``experiment" means the treatment assignment stage, i.e., the experimenter randomly assigns the sampled units to the treatment or control condition to explore the effect of an intervention. In the past decades, the survey experiment has gained increasing popularity in many fields, such as political science, education, and economics, because it is easy to implement and clear to distinguish cause and effect for the target population. For example, in the 1980s, the US Department of Labor conducted the Pennsylvania Reemployment Bonus survey experiment to test the incentive effects for unemployment insurance \citep{bilias2000sequential}. There is also another example that researchers at Boston University conducted a survey experiment to investigate the influence of public opinion on citizen perception \citep{dancey2018partisanship}. 
The survey experiment provides grounded inferences about real-world behavior based on a representative sample \citep{survey2007reexam} and can effectively solve the problem of insufficient external validity of traditional one-stage randomized experiments \citep{mutz2011population}, which lack the stage of random sampling of a subset of units into the experiments.


Simple random sampling and complete randomization in a survey experiment balance observed and unobserved confounding factors on average and justify simple comparisons of average outcomes among the treatment and control groups. However, covariate imbalance happens not only between the sampled experimental units and the overall population of interest but also between treatment and control groups \citep{yang2021rejective}. In the survey sampling literature, stratified random sampling is widely used to balance covariates, generate a more representative sample, and increase the efficiency of a sample design concerning cost and precision \citep{imbens1996efficient}. Stratification or blocking is also used to balance covariates in the treatment assignment stage \citep{Fisher1926,Imbens2015}. According to a recent survey \citep{Lin2015}, stratification has been prevalently used in randomized clinical trials.
In recent studies, researchers introduced a finely stratified design approach grounded in the semi-parametric efficiency principle and super-population framework \citep{bai2022inference,cytrynbaum2023optimal,Tabord2023}.




Design-based or finite-population asymptotic results for a general stratification strategy in the survey experiment have not been fully explored. The design-based inference only considers the randomness of sampling and treatment assignment, with potential outcomes and covariates being fixed. It can be seen as a form of conditional inference within the super-population framework, i.e., drawing causal inferences conditional on the potential outcomes and covariates. The design-based inference can be traced back to Fisher and Neyman \citep{Neyman1923,Fisher1935} and is gaining increasing popularity in causal inference theory and practice \citep{lin2013,Imbens2015,Li2016,Liu2019,lu2022design}. 
 
Our first contribution is to establish the design-based asymptotic theory for the stratified difference-in-means estimator of the average treatment effect under the stratified randomized survey experiment. Specifically, we show that this estimator is consistent and asymptotically normal, and that its asymptotic variance is usually smaller than that under a completely randomized survey experiment. Hence, stratification can improve efficiency. We also obtain the optimal stratum-specific sampling proportion and treated proportion to minimize the asymptotic variance of the stratified difference-in-means estimator. Note that we consider a general asymptotic regime that allows the sampling fraction to tend to zero, the proportions of treated units to vary across strata, as well as the number of strata or the associated stratum sizes to diverge. Thus, our theoretical results cover the scenarios of a few large strata, many small strata, and a combination thereof. {To establish a central limit theorem (CLT) for our estimator, we develop a Hájek coupling technique that reduces the estimator to a sum of independent random variables. Our setting introduces additional challenges due to the presence of two sources of design randomness, sampling and treatment assignment, which are not jointly handled in existing work.

To address this, we recast the joint design as a unified one-stage within-stratum permutation that partitions units into treated, control, and unsampled groups. Unlike standard one-stage asymptotic theory, which requires group proportions to be bounded away from zero and one, our framework allows the sampling fraction to vanish. We further develop a refined large-/small-strata decomposition, building on \cite{bickel1984asymptotic}, and introduce a threshold that ensures uniform control of the coupling remainders across strata. This enables us to establish a general CLT under diverging strata and vanishing sampling fractions.}

Rejective sampling \citep{fuller2009reject} and rerandomization \citep{Morgan2012,Li2020} are more general approaches to balance covariates. This approach needs covariates information from another survey or previous research to calibrate sampling or assignment procedure to decrease the chance imbalance.
Recently, \cite{yang2021rejective} proposed a rejective sampling and rerandomized experimental design to avoid covariate imbalance at both the sampling and treatment assignment stages in the survey experiment. They also took into account a stratified design but provided that the sampling ratios and propensity scores are asymptotically the same across strata. \cite{Wang2021} considered the combination of stratification and rerandomization as suggested by Donald Rubin and verified that such a design can achieve better computational and efficiency performance, but their discussion is only for one-stage experiments. Stratification, rejective sampling, and rerandomization are easy to interpret and implement, and intuitively, applying these three methods together can further reduce the covariate imbalance and improve the computational and estimation efficiency. However, limited studies have addressed the statistical properties of the combination in the survey experiment under general settings, allowing the diverging strata number or both large strata and small strata to exist.



Our second contribution is to introduce an innovative Stratified Rejective Sampling and ReRandomized (SRSRR) experimental design and provide a covariate adjustment procedure within the SRSRR design’s analysis stage. This design inherits the advantages of traditional stratification in exploring subpopulation properties and organizational convenience \citep{cochran1977}, and simultaneously enhances covariate balance and improves the efficiency of treatment effect estimation. We establish the design-based asymptotic theory for the stratified difference-in-means estimator under the proposed SRSRR experiment and show that its limiting distribution is no longer normal but a convolution of a normal distribution and two truncated normal distributions. This limiting distribution is more concentrated at the true average treatment effect than that under the stratified randomized survey experiment. Thus, it verifies that SRSRR can further improve the estimation efficiency for treatment effects. Moreover, we provide a conservative estimator for the asymptotic distribution to facilitate valid inferences. Lastly, we propose a covariate adjustment method to further improve the estimation efficiencies. The validity and improved efficiency of the proposed methods are demonstrated through numerical studies. 

\section{Stratified rejective sampling and rerandomization}\label{sec:srsrr}

\subsection{Stratified population-based survey experiment}\label{srse}

To eliminate the potential risk of chance imbalance and conduct the experiment more conveniently, researchers often stratify the population according to important covariates that are predictive of the outcomes and conduct a stratified randomized survey experiment.
{Let $N$ denote the size of a finite target population. Units are partitioned into $\K$ strata indexed by $k=1,\ldots,\K$, where stratum $k$ contains $\Nk$ units and $\sum_{k=1}^{\K}\Nk = N$. Throughout the paper, we use a subscript $[k]$ to denote stratum-specific quantities.} The stratified randomized survey experiment consists of two stages: (1) sampling stage: randomly sample $n$ units from the target population by stratified random sampling without replacement with $\nk$ units sampled in stratum $k$; and (2) treatment assignment stage: the sampled units are randomly assigned to the treatment and control groups using stratified randomization. Within each stratum $k$, $\nkt$ units are randomly selected to receive the treatment, and the remaining $\nkc = \nk - \nkt$ units are designated for the control group. The treatment assignments across strata are independent, and the total sample size $n$, the sample size $\nk$ in stratum $k$, and the number of treated units $\nkt$ in stratum $k$ are fixed and satisfy $\sumk \nk = n$ with $2 \leq \nkt \leq \nk - 2$ for $k=1,\ldots,\K$.



We use $\mathcal{S}$ to denote the set of sampled units. For each unit $i$ ($i=1,\ldots,N$), let $Z_i$ be the sampling indicator; $Z_i=1$ if unit $i$ is sampled and $Z_i=0$ otherwise. For $i \in \mathcal{S}$, let $T_i$ be a binary treatment assignment indicator with $T_i=1$ if unit $i$ is assigned to the treatment group and $T_i=0$ otherwise. 
Define $\Pik =\Nk/N$ as the proportion of stratum size in the target population and $\pik =\nk/n$ as the proportion of stratum size in the sample. In practice, we often employ the proportionate stratified sampling with $\pik=\Pik$ for $k=1,\ldots,\K$ \citep{cochran1977}.

We adopt the potential outcomes framework to define the treatment effect.
For unit $i$, let $Y_i(1)$ and $Y_i(0)$ denote the potential outcomes under treatment and control, respectively. The unit-level treatment effect is $\tau_i = Y_i(1) - Y_i(0)$, and the average treatment effect is $\tau = N^{-1}\sumk\sum_{i \in [k]}\tau_i = \sumk\Pik \tau_{[k]}$, where $\tau_{[k]} = \Nk^{-1} \sum_{i \in [k]} \tau_i$ is the average treatment effect in stratum $k$ and $i \in [k]$ indexes unit $i$ in stratum $k$. The observed outcome is $Y_i = T_i Y_i(1) + (1 - T_i) Y_i(0)$ for $i \in \mathcal{S}$.
The difference in the sample means of the outcomes between the treatment and control groups within that stratum is an unbiased estimator for $\tau_{[k]}$:
$$
\taukhat = \barYkt - \barYkc =  \nkt^{-1}\sum_{i\in[k]}Z_i T_i Y_i - \nkc^{-1}\sum_{i\in[k]}Z_i (1-T_i) Y_i.
$$  
By plugging in, we further obtain an unbiased estimator for $\tau$ as
$\hat\tau= \bar{Y}_1 - \bar{Y}_0 = \sumk\Pik \hat\tau_{[k]},$ 
where $\bar{Y}_1 = \sumk\Pik \barYkt  $ and $\bar{Y}_0 = \sumk\Pik \barYkc  $.


\subsection{Stratified rejective sampling and rerandomization}

Stratified random sampling and stratified randomization can balance discrete covariates that are most predictive of the outcomes. However, additional covariates aside from the stratification variables may remain unbalanced. 
In the sampling stage, if these covariates are not well-balanced such that the sampled units are not representative, the estimation accuracy on the target population decreases \citep{banerjee2017decision}. In the treatment assignment stage, the covariate imbalance between the treatment and control groups may lead to large variability and conditional bias. Thus, balancing additional covariates at both stages is desirable. Suppose that we collect covariates $W_i\in \mathbb{R}^{J_1}$ at the sampling stage from another survey or previous research. After the sampling stage, only a small proportion of the population is sampled, and more covariates of the sampled units, denoted by $X_i \in \mathbb{R}^{J_2}$, can be collected at the treatment assignment stage. Next, we introduce how to use SRSRR to balance $W_i$ and $X_i$. SRSRR consists of the following two stages:


{\it Stratified rejective sampling.} In this stage, to make the sampled units more representative of the target population, we can reject the samples that result in covariate imbalance and repeat the stratified random sampling procedure until we obtain a sample with satisfactory covariate balance between the sampled units and the population \citep{fuller2009reject}. Let $\bar W=N^{-1} \sum_{i=1}^{N} W_i = \sumk \Pik \bar W_{[k]}$  and $\barWS =\sumk \Pik \barWSk$ denote the population mean and weighted sample mean, respectively, where   $\bar W_{[k]} = \Nk^{-1} \sumik W_i$ and $\barWSk = \nk^{-1} \sumik Z_i W_i$ are the stratum-specific population mean and sample mean, respectively.
We can use the Mahalanobis distance between $\bar W$ and $\barWS$ to measure the covariate imbalance, which has the advantage of being affinely invariant \citep{Morgan2012}.  
The Mahalanobis distance is defined as 
$$\MS=(\barWS-\bar{W})^\mathrm{T} \Big\{\sumk \Pik ^2 \Big( \frac{1}{\nk }-\frac{1}{\Nk} \Big)S_{[k]W}^2 \Big\}^{-1}(\barWS-\bar W),$$
where $S_{[k]W}^2 = (\Nk - 1)^{-1} \sumik ( W_i - \bar{W}_{[k]} ) ( W_i - \bar{W}_{[k]} )^\T $ is the stratum-specific population covariance of $W_i$. The acceptability of the sample set is determined by the condition $\MS \leq a_S$, where $a_S>0$ is a pre-specified fixed threshold.  We repeat the stratified sampling procedure until the sample set satisfies this acceptability criterion. We can choose $a_S$ such that the acceptance probability achieves a given level.


{\it Stratified rerandomization.} After sampling, we balance the covariates of the sampled units between treatment and control groups by stratified rerandomization \citep{Wang2021}.
Specifically, for the sampled units, let $\barXt=\sumk\Pik \barXkt$ and $\barXc=\sumk\Pik \barXkc$ be the weighted sample means of the covariates $X_i$ under the treatment and control conditions, respectively, where $\barXkt = \nkt^{-1} \sumik Z_i T_i X_i$ and $\barXkc = \nkc^{-1} \sumik Z_i (1 - T_i) X_i$ are the stratum-specific sample means of $X_i$ under treatment and control in stratum $k$. Let $\hat\tau_X=\barXt-\barXc$. The Mahalanobis distance between $\barXt$ and $\barXc$ given the sample $\mathcal S$ is defined as
$$\MT =\tauunadj_X^\mathrm{T} \cov (\tauunadj_X \mid \mathcal S)^{-1}\tauunadj_X=(\barXt-\barXc)^\mathrm{T} \Big\{\sumk \Pik^2\frac{\nk }{\nkt \nkc }S_{[k]X \mid \mathcal S}^2 \Big\}^{-1}(\barXt-\barXc),$$
where $S_{[k]X \mid \mathcal S}^2 = (\nk - 1)^{-1} \sumik Z_i (X_i - \bar{X}_{[k]\mathcal S})(X_i - \bar{X}_{[k]\mathcal S})^\T$ is the covariance of $X_i$ conditional on the sample set $\mathcal S$  with $\bar{X}_{[k]\mathcal S} = \nk^{-1}\sumik Z_i X_i$. The treatment assignment is acceptable if and only if $\MT \leq a_T$ for a pre-specified fixed $a_T > 0$. 
We discard the undesired treatment assignments and repeat the stratified randomization until the treatment assignment is acceptable.
To ensure covariate balance and powerful design-based inference, we can choose $a_T$ such that the probability of a treatment assignment being acceptable achieves a given level, such as $0.01$ or $0.001$ as suggested by \cite{Morgan2012}.

{
\begin{remark}
 While we implement stratified rerandomization using the Mahalanobis distance, a large class of alternative rerandomization methods can be accommodated by replacing $\MT$ with a general quadratic-form balance metric. Specifically, many rerandomization criteria can be written as an ellipsoidal acceptance rule
$Q_A(\hat\tau_X)=\hat\tau_X^\top A\,\hat\tau_X \le a_T$ with $A\succeq 0,$
which includes tiered rerandomization \citep{Morgan2015}, ridge rerandomization \citep{branson2021ridge}, PCA rerandomization \citep{zhang2024pca}, and joint-test based acceptance rules that can be expressed via quadratic forms \citep{zhao2024no}. Recent work by \citet{schindl2024unified} develops a general framework for rerandomization with ellipsoidal (quadratic-form) acceptance regions, characterizes the resulting covariance reduction along the covariates' eigen-directions, and provides guidance and optimality results for selecting $A$ in practice. In particular, the Mahalanobis choice maximizes the total variance reduction across covariate directions, whereas the Euclidean choice $A=I$ minimizes the Frobenius norm of the covariance matrix after quadratic form rerandomization. It would be interesting to generalize these optimality results to the SRSRR framework.
\end{remark}}

{When $K=1$, our SRSRR framework reduces to \citet{yang2021rejective}; when $f=1$, it reduces to \citet{Wang2021}. More broadly, SRSRR clarifies the division of labor across stages: stratification and rejective sampling operate at the sampling stage to improve representativeness and efficiency, whereas rerandomization operates at the assignment stage to improve covariate balance; thus, the two tools are complementary rather than interchangeable. This separation also provides a computational benefit, since stratification shrinks the assignment space for rerandomization \citep{Schultzberg2019,Wang2021}.}

\section{Design-based asymptotic theory}\label{sec:asym}

In this section, we investigate the asymptotic distribution of $\tauunadj$ under the stratified randomized survey experiment (SRSE) and the SRSRR experiment, respectively. 


\subsection{Asymptotic properties under SRSE}

We introduce some notations first. Let $n_1=\sumk \nkt$ and $n_0 = \sumk \nkc$ be the total numbers of units assigned to the treatment and control groups, respectively.  Let $f= n/N$, $e_1=n_1/n$, and $e_0 = n_0 / n$ be the total proportions of sampled, treated, and control units, respectively.  Accordingly, for $k=1,\ldots,\K$, let $\fk =\nk /\Nk$, $\ekt =\nkt /\nk $, and $\ekc = \nkc / \nk $ be the proportions of sampled, treated, and control units in stratum $k$, respectively.  
Within stratum $k$, denote $\bar Y_{[k]}(t) = (1/\Nk) \sumik Y_i(t)$ for   $t=0,1$ and $\bar{X}_{[k]}=(1/\Nk) \sumik X_i$  as  the stratum-specific finite-population means of $Y_i(t)$ and $X_i$, respectively. Define  $
S^2_{[k]t} = (\Nk - 1)^{-1} \sumik \{ Y_i(t) - \bar Y_{[k]}(t) \}^2
$ and $S^2_{[k]\tau} =  (\Nk - 1)^{-1} \sumik ( \tau_i - \tauk )^2 $  as the stratum-specific finite-population variances of $Y_i(t)$ and $\tau_i$. Similarly,  we use  $S_{[k]X}^2 = (\Nk - 1)^{-1} \sumik ( X_i - \bar X_{[k]} ) ( X_i - \bar X_{[k]} )^\T $ and $
S_{[k]X,t} = S_{[k]t,X}^\T =(\Nk - 1)^{-1} \sumik ( X_i - \bar X_{[k]} ) \{ Y_i(t) - \bar Y_{[k]}(t) \}
$ to denote the stratum-specific finite-population covariance of $X_i$ and covariance between  $X_i$ and $Y_i(t)$, respectively.
Substituting 
$X_i$  with $W_i$, we can analogously  define   $S_{[k]W}^2$ and $S_{[k]W,t} = S_{[k]t,W}^\T$. Let $S_{[k]W,\tau} = S_{[k]\tau,W}^\T = (\Nk - 1)^{-1} \sumik  ( W_i - \bar W_{[k]} ) ( \tau_i - \tauk )$.
Denote $ I_k$ as a $k \times k$ identity matrix and $0_{J_1 \times J_2}$ as  a $J_1 \times J_2$ zero matrix.  Let $\|\cdot\|_\infty$ and $\|\cdot\|_2$ be the $\ell_\infty$ and $\ell_2$ norms of a vector, respectively. 
Define $V_{[k]}$ as the matrix given by
$$
\left(\begin{matrix}\ekt ^{-1} S_{[k]1}^2+\ekc ^{-1} S_{[k]0}^2-\fk S_{[k]\tau}^2&&\ekt ^{-1}S_{[k]1,X}+\ekc ^{-1}S_{[k]0,X}&&(1-\fk )S_{[k]\tau,W}\\\ekt ^{-1}S_{[k]X,1}+\ekc ^{-1}S_{[k]X,0}&&(\ekt \ekc )^{-1}S_{[k]X}^2&&0_{J_2\times J_1}\\(1-\fk )S_{[k]W,\tau}&& 0_{J_1\times J_2}&&(1-\fk )S_{[k]W}^2\end{matrix}\right).
$$
Proposition~\ref{prop Cov} in the Supplementary Material shows that 
$$
V: = \textnormal{cov} \big\{ \sqrt n (\hat\tau -\tau,\hat\tau_X^\mathrm{T},\hat\delta_W^\mathrm{T})^\mathrm{T}  \big\} = \left(\begin{matrix}V_{\tau\tau}&&V_{\tau X}&&V_{\tau W}\\V_{X\tau}&&V_{XX}&&V_{XW}\\V_{W\tau}&&V_{WX}&&V_{WW}\end{matrix}\right) =  \sumk\Pik^2 \pik^{-1} V_{[k]}.
$$

To infer $\tau$ based on $\tauunadj$, we need to further determine its asymptotic distribution. Let $\mathcal{M}_{L}$ be the set of finite-population quantities satisfying the maximum squared distance requirement and bounded stratum-specific second-moment conditions: {$\mathcal{M}_{L} = \big\{ (a_1,\ldots,a_N): \max_{k=1,...,\K } \max_{i\in [k]}n^{-1}\|a_i -\bar a_{[k]}\|_\infty^{2}\rightarrow 0,\max_{k=1,...,\K }N_{[k]}^{-1}\sumik \| a_i -\bar a_{[k]}\|^{2}_\infty \le L  \big \}$, }
where $\bar a_{[k]} = (1/\Nk)\sumik a_i$. We assume that the dimensions of covariates $X$ and $W$ are fixed, and we need Condition~\ref{cond srse} below to derive the joint asymptotic normality of $\sqrt n(\hat\tau -\tau,\hat\tau_X^\mathrm{T},\hat\delta_W^\mathrm{T})$.

\begin{condition}\label{cond srse}
For $t=0,1$, as $n\rightarrow \infty$, assume that
 \begin{enumerate}[label=(\roman*),leftmargin=*, align=left]
     \item there exist constants $c_1,c_2 \in (0,1)$ and $c_3 > 0$, such that for $k=1,\ldots,\K$, $c_1 \leq \ekt \leq 1- c_1 $, $\fk \leq c_2 $, and $f/\fk \leq c_3$;
     \item $f$ has a limit in $[0,1]$; 
     \item
$\sumk\Pik^2 \pik^{-1} e_{[k]t}^{-1}S_{[k]t}^2$, $\sumk \Pik^2 \pik^{-1} e_{[k]t}^{-1} S_{[k]X,t}$, $ \sumk \Pik^2 \pik^{-1} (1-\fk )S_{[k]W,t},$   $\sumk $ $ \Pik^2 \pik^{-1} \fk S_{[k]\tau}^2$, {$\sumk \Pik^2 \pik^{-1} (\ekt \ekc )^{-1} S_{[k]X}^2,$ and $\sumk \Pik^2  \pik^{-1} (1-\fk )S_{[k]W}^2$} have finite limits and the limits of $V_{XX}$ and $V_{WW}$ are invertible;
\item there exists a constant $L>0$ independent of $N$, such that $\{Y_i(t)\}_{i=1}^{N}$, $\{X_i\}_{i=1}^{N}$, $\{W_i\}_{i=1}^{N} \in \mathcal{M}_{L}$. 
 \end{enumerate}
\end{condition}

Condition~\ref{cond srse}(i) requires that the proportion of treated units within each stratum, $\ekt$, is bounded away from zero and one.
For simplicity, Condition~\ref{cond srse}(i) assumes that $\fk$ is bounded away from one. Our method and theory can be extended to the case $\fk=1$ for some $k$, though the notation and theoretical statements become more involved. Condition~\ref{cond srse}(i) also requires that $\fk$ is not negligible relative to $f$.
Otherwise, some stratum information dominates relative to other strata, and it is impossible to infer the overall treatment effect. Condition~\ref{cond srse}(i)-(ii) together allow both the total sampling proportion $f$ and the stratum-specific sampling proportion $\fk$ to tend to zero, which are natural requirements, and similar conditions can be found in the literature \citep{Li2016,yang2021rejective}.
Condition~\ref{cond srse}(iii) ensures that the covariance $V$ has a finite limit. For notation simplicity,  we still use the same notation to denote the limit of  $V$ when no confusion would arise. {When $f_{[k]}=f$ and $e_{[k]t}=e_t$ for all $k=1,\ldots,K$ and $t\in\{0,1\}$, Condition~\ref{cond srse}(iii) reduces to requiring that weighted sums  $e_t^{-1}\sumk\Pik S_{[k]t}^2, e_t^{-1}\sumk\Pik S_{[k]X,t}^2,\sumk\Pik S_{[k]W,t}^2,\sumk\Pik S_{[k]\tau}^2$ have finite limits and that $\sumk\Pik S_{[k]X}^2,\sumk\Pik S_{[k]W}^2$ converge to finite, positive-definite limits.}
Condition~\ref{cond srse}(iv) is a bounded moment type condition on the potential outcomes and covariates. In particular, the condition $\max_{k=1,...,\K }\max_{i\in [k]}n^{-1}\{ Y_i(1) -\bar Y_{[k]}(1) \}^{2}\rightarrow 0$, requires that the stratum-specific deviations of the potential outcomes from their stratum-specific mean should not be very large. It is a typical condition
used in the finite-population asymptotic theory \citep{Li2016,Li2018,yang2021rejective}. 
{Moreover, a sufficient condition for the array $a_i$'s ($a_i=Y_i(1),Y_i(0),X_i,$ or $W_i$) to belong to $\mathcal M_L$ is that there exist constants $L'>0$ and $\delta>0$ such that $\ \max_{k=1,...,K}N_{[k]}^{-1}\sumik || a_i -\bar a_{[k]}||^{2+\delta}_\infty \le L'$ with $f^2n^\delta\rightarrow\infty$. Lighter tails of the $a_i$'s correspond to larger values of $\delta$, which in turn permit a faster decay rate for $f$.} We then have the following theorem.





\begin{theorem}\label{theorem clt-srse}
Under Condition \ref{cond srse} and the stratified randomized survey experiment,  we have
$\sqrt n(\hat\tau -\tau,\hat\tau_X^\mathrm{T},\hat\delta_W^\mathrm{T})^\mathrm{T} \cd \mathcal N(0, V).$
\end{theorem}


Theorem~\ref{theorem clt-srse} provides a normal approximation for the 
distribution of $\tauunadj$ under the stratified randomized survey experiment, which can be used to make design-based inferences for the average treatment effect if one can estimate the asymptotic variance consistently or conservatively. Moreover, it provides the theoretical basis for deriving the asymptotic distribution of $\tauunadj$ under the  SRSRR  experiment.

Theorem~\ref{theorem clt-srse} requires very weak conditions on the number of strata, stratum sizes, and the proportion of treated units in each stratum. The asymptotic regime is very general and covers a wide range of cases, such as a diverging number of strata with fixed stratum sizes, diverging stratum sizes with a fixed number of strata, diverging numbers of strata and stratum sizes, or some combination thereof, along with varied proportions of treated units across strata.
Theorem~\ref{theorem clt-srse} also allows us to compare the asymptotic efficiency of $\hat \tau$ under the completely randomized and stratified randomized survey experiments. In the completely randomized survey experiment, the commonly-used average treatment effect estimator is the standard difference-in-means estimator $\hat\tau_{\textnormal{C}} = (1/n_1)\sum_{i=1}^N Z_i T_i Y_i - (1/n_0)\sum_{i=1}^N Z_i (1 - T_i) Y_i.$
To make a fair comparison and guarantee that $\hat \tau = \hat\tau_{\textnormal{C}}$, we consider a scenario with equal sampling and treated proportions across all strata, i.e., $\ekt=e_1$ and $\fk=f$. 
This uniformity ensures that each unit has the same probability to be sampled and the same probability to be assigned to the treatment as a non-stratification strategy.   Consequently, the improvement in estimation efficiency solely arises from the division of units into distinct strata.
Let $V_{\tau\tau,{\textnormal{C}}}$ denote the variance of $\hat\tau_{\textnormal{C}}$ under the completely randomized survey experiment. We can deduce Corollary~\ref{coro:efficiency gain S} below.

\begin{coro}\label{coro:efficiency gain S}
Under Condition \ref{cond srse}, if $\ekt=e_1$ and $\fk=f$, for $k=1,\ldots,\K$, we have
$\Vcr-\Vsr=  \frac{N}{N-1} \sumk \Pik d_{[k]}- \frac{1}{N-1}\sumk (1 - \Pik )V_{[k]\tau\tau},$
where $d_{[k]}=  [(e_0/e_1)^{1/2}\{\bar Y_{[k]}(1)-\bar Y(1)\}+ (e_1/e_0)^{1/2}\{\bar Y_{[k]}(0)-\bar Y(0)\}]^2+(1-f)(\tauk-\tau)^2$.
\end{coro}




\begin{remark}
When $f=1$, Corollary~\ref{coro:efficiency gain S} represents the comparison of the asymptotic variance of $\hat \tau$ under the one-stage completely randomized experiment and stratified randomized experiment. We focus on the case that $f<1$. The first term in the difference of the asymptotic variance, $N(N-1)^{-1}\sumk\Pik d_{[k]} \geq 0$, measures a weighted between-strata variation, which equals to zero if and only if $\bar Y_{[k]}(t)=\bar Y(t)$ for all $k=1,\ldots,\K$ and $t=0,1$. It corresponds to the case  that the stratification variable is not predictive at all. Otherwise, this term and its limit are usually larger than zero. The second term, $(N-1)^{-1}\sumk (1 - \Pik )V_{[k]\tau\tau}$, measures a weighted within-strata variation. Condition~\ref{cond srse} implies that 
$
(N-1)^{-1}\sumk (1 - \Pik )V_{[k]\tau\tau} \leq (4L/c_1) (\K-1)/(N-1),
$
which tends to zero if $\K/N \rightarrow 0$. Thus, stratification improves the asymptotic efficiency as long as the number of strata is not too large.  If $\K/N \nrightarrow 0$, the second term 
may not be negligible. However, the between-strata variation is usually larger than the within-strata variation. Hence, stratification can lead to a discernible boost in efficiency for most cases. 

\end{remark}

In the stratified randomized survey experiment, it is crucial to determine the optimal stratum-specific sampling proportion $\fk$ and treated proportion $\ekt$, which can be obtained via minimizing the asymptotic variance of $\hat \tau$. Theorem~\ref{coro sampling} below solves this problem.

\begin{theorem}\label{coro sampling}
The asymptotic variance of $\hat \tau$ under the stratified randomized survey experiment  is minimized if (i) for given treated proportion $e_{[k]1}$ and total sampling proportion $f$,  we have
$f_{[k]} / f = \sqrt{e_{[k]1}^{-1}S_{[k]1}^2+e_{[k]0}^{-1}S_{[k]0}^2} / (\sum_{k' = 1}^{\K} $ $ \Pi_{[k']} \sqrt{e_{[k']1}^{-1}S_{[k']1}^2+e_{[k']0}^{-1}S_{[k']0}^2} ) ,$
or (ii) for a given sampling proportion $\fk$, we have
$e_{[k]1} = \sqrt{S_{[k]1}^2} / ( \sqrt{S_{[k]1}^2} + \sqrt{S_{[k]0}^2 } ).$
\end{theorem}
If the potential outcomes are homogeneous across strata such that $e_{[k]1}^{-1}S_{[k]1}^2+e_{[k]0}^{-1}S_{[k]0}^2$ is a constant across $k$, then the optimal $\fk = f$ and $\pik = \Pik$. That is, the proportionate stratified sampling is optimal. Moreover, if the stratum-specific variances of $Y_i(1)$ and $Y_i(0)$ are the same, then the balanced design (i.e., $\ekt=0.5$) is optimal. Generally, if we can estimate the stratum-specific variance of the potential outcomes $Y_i(1)$ and $Y_i(0)$ by domain knowledge or prior studies, then we can choose the stratum-specific sampling proportion and treated proportion by Theorem~\ref{coro sampling} to achieve the smallest asymptotic variance of $\hat \tau$. \citet{cytrynbaum2023optimal} discussed the
 optimal stratification methods under budget constraints in survey experiments, but they considered the super-population framework.



\subsection{Asymptotic properties under SRSRR}

The asymptotic distribution of $\tauunadj$ under the SRSRR experiment  depends on the squared multiple correlations between $\hat\tau$ and  $\hat\delta_W$ and that between  $\hat\tau$ and $\hat\tau_X$ defined by
$R_W^2=(V_{\tau_W}V_{WW}^{-1}V_{W\tau})/V_{\tau\tau}$ and $ R_X^2=(V_{\tau X}V_{XX}^{-1}V_{X\tau})/V_{\tau\tau}.$
Under Condition~\ref{cond srse}, $R_W^2$ and $R_X^2$ have limiting values and we will still use the same notation to denote their limiting values for notation simplicity when no confusion would arise.


Theorem~\ref{theorem clt-srse} gives an asymptotic theory under the stratified randomized survey experiment conducted without rejective sampling and rerandomization. 
When these two techniques apply in the experiment, the asymptotic distribution of $\sqrt n(\hat\tau-\tau)$ will no longer be normal generally. Instead, it becomes a convolution of a normal distribution and two truncated normal distributions; see Theorem~\ref{theorem clt-SRSRR} below for the details. Let  $\varepsilon$ be a standard normal random variable and $L_{J,a}\sim D_1\mid  D^\mathrm{T} D\le a$ be a truncated normal random variable with $ D=(D_1,D_2,...,D_J)^\mathrm{T}\sim \mathcal N(0, I_J)$. As shown by \cite{Li2018}, $\var(L_{J,a}) = \nu_{J,a}=\pr(\chi_{J+2}^2\le a)/\pr(\chi_J^2\le a) \in (0,1)$, where $\chi^2_J$ is a chi-squared distribution  with $J$ degrees of freedom. Moreover, $L_{J,a}$ is unimodal, symmetric, and more concentrated at zero than the normal random variable with the same variance. 
Theorem \ref{theorem clt-SRSRR} below describes the asymptotic distribution of $\tauunadj$ under the SRSRR experiment. 
\begin{theorem}\label{theorem clt-SRSRR}
Under Condition \ref{cond srse} and the SRSRR experiment, we have
$$\sqrt{n} \{\hat\tau-\tau\} \cd V_{\tau\tau}^{1/2}\Big \{\sqrt{1-R_W^2-R_X^2}\cdot \varepsilon +\sqrt{R_W^2}\cdot L_{J_1,a_S}+\sqrt{R_X^2}\cdot L_{J_2,a_T} \Big\},$$
where $\varepsilon,L_{J_1,a_S},$ and $L_{J_2,a_T}$ are independent.
\end{theorem}

\begin{remark}
Another approach is to reject sampling and treatment assignment until both $M_S\le a_S$ and $M_T\le a_T$ are satisfied. This method needs more computational cost but is asymptotically equivalent to our method; see Section \ref{sec:E} in the Supplementary Material. 
    \cite{wang2022rerandomization} considered an asymptotic regime with a diminishing threshold and diverging covariate dimension. Extending their analysis, when $\log(\pr(M_S\leq a_S)^{-1})/J_{1}\rightarrow \infty$, the truncated normal term $L_{J_1,a_S}$ is asymptotically negligible. A similar result applies to $L_{J_2,a_T}$. 
\end{remark}

Theorem~\ref{theorem clt-SRSRR} implies that the asymptotic distribution of $\tauunadj$ is no longer normal, but by \cite{Li2018}, this limiting distribution is still unimodal and symmetric around zero. Moreover, SRSRR improves the efficiency compared to the standard stratified randomized survey experiment. 
The results are summarized in Corollary~\ref{coro avar} below.



\begin{coro}\label{coro avar}
Under Condition \ref{cond srse}, the percentage reduction in the asymptotic variance of $\tauunadj$ under the SRSRR experiment compared to that under the stratified randomized survey experiment is
$ [ 100\times\{(1-\nu_{J_1,a_S})R_W^2+(1-\nu_{J_2,a_T})R_X^2\}] \%$. Moreover, the asymptotic distribution of $\tauunadj$ under the SRSRR experiment is more concentrated at $\tau$ than that under the stratified randomized survey experiment.
\end{coro}

Similar to Theorem~\ref{coro sampling}, we can determine the optimal stratum-specific sampling proportion $\fk$ and treated proportion  $\ekt$ to minimize the asymptotic variance of $\hat \tau$ under the SRSRR experiment. However, since the squared multiple correlations $R_W^2$ and $R_X^2$ depend on $\fk$ and $\ekt$, the optimizer does not have a closed-form solution. In the Supplementary Material, we provide an algorithm to approximate the optimal $\fk$ and $\ekt$.

\subsection{Conservative variance and distribution estimators}
To infer the average treatment effect under the SRSRR experiment, we need to estimate the asymptotic distribution of $\tauunadj$.  Intuitively, we can use the sample variance and covariance to estimate the corresponding population quantities. Specifically, let $\skt^2=(\nkt-1)^{-1}\sumik Z_i T_i(Y_i-\barYkt)^2$ and $\skc^2=(\nkc-1)^{-1}\sumik Z_i (1-T_i)(Y_i-\barYkc)^2$ be the stratum-specific variances of $Y_i(1)$ and $Y_i(0)$ under the treatment and control, respectively. Denote $\skXt=(\nkt-1)^{-1}\sumik Z_i T_i(X_i-\barXkt)(Y_i-\barYkt)$ and $\skXc=(\nkc-1)^{-1}\sumik Z_i (1-T_i)(X_i-\barXkc)(Y_i-\barYkc)$ as  the stratum-specific covariances of $Y_i(1)$ and $Y_i(0)$ with $X_i$ under treatment and control, respectively. Similarly,  we define $\skWt=(\nkt-1)^{-1}\sumik Z_i T_i (W_i-\barWkt)(Y_i-\barYkt)$ and $\skWc=(\nkc-1)^{-1}\sumik Z_i (1-T_i)(W_i-\barWkc)(Y_i-\barYkc)$.

Note that the term 
$S_{[k]\tau}^2$, i.e., the stratum-specific finite-population variance of $\tau_i$, is not estimable because we cannot observe $Y_i(1)$ and $Y_i(0)$ simultaneously for any unit. Thus, we need to estimate $V_{\tau\tau}$ in Theorem~\ref{theorem clt-SRSRR} by a conservative estimator:
$
\hat V_{\tau\tau} = \sumk \Pik^2 \pik^{-1} (\ekt ^{-1}\skt^2+\ekc ^{-1}\skc^2).  
$
The covariate $W_i$ is often available for all units by preliminary surveys. Thus, we do not need to estimate $V_{WW}$. In contrast, the covariate $X_i$ is often available only for the sampled units.  Accordingly, we can estimate $V_{XX}$ by $\hat V_{XX}= \sumk \Pik^2 \pik^{-1}  (\ekt \ekc )^{-1}S_{[k]X \mid \mathcal S}^2$. The other elements of $V$ can be estimated by
$
\hat V_{X \tau } = \hat V_{\tau X}^\T  = \sumk \Pik^2\pik^{-1}(\ekt ^{-1}s_{[k]X,1}+\ekc ^{-1}s_{[k]X,0})$ and
$
\hat V_{W \tau } = \hat V_{\tau W}^\T  = \sumk \Pik^2 \pik^{-1}(1-\fk)(s_{[k]W,1}-s_{[k]W,0}).$
The squared multiple correlations can also be estimated by a plug-in method
$\hat R_W^2= \hat V_{\tau W}V_{WW}^{-1}\hat V_{W\tau} / \hat V_{\tau\tau}$ and $\hat R_X^2=\hat V_{\tau X} \hat V_{XX}^{-1}\hat V_{X\tau} / \hat V_{\tau\tau}.$ 
Then, the asymptotic distribution in Theorem~\ref{theorem clt-SRSRR} can be conservatively estimated by 
$ \hat V_{\tau\tau}^{1/2} \{\sqrt{1- \hat R_W^2- \hat R_X^2}\cdot \varepsilon +\sqrt{ \hat R_W^2}\cdot L_{J_1,a_S}+\sqrt{ \hat R_X^2}\cdot L_{J_2,a_T} \}.$
Let $\nu_{\xi}(\hat V_{\tau\tau}, \hat R_W^2, \hat R_X^2)$ be the $\xi$-th quantile of the above distribution. For a given $\alpha \in (0,1)$, we have the following theorem.



\begin{theorem}\label{theorem conserv}
    {Under Condition \ref{cond srse} and the SRSRR experiment, $\hat V_{\tau\tau} - V_{\tau\tau} = \sumk  \Pik^2\pik^{-1} $ $\fk S_{[k]\tau}^2 + \op$, and 
    $
    \big[\hat\tau- n^{-1/2} \nu_{1-\alpha/2}(\hat V_{\tau\tau}, \hat R_W^2,\hat R_X^2), \ \hat\tau + n^{-1/2}\nu_{1-\alpha/2}(\hat V_{\tau\tau},$ $ \hat R_W^2,\hat R_X^2) \big]
    $
    is an asymptotic conservative $1-\alpha$ confidence interval for $\tau$.}
\end{theorem}
\begin{remark}
For finely stratified randomized experiments, $\nkt=1$ or $\nkc=1$. In this case, we can use the variance estimation method proposed by  \cite{Pashley2017} and \cite{wang2022rerandomization}. 
\end{remark}

{
\begin{remark}
As an alternative variance estimator, one may use the causal bootstrap of \citet{imbens2021causal} and its stratified extension in \citet{yu2025sharp}. Conceptually, the procedure first constructs a finite-population ``science table'' by imputing the missing potential outcomes via rank-preserving imputation, and then repeatedly replays the sampling and treatment-assignment mechanisms to generate bootstrap replicates $\hat\tau^{*1},\ldots,\hat\tau^{*B}$. The variance $V_{\tau\tau}$ can be estimated by the empirical variance of $\{\sqrt{n}\hat\tau^{*b}\}_{b=1}^B$, which may have better finite-sample performance than the Neyman-type conservative variance estimator $\hat V_{\tau \tau}$.
\end{remark}}

Theorem~\ref{theorem conserv} ensures that the covariance estimator is asymptotically conservative; it is consistent if and only if the unit level treatment effect is constant within each stratum, i.e., $S^2_{[k]\tau} = 0$ for $k=1,\ldots,\K$. 


\section{Regression adjustment}\label{sec:reg}


Regression or covariate adjustment is widely used in the analysis stage of randomized experiments to improve the estimation efficiency \citep[see, e.g.,][]{lin2013,bloniarz2015lasso,Liu2019,Zhu2021block,liu2022straclt,Zhao2021}. Rerandomization and regression adjustment can be combined to further improve the efficiency \citep{Li2020,Wang2021,lu2022design}. In this section, we propose a covariate adjustment method under the SRSRR experiment.

In the analysis stage, more covariates may be collected for both the population and the sampled units. Assume that we have collected covariates $E_i \in \mathbb R^{J_3}$ for all $N$ units and $C_i \in \mathbb R^{J_4}$ for the sampled units in the analysis stage.  As shown by \cite{Li2018} and \cite{yang2021rejective}, if the designer and analyzer do not communicate well such that some of the covariates used in the design stage are not used in the analysis stage, the covariate-adjusted estimator may degrade the efficiency. Therefore, to avoid the efficiency loss, it is reasonable to assume that more covariates are used in the analysis stage, i.e., $E_i \supset W_i $ and $C_i \supset X_i$.



Substituting $W_i$ (or $X_i$) with $E_i$ (or $C_i$), we can further obtain quantities related to $E_i$ (or $C_i$), such as $\hat\delta_E$, $V_{EE}$, and $R_E^2$.
As $\hat \delta_E$ and $\tauunadj_C$ characterize the covariate imbalance in the sampling and treatment assignment stages, we consider the linearly adjusted estimator $\tauadj = \hat\tau-\beta^\mathrm{T}\tauunadj_C-\gamma^\mathrm{T}\hat \delta_E$, where $\beta$ and $\gamma$ are vectors adjusting for the remaining covariate imbalance. The adjusted vectors corresponding  to the most efficient (i.e., smallest asymptotic variance) linearly adjusted estimator are given by
$
(\beta_{\opt}, \gamma_{\opt}) = \argmin_{\beta, \gamma} E( \tauunadj - \beta^\mathrm{T}\tauunadj_C-\gamma^\mathrm{T}\hat \delta_E - \tau  )^2.
$
Recall that $V_{CC}$, $V_{C\tau}$, $V_{EE}$, and $V_{E\tau}$ are defined analogously to $V_{XX}$, $V_{X\tau}$, $V_{WW}$, and $V_{W\tau}$.
By Proposition~\ref{prop Cov}, we have $\beta_{\opt} = V_{CC}^{-1}V_{C\tau}$ and $\gamma_{\opt} = V_{EE}^{-1}V_{E\tau}$.
In practice, $\beta_{\opt}$ and $\gamma_{\opt}$ can be consistently estimated by 
\begin{eqnarray}
\hat \beta &=& \hat V_{CC}^{-1} \sumk\Pik^2 \pik^{-1}(\ekt ^{-1}s_{[k]C,1}+\ekc ^{-1}s_{[k]C,0}), \nonumber \\
\hat{\gamma} &=& V_{EE}^{-1}  \sumk\Pik^2 \pik^{-1}(1-\fk )(s_{[k]E,1}-s_{[k]E,0}), \nonumber
\end{eqnarray}
where $\hat V_{CC} = \sumk\Pik^2 \pik^{-1}(\ekt \ekc )^{-1}S^2_{[k]C \mid \mathcal S} $.
Then, we obtain a linearly adjusted estimator $\tauadj = \hat\tau- \hat \beta^\mathrm{T}\tauunadj_C- \hat \gamma^\mathrm{T}\hat \delta_E$. 
Next, we demonstrate that $\tauadj$ offers further efficiency improvement compared to $\tauunadj$ in both the stratified randomized survey experiment and the SRSRR experiment.  For this purpose, we require a regularity Condition~\ref{cond srse-adj} on $E_i$ and $C_i$ below, which is analogous to Condition \ref{cond srse} on $W_i$ and $X_i$.

\begin{condition}\label{cond srse-adj}
As $n\rightarrow \infty$, (i)   the following finite limits exist for $t = 0, 1$:
$\sumk \Pik^2 \pik^{-1}$    $e_{[k]t}^{-1} S_{[k]C,t},$ $\sumk \Pik^2 \pik^{-1}(1-\fk ) S_{[k]E,t},$   $\sumk \Pik^2\pik^{-1} (\ekt \ekc )^{-1} S_{[k]C}^2,$ and $\sumk \Pik^2$ $\pik^{-1} (1-\fk )S_{[k]E}^2$  and the limits of $V_{CC}$ and $V_{EE}$ are invertible;
(ii) there exists a constant $L>0$ independent of $N$, such that $\{C_i\}_{i=1}^{N}$, $\{E_i\}_{i=1}^{N} \in \mathcal{M}_{L}$. 
\end{condition}

\begin{theorem}\label{theorem S-optimal}
Under Condition \ref{cond srse}--\ref{cond srse-adj} and the SRSRR experiment (or the stratified randomized survey experiment), 
 we have 
$ \sqrt{n} ( \tauadj -\tau ) \cd \mathcal N(0, (1-R_E^2-R_C^2) V_{\tau\tau}  ).$
Moreover, $\tauadj$ has the smallest asymptotic variance among the class of linearly adjusted estimators $\{ \hat\tau-\beta^\mathrm{T}\tauunadj_C-\gamma^\mathrm{T}\hat \delta_E, \ \beta \in \mathbb R^{J_4}, \ \gamma \in \mathbb R^{J_3} \}.$
\end{theorem}

Theorem~\ref{theorem S-optimal} implies that the asymptotic distribution of $ \tauadj$ under the stratified randomized survey experiment ($a_S=a_T=\infty$) and SRSRR experiment are both normal and it is optimal in the class of linearly adjusted estimators if we use more covariate information at the analysis stage.  As $\tauunadj$ is also a linearly adjusted estimator with $\beta = 0$ and $\gamma = 0$, $\tauadj$ is asymptotically more efficient (no larger asymptotic variance) than the unadjusted estimator $\tauunadj$ under the stratified randomized survey experiment. Under the SRSRR experiment, the normal component in the asymptotic distribution of $\tauunadj$ is $\mathcal N(0, (1- R_W^2 - R_X^2 )V_{\tau\tau})$. Since $E_i \supset W_i $ and $C_i \supset X_i$,  we have $R_W^2 \leq R_E^2$ and $R_X^2 \leq R_C^2$. Accordingly, $\tauadj$ is asymptotically more efficient than $\tauunadj$ under the SRSRR experiment.

To estimate the asymptotic variance of $\tauadj$, recall that $V_{\tau \tau}$ can be conservatively estimated by $\hat V_{\tau\tau}=\sumk\Pik^2 \pik^{-1}(\ekt ^{-1}\skt^2+\ekc ^{-1}\skc^2)$. Thus, we only need to estimate $R_E^2 $ and $R_C^2 $. By plugging-in, their estimators are
$\hat R_E^2= \hat V_{\tau E} V_{EE}^{-1}\hat V_{E\tau}/ \hat V_{\tau\tau} $ and $\hat R_C^2=\hat V_{\tau C} \hat V_{CC}^{-1}\hat V_{C\tau}/ \hat V_{\tau\tau}$.
Denote $q_{\xi}$ as the $\xi$-th quantile of a standard normal distribution. Theorem~\ref{theorem ci-adj} below justifies the Wald type inference of $\tau$ based on $\tauadj$ under the stratified randomized survey experiment and SRSRR experiment, respectively.

\begin{theorem}\label{theorem ci-adj}
Under Condition \ref{cond srse}--\ref{cond srse-adj} and the SRSRR experiment (or the stratified randomized survey experiment), the confidence interval
$
[ \tauadj - n^{-1/2}q_{1-\alpha/2}\hat V_{\tau\tau}^{1/2}\sqrt{1-\hat R_E^2-\hat R_C^2},$ $\ \tauadj + n^{-1/2}q_{1-\alpha/2}\hat V_{\tau\tau}^{1/2}\sqrt{1-\hat R_E^2-\hat R_C^2} ]
$
has an asymptotic coverage rate greater than or equal to $1 - \alpha$, and it is asymptotically shorter than, or at least as short as, the confidence interval based on $\tauunadj$   in Theorem~\ref{theorem conserv}.
\end{theorem}
Theorem~\ref{theorem ci-adj} implies that the confidence interval based on $\tauadj$ is asymptotically conservative and $\tauadj$ improves the inference efficiency compared to $\hat \tau$. This conclusion holds for either the stratified randomized survey experiment ($a_S=a_T=\infty$) or  SRSRR experiment.


\section{Numerical studies}\label{sec:num}

\subsection{Synthetic data}

In this section, we conduct a simulation study to evaluate the finite-sample performance of $\tauunadj$ and $\tauadj$ 
under the stratified randomized 
survey experiment (SRSE) and SRSRR experiment.
We consider three scenarios for the number of strata $\K$ and stratum size  $N_{[k]}$: \case{case:1} 
{small strata only with $\K=200$ and each stratum of size $N_{[k]}=200$ for $1\leq k\leq \K$}; \case{case:2} small strata together with two large strata, i.e., $N_{[k]}=200$ for $1\leq k\leq \K-2$  and    $N_{[k]}=2000$ for $\K-1\leq k\leq \K$ with  $\K=22$, and \case{case:3} two large strata with $\K=2$ and $N_{[k]}=4000$ for $1\leq k\leq \K$. 

We independently generate the covariate $C_i=(C_{i1},C_{i2},C_{i3},C_{i4})^\mathrm{T}\sim\mathcal N(0, \Sigma)$ for  $1\leq i\leq N$, where $\Sigma= (\sigma_{jk})$ with $\sigma_{jk}=0.5^{|j-k|}$, $1\leq j,k\leq 4$.  In the design stage, let $W_i = C_{i1}$ and $X_i = (C_{i1},C_{i2})^\T$.  While in the analysis stage, we set $E_i = (C_{i1},C_{i2},C_{i3})^\T$. The potential outcome is generated from the following random effect model:
$$Y_i(t)= \sum_{j=1}^{4} C_{ij}\beta_{t1,j}  +\exp \Big( \sum_{j=1}^{4}  C_{ij} \beta_{t2,j} \Big)+D_{[k]}+\varepsilon_i(t),\quad t\in \{0,1\}, \quad i \in [k].$$
For $ 1\leq j\leq 4$,  we generate $\beta_{11,j}\sim t_3,\ \beta_{12,j}\sim 0.1 t_3,\ \beta_{01,j}\sim \beta_{11,j}+t_3$ and  $\beta_{02,j}\sim\beta_{12,j}+0.1t_3$, where $t_3$ represents the $t$ distribution with three degrees of freedom. We further generate $D_{[k]}\sim t_3$ for $1\leq k\leq \K$  and $\varepsilon_i(t) \sim \mathcal N(0,\sigma^2)$ for $t=0,1$. We choose $\sigma^2$ such that the signal-to-noise ratio is fixed at 10. In Cases \ref{case:1} and \ref{case:2},  we consider homogeneous strata with $(\beta_{t1,j},\beta_{t2,j})$ being the same across all strata.  In Case \ref{case:3}, we consider two large heterogeneous strata  and 
   the coefficients $(\beta_{t1,j},\beta_{t2,j})$'s  are generated independently for each  stratum.

 The covariates and potential outcomes are generated once and then kept fixed. We repeat the sampling and treatment assignment procedure $10^3$ times to approximate the distributions of $\tauunadj$ and $\tauadj$ and compute the bias, standard deviation (SD), root mean squared error (RMSE), empirical coverage probability (CP), and mean confidence interval length (Length) of $95\%$ confidence intervals.

We consider the  proportionate stratified sampling with  $\Pik =\pik$  for $k=1,\ldots,\K$ and set $f=n/N = 1/10$. We choose the threshold $a_S$ in the stratified rejective sampling to make the asymptotic accept rate  $p_S=\pr(\MS \leq a_S)=0.01$. In the treatment assignment stage, within each stratum, half of the sampled units are randomly assigned to the treatment group, and the rest are assigned to the control group. We choose the threshold $a_T$ such that the asymptotic accept rate in the stratified rerandomization is equal to $p_T=\pr(\MT \leq a_T) = 0.01$.  Next, we set $a_S = \infty$, which means that the stratified rejective sampling is not used, and denote the associated stratified randomized survey experiment with
stratified rerandomization only by SRSE-R for short. Similarly, we set $a_T = \infty$, which means that the stratified rerandomization is not used in the experiment, and denote the associated stratified randomized survey experiment with
stratified rejective sampling only by SRSE-S for short.

\begin{table}
    \centering
\caption{\label{tab:simu_case2}Simulation results for Case 2}
    \renewcommand\arraystretch{1}
    \resizebox{\linewidth}{!}{
    \begin{threeparttable}
    \begin{tabular}{cccccccccc}
    \hline
     Design & Str & Rej-Sam & ReR & Estimator & Bias($\times 10^2$) & SD & RMSE & Length & CP(\%) \\
    \hline
     SRSE & \cmark &  &   & $\tauunadj$ & 0.029  & 0.137 & 0.137 & 0.552 & 96.0 \\
    SRSE-S & \cmark & \cmark & & $\tauunadj$ & 0.314  & 0.134 & 0.134 & 0.536 & 95.9 \\
    SRSE-R & \cmark &  & \cmark & $\tauunadj$ & -0.479 & 0.119 & 0.119 & 0.489 & 96.6 \\
    SRSRR  & \cmark & \cmark & \cmark & $\tauunadj$ & -0.074 & 0.111 & 0.111 & 0.466 & 96.6 \\
    SRSE  & \cmark &  & & $\tauadj$   & 0.039  & 0.072 & 0.072 & 0.300 & 95.5 \\
    SRSE-S & \cmark & \cmark & & $\tauadj$   & 0.366  & 0.069 & 0.069 & 0.301 & 96.6 \\
    SRSE-R & \cmark &  & \cmark &  $\tauadj$   & -0.161 & 0.071 & 0.071 & 0.301 & 96.1 \\
    SRSRR   & \cmark & \cmark & \cmark & $\tauadj$   & -0.261 & 0.071 & 0.071 & 0.300 & 95.9 \\
    CRSE & &  &     & $\tauunadj$ & -0.720 & 0.160 & 0.160 & 0.630 & 95.5 \\
    CRSE-S & & \cmark &  & $\tauunadj$ & -1.062 & 0.149 & 0.150 & 0.617 & 95.7 \\
    CRSE-R & &  & \cmark & $\tauunadj$ & -0.275 & 0.148 & 0.148 & 0.578 & 95.1 \\
    RRSE  & & \cmark & \cmark & $\tauunadj$ & -0.137 & 0.136 & 0.136 & 0.560 & 96.5 \\
    CRSE  & &  & & $\tauadj$   & -0.111 & 0.107 & 0.107 & 0.432 & 96.5 \\
    CRSE-S & & \cmark &  & $\tauadj$   & -0.601 & 0.107 & 0.107 & 0.432 & 95.4 \\
    CRSE-R & &  & \cmark & $\tauadj$   & 0.665  & 0.109 & 0.109 & 0.432 & 94.7 \\
    RRSE & & \cmark & \cmark  & $\tauadj$   & 0.030  & 0.108 & 0.108 & 0.432 & 96.1 \\
    \hline
    \end{tabular}
    \begin{tablenotes}
    \item {\small Str, stratification; Rej-Sam, rejective sampling; ReR, rerandomization; SRSE, stratified randomized survey experiment; SRSE-S, stratified randomized survey experiment with stratified rejective sampling; SRSE-R, stratified randomized survey experiment with rerandomization; SRSRR, stratified rejective sampling and rerandomized survey experiment; CRSE, completely randomized survey experiment; CRSE-S, completely randomized survey experiment with rejective sampling; CRSE-R, completely randomized survey experiment with rerandomization; RRSE, rejective sampling and rerandomized survey experiment.}
    \end{tablenotes}
    \end{threeparttable}}
\end{table}

{Table~\ref{tab:simu_case2} shows the results for Case \ref{case:2}, which features many small strata along with two large strata. The results for Cases~\ref{case:1} and \ref{case:3} are presented in Tables~\ref{tab:simu_case1} and \ref{tab:simu_case3} of the Supplementary Material, and lead to similar conclusions.

First, the bias is negligible across all estimators and designs, so the RMSE is nearly identical to the SD. 


Second, all design- and analysis-stage strategies yield substantial efficiency gains. For the unadjusted estimator $\tauunadj$, stratification reduces the SD from 0.160 under CRSE to 0.137 under SRSE, a reduction of 14.4\%. Building on SRSE, rejective sampling further reduces the SD from 0.137 to 0.134, a modest gain of 2.2\%, whereas rerandomization reduces it from 0.137 to 0.119, corresponding to a 13.1\% reduction. When stratification, rejective sampling, and rerandomization are combined, the SD further decreases to 0.111 under SRSRR, representing a 19.0\% reduction relative to SRSE and a 30.6\% reduction relative to the baseline design CRSE. In addition, regression adjustment leads to further efficiency gains across all designs, reducing the SD by about 21\% to 49\% relative to the corresponding unadjusted estimator.

Third, for the regression-adjusted estimator $\tauadj$, the differences among SRSE, SRSE-S, SRSE-R, and SRSRR are relatively small, suggesting that regression adjustment already captures a substantial part of the covariate imbalance. In particular, $\tauadj$ performs similarly under SRSE and SRSRR. However, the supplementary results show that when the sample size becomes smaller, the RMSE of $\tauadj$ under SRSRR can still be appreciably smaller than that under SRSE, by about $17\%$; see Table \ref{tab:finite compare} in the Supplementary Material. This indicates that design-stage covariate balance may still provide additional gains beyond analysis-stage regression in finite samples. 

Finally, the lengths of the confidence intervals exhibit similar patterns when stratification, rejective sampling, rerandomization, or regression adjustment is applied. Moreover, all methods produce confidence intervals with coverage probabilities close to or slightly above the nominal level 95\%.

}

\subsection{Cooperative congressional election study data}

In this section, we consider the Cooperative Congressional Election Study (CCES) to illustrate the proposed design (SRSRR) and evaluate the performance of the unadjusted and covariate-adjusted average treatment effect estimators.
{CCES was a national stratified sample survey administered by YouGov \citep[CCES]{christenson2017realdata1}. Many teams participated and pooled their studies into the whole survey experiment. We follow \cite{yang2021rejective} choosing the team survey of Boston University from CCES 2014 \citep{DVN/UN0USO_2017, DVN/II2DB6_2022}. Participants were asked about their views on federal spending on scientific research. Researchers assigned about 1000 participants to two groups to ask their opinions on federal scientific research spending. Those participants assigned to the treatment group were provided with additional information about the current federal research budget: ``each year, just over 1\% of the federal budget is spent on scientific research", while the control group was blind to this information. The outcome of interest is whether federal spending on scientific research should be increased, the same, or decreased, coded as 1, 2, and 3.} 
We treat all the CCES 2014 participants as our target population with $N=49452$. We stratify the population into four strata based on race:
White, Black, Hispanic, and Other. In addition to the stratification variable race,
eight pretreatment covariates are used at different stages of the SRSRR experiment.
Specifically, three covariates are observed at the sampling stage, i.e.,
\(W_i\in\mathbb R^3\). In the treatment assignment stage, we use five covariates,
i.e., \(X_i\in\mathbb R^5\) with \(W_i\subset X_i\), which means that two additional
covariates are observed at the treatment assignment stage. Moreover, all eight
covariates are observed at the analysis stage with \(C_i\in\mathbb R^8\).
The covariate names are provided in the Supplementary Material.

 Next,  we sample $n=1000$ units with the same proportion across strata to conduct the experiment. Half of the sampled units are assigned to the treatment group and the rest to the control group.
Because not all the potential outcomes are observed in the study, we cannot obtain the actual gains of the proposed design and the covariate-adjusted treatment effect estimator. So, we estimate the gains by generating a synthetic dataset with unobserved potential outcomes imputed by a
linear model. Specifically, we fit a linear regression model of the observed outcome on the treatment indicator, all covariates, and their interactions, and then impute the unobserved potential outcomes using the fitted linear model. To make the experiment more realistic, we fine-tune the dataset and get three different datasets. Specifically, 
we first compute the individual treatment effect $\tau_i$, 
and then shrink the potential outcomes of the control group toward the population means by a factor $\lambda$ with $\lambda=0.1,0.2$ and $0.3$, respectively. The potential outcomes under the control are defined as $Y_i^*(0)=\bar Y(0) +\lambda \{ Y_i(0)-\bar Y (0) \}$. Accordingly, we redefine the potential outcomes under the treatment by $Y_i^*(1)=Y_i^*(0)+\tau_i$. The squared multiple correlations of the three generated datasets are presented in Table \ref{table dataset} in the Supplementary Material. The potential outcomes and covariates are fixed in the follow-up analysis. We set $a_S$ and $a_T$ to achieve an asymptotic acceptance probability of $p_S=p_T=0.001$ for the SRSRR experiment, $a_T = \infty$ for SRSE-S, and $a_S = \infty$ for SRSE-R. We repeat the sampling and treatment assignment stages $1000$ times to evaluate the finite-sample performance of $\tauunadj$ and $\tauadj$ for each dataset.

The results are summarized in {Tables~\ref{tab:cces_scenario1}--\ref{tab:cces_scenario3}} in the Supplementary Material. {Across all three datasets, both} SRSE-S and SRSE-R improve the efficiency of SRSE and their combination, SRSRR, performs the best among the four {stratified} designs. The unadjusted estimator $\hat\tau$ has the largest RMSE and longest CI length in the third dataset. This is mainly because the third dataset has the smallest $R_W^2$ and a relatively small $R_X^2$. Moreover, the covariate-adjusted estimator $\tauadj$ further improves the efficiency compared to the unadjusted estimator $\hat\tau$.  It reduces the RMSE by $17.0\%-54.5\%$ and CI length by $18.6\%-53.6\%$, respectively.


{Tables~\ref{tab:cces_scenario1}--\ref{tab:cces_scenario3} also include the results for rejective sampling and rerandomization without stratification, allowing comparison with the design of \citet{yang2021rejective}.} {Across all three datasets, the corresponding non-stratified designs generally yield larger SDs, RMSEs, and longer confidence intervals for both $\tauunadj$ and $\tauadj$ than the stratified designs, providing further evidence that stratification improves} the efficiency of survey experiments.


\section{Discussion}

The one-stage randomized experiment may lack external validity because of the distinction between the population and the sample. In practice, the survey experiment considered the data collection procedure and can address this issue. However, covariate imbalance often exists in the completely randomized survey experiment. In this work, we proposed a two-stage stratified rejective sampling and rerandomized (i.e., SRSRR) experimental design and further covariate adjustment method to balance covariates and improve efficiency. Additionally, we develop a design-based asymptotic theory for the stratified difference-in-means estimator under the proposed design, allowing heterogeneous strata regimes. Both theoretical and numerical results demonstrate that stratification and rerandomization can improve asymptotic efficiency. Finally, we discuss the optimal stratum-specific sampling proportion and treated proportion to achieve the smallest asymptotic variance. Our theory is purely design-based, and the validity of the resulting estimates and inference procedures does not require the correct specification of the underlying outcome model. 


{Practically, we recommend a staged workflow: first, stratify units using domain knowledge whenever possible; next, implement rejective sampling and rerandomization when feasible; and finally, apply regression adjustment to further reduce residual covariate imbalance and improve estimation precision.
}

\section*{Acknowledgements}
Yingying Ma was supported by the National Natural Science Foundation of China (No.12171020). Pengfei Tian was supported by the Beijing Natural Science Foundation (QY23081).

\bibliographystyle{main}
\bibliography{main}

\appendix
\phantomsection 
\setcounter{equation}{0}
\renewcommand{\theequation}{S\arabic{equation}}
\setcounter{table}{0}
\renewcommand{\thetable}{S\arabic{table}}
\setcounter{figure}{0}
\renewcommand{\thefigure}{S\arabic{figure}}
\setcounter{theorem}{0}
\renewcommand{\thetheorem}{S\arabic{theorem}}
\setcounter{lemma}{0}
\renewcommand{\thelemma}{S\arabic{lemma}}
\setcounter{condition}{0}
\renewcommand{\thecondition}{S\arabic{condition}}
\setcounter{remark}{0}
\renewcommand{\theremark}{S\arabic{remark}}
\setcounter{prop}{0}
\makeatletter
\renewcommand{\theprop}{S\arabic{prop}}
\makeatother
\setcounter{coro}{0}
\makeatletter
\renewcommand{\thecoro}{S\arabic{coro}}
\makeatother


\makeatletter
\renewcommand*{\theHtheorem}{S\arabic{theorem}}
\renewcommand*{\theHlemma}{S\arabic{lemma}}
\renewcommand*{\theHprop}{S\arabic{prop}}
\renewcommand*{\theHcoro}{S\arabic{coro}}
\makeatother

\noindent
\newpage

\renewcommand{\thepage}{S\arabic{page}} 
\setcounter{page}{1}

\begin{center}
\Huge
Supplementary Material
\end{center}
Section \ref{A}  contains additional theoretical results on the mean and covariance of $\sqrt{n}(\hat\tau -\tau,\hat\tau_X^\mathrm{T},\hat\delta_W^\mathrm{T})^\mathrm{T}$ under the stratified randomized survey experiment and the optimal sampling and treated proportions under the SRSRR experiment with and without covariate adjustment. \\
Section~\ref{C} provides the names of covariates used in the Cooperative Congressional Election Study. \\
Section~\ref{sec:tables} provides tables for numerical results in the main text.\\
Section \ref{sec::Simulation} provides additional simulation results to show the finite sample advantage for design-stage rerandomization over analysis-stage regression adjustment. \\
Section~\ref{sec:E} shows the equivalence between two stratified rejective sampling and rerandomization procedures.\\
Section \ref{B}  contains theoretical proofs for Theorem~\ref{theorem clt-srse}--\ref{theorem ci-adj}, Corollary~\ref{coro:efficiency gain S}--\ref{coro avar} in the main text, and Proposition~\ref{prop Cov}, Theorem~\ref{coro sampling-srrse}--\ref{coro sampling-adj} in Section \ref{A}, Theorem \ref{thm:TV_bound}, Corollary \ref{coro:d_TV to 0} in Section \ref{sec:E} in the Supplementary Material.

\section{Additional theoretical results}\label{A}
{First, we provide notations in Tables \ref{tab:notation} and \ref{tab:notation 2}.}
\begin{table}[]
    \centering
    \caption{Notation Table}
    \label{tab:notation}
    \resizebox{\linewidth}{!}{%
    \begin{tabular}{cc}
    \toprule
        $N,n,\K$ & Population size; sample size; strata number \\
        $\Nk,\nk$ & Size of stratum $k$ in the population; sample size in stratum $k$\\
        $\Pik,\pik$ & Population proportion of stratum $k$; sample proportion of stratum $k$\\
        $f,\fk$ & Sampling ratio in the population; sampling ratio in the $k$-th stratum\\
       $\nkt,\nkc$& Numbers of treated and control units in stratum $k$\\
        $n_1,n_0$ & Total numbers of treated and control units: $n_1=\sumk \nkt$, $n_0=\sumk \nkc$ \\
        $\ekt,\ekc$& Stratum-$k$ treatment and control assignment probabilities\\
        $\mathcal S$& Index set of sampled units\\
       $Z_i,T_i$ & Sampling indicator of unit $i$; treatment indicator of unit $i$\\
       $Y_i(1),Y_i(0),Y_i$ &Potential outcomes under treatment and control; observed outcome\\
    $\tau_i,\tauk,\tau$&Unit-level effect; stratum-$k$ average effect; overall average effect \\
       $\bar Y_{[k]1},\bar Y_{[k]0}$ & Stratum-$k$ sample mean outcomes (treatment/control)\\
$\bar Y_{1},\bar Y_{0}$ & Overall sample mean outcomes (treatment/control)\\
       $\taukhat,\tauunadj$& Stratum-specific estimator; unadjusted overall estimator\\
       $W_i\in \mathbb R^{J_1},X_i\in \mathbb R^{J_2}$& Covariates at the sampling stage; covariates at the treatment stage\\
       $\barW,\barWS$&Population mean and weighted sample mean of covariate $W$\\
       $\barW_{[k]},\barWSk$& Stratum $k$ mean and stratum $k$ sample mean of covariate $W$\\
        $\hat\delta_W$ & Sampling-stage covariate mean difference based on $W_i$ \\
       $\barXt,\barXc$& Weighted sample mean of the treated units and control units\\
$\bar X_{[k]1},\bar X_{[k]0}$ & Stratum-$k$ sample means of $X_i$ in treatment and control groups \\
$\hat\tau_X$ & Covariate mean difference: $\hat\tau_X=\bar X_t-\bar X_c$ \\
$S_{[k]W}^2$ & Stratum-$k$ population covariance matrix of $W_i$ \\
$\bar X_{[k]\mathcal S},S_{[k]X\mid\mathcal S}^2$ & Stratum-$k$ mean vector and covariance matrix of $X_i$ conditional on $\mathcal S$ \\
$M_S,a_S$ & Sampling-stage Mahalanobis distance and its acceptance threshold \\
$M_T,a_T$ & Assignment-stage Mahalanobis distance and its acceptance threshold \\
$q_\xi$ & $\xi$-th quantile of the standard normal distribution \\
    \bottomrule
    \end{tabular}}
\end{table}

\begin{table}[]
    \centering
    \caption{Notation Table (Continued)}
    \label{tab:notation 2}
    \resizebox{\linewidth}{!}{%
    \begin{tabular}{cc}
    \toprule
$e_1,e_0$ & Overall assignment proportions: $e_1=n_1/n$, $e_0=n_0/n$ \\
$\bar Y_{[k]}(t),\bar X_{[k]}$ & Stratum-$k$ finite-population means of $Y_i(t)$ and $X_i$ \\
$S^2_{[k]t},S^2_{[k]\tau}$ & Stratum-$k$ finite-population variances of $Y_i(t)$ and $\tau_i$ \\
$S^2_{[k]X},S_{[k]X,t}$ & Stratum-$k$ covariance of $X_i$ and covariance between $X_i$ and $Y_i(t)$ \\
$S_{[k]W,t},S_{[k]W,\tau}$ & Stratum-$k$ covariances between $W_i$ and $Y_i(t)$, and between $W_i$ and $\tau_i$ \\
$V_{[k]},V$ & Stratum-$k$ covariance block matrix and $V=\sumk \Pik^2\pik^{-1}V_{[k]}$ \\
$R_W^2,R_X^2$ & Squared multiple correlations: $R_W^2=(V_{\tau W}V_{WW}^{-1}V_{W\tau})/V_{\tau\tau}$, $R_X^2=(V_{\tau X}V_{XX}^{-1}V_{X\tau})/V_{\tau\tau}$ \\
$L_{J,a}$ & truncated normal random variable $L_{J,a}\sim D_1\,|\,D^\T D\le a$ with $D\sim\mathcal N(0,I_J)$ \\
$\nu_{J,a}$ & variance of truncated normal random variable $\nu_{J,a}=\var(L_{J,a})=\pr(\chi_{J+2}^2\le a)/\pr(\chi_J^2\le a)$ \\
$E_i\in\mathbb R^{J_3},C_i\in\mathbb R^{J_4}$ & Analysis-stage covariates: $E_i$ observed for all $N$ units and $C_i$ observed for sampled units \\
$\hat\delta_E$ & Sampling-stage covariate mean difference based on $E_i$ (analogous to $\hat\delta_W$) \\
$\hat\tau_C$ & Covariate mean difference between treatment and control based on $C_i$ (analogous to $\hat\tau_X$) \\
$\tauadj$ & Linearly adjusted estimator: $\tauadj=\hat\tau-\beta^\T\hat\tau_C-\gamma^\T\hat\delta_E$ \\
$\beta_{\opt},\gamma_{\opt}$ & Optimal coefficients: $\beta_{\opt}=V_{CC}^{-1}V_{C\tau}$, $\gamma_{\opt}=V_{EE}^{-1}V_{E\tau}$ \\
$\hat\beta,\hat\gamma$ & Plug-in estimators of $\beta_{\opt}$ and $\gamma_{\opt}$ \\
$V_{CC},V_{EE}$ & Asymptotic covariance matrices of $\hat\tau_C$ and $\hat\delta_E$ \\
$V_{C\tau},V_{E\tau}$ & Cross-covariances between $\hat\tau_C$ and $\hat\tau$, and between $\hat\delta_E$ and $\hat\tau$ \\
$R_C^2,R_E^2$ & Squared multiple correlations: $R_C^2=(V_{\tau C}V_{CC}^{-1}V_{C\tau})/V_{\tau\tau}$, $R_E^2=(V_{\tau E}V_{EE}^{-1}V_{E\tau})/V_{\tau\tau}$ \\
$\hat R_C^2,\hat R_E^2$ & Plug-in estimators of $R_C^2$ and $R_E^2$ \\
$\hat V_{CC}$ & Plug-in estimator of $V_{CC}$ (e.g., $\hat V_{CC}=\sumk \Pik^2\pik^{-1}(\ekt\ekc)^{-1}S^2_{[k]C\mid\mathcal S}$) \\
$\textup{SRSRR}$ & Stratified rejective sampling and rerandomized survey experiment (Our method)\\
$\textup{SRSE}$ & Stratified randomized survey experiment \\
$\textup{SRSE-S}$ & SRSE with rejective sampling only ($a_T=\infty$) \\
$\textup{SRSE-R}$ & SRSE with rerandomization only ($a_S=\infty$) \\
$\textup{CRSE}$ & Completely randomized survey experiment (no stratification) \\
$\textup{CRSE-S}$ & CRSE with rejective sampling only \\
$\textup{CRSE-R}$ & CRSE with rerandomization only \\
$\textup{RRSE}$ & Rejective-sampling and rerandomized survey experiment without stratification \\
    \bottomrule
    \end{tabular}}
\end{table}

\subsection{Mean and covariance of $\sqrt{n}(\hat\tau -\tau,\hat\tau_X^\mathrm{T},\hat\delta_W^\mathrm{T})^\mathrm{T}$}
Recall the definition of Section \ref{sec:asym}, we establish the first two moments of $\sqrt{n}(\hat\tau -\tau,\hat\tau_X^\mathrm{T},\hat\delta_W^\mathrm{T})^\mathrm{T}$ as follows.
\begin{prop}\label{prop Cov}
Under the stratified randomized survey experiment, $\sqrt n (\hat\tau -\tau,\hat\tau_X^\mathrm{T},\hat\delta_W^\mathrm{T})^\mathrm{T}$ has mean zero and covariance
$$
V = \left(\begin{matrix}V_{\tau\tau}&&V_{\tau X}&&V_{\tau W}\\V_{X\tau}&&V_{XX}&&V_{XW}\\V_{W\tau}&&V_{WX}&&V_{WW}\end{matrix}\right) =  \sumk\Pik^2 \pik^{-1} V_{[k]}.
$$
\end{prop}

\subsection{Without covariate adjustment} 
In the SRSRR experiment, it is important to determine the optimal stratum-specific sampling proportion $\fk$ and treated proportion $\ekt$, which minimizes the asymptotic variance of $\hat \tau$.  The below theorem solves this problem.

\begin{theorem}\label{coro sampling-srrse}
The asymptotic variance of $\hat \tau$ under the SRSRR experiment is minimized if (i) for given treated proportion $e_{[k]1}$ and total sampling proportion $f$, we have
$$f_{[k]} / f =  A_k\Big/ \Big(\sum_{k' = 1}^{\K} \Pi_{[k']} A_{k'} \Big) ,$$
where 
\begin{align*}
A_k=&\Big\{\ekt^{-1}S_{[k]1}^2 +\ekc^{-1}S_{[k]0}^2-2(1-\nu_{J_1,a_S}) S_{[k]\tau W}V_{WW}^{-1}V_{W\tau} \\
&+(1-\nu_{J_1,a_S}) V_{\tau W}V_{WW}^{-1}S_{[k]W}^2V_{WW}^{-1}V_{W\tau}\\
&-2(1-\nu_{J_2,a_T})(\ekt^{-1}S_{[k]1,X}+\ekc^{-1}S_{[k]0,X})V_{XX}^{-1}V_{X\tau}\\
&+(1-\nu_{J_2,a_T})(\ekt \ekc)^{-1}V_{\tau X}V_{XX}^{-1}S_{[k]X}^2V_{XX}^{-1}V_{X\tau}\Big\}^{1/2},
\end{align*}
and (ii) for given sampling proportion $\fk$, we have
$$e_{[k]1} = ( a_{[k]1} ) / ( a_{[k]1} + a_{[k]0} ),$$
where $a_{[k]t}=(|S_{[k]t}^2+(1-\nu_{J_2,a_T})V_{\tau X}V_{XX}^{-1}S_{[k]X}^2V_{XX}^{-1}V_{X\tau}-2(1-\nu_{J_2,a_T})S_{[k]t,X}V_{XX}^{-1}V_{X\tau}|)^{1/2}$
for $t=0,1$.
\end{theorem}

Similar to Theorem \ref{coro sampling}, we can derive the optimal sampling proportion and treated proportion based on domain knowledge or prior studies. However, the difference lies  in that the computation of $V_{\tau\tau}$, $V_{\tau W}$, and $V_{WW}$ relies on the sampling proportion $f_{[k]}$, while $V_{\tau X}$ and $V_{XX}$ depend on the treated proportion $\ekt$.
We can use the following iterative process to obtain the optimal sampling proportion and treated proportion. 

We use superscripts $[\cdot]^{(m)}$ to represent the value of $[\cdot]$ in the $m$th iteration. Based on  Theorem \ref{coro sampling-srrse}, we can update $\fk$ and $\ekt$ as $\fk^{(m+1)}=fA_k^{(m)}/(\sumk \Pik A_k^{(m)})$ and $e_{[k]t}^{(m+1)}=a_{[k]t}^{(m)}/(a_{[k]1}^{(m)}+a_{[k]0}^{(m)})$ for $t=0,1$. Accordingly, we can iteratively update $A_k^{(m+1)}$, $a_{[k]1}^{(m+1)}$, and $a_{[k]0}^{(m+1)}$ until convergence is achieved. We recommend  the selection of initial values of $\fk^{(0)}$ and $\ekt^{(0)}$ via the results of Theorem \ref{coro sampling}. 

\subsection{With covariate adjustment}

In the SRSRR experiment, after covariate adjustment,  we can determine the optimal stratum-specific sampling proportion $\fk$ and treated proportion $\ekt$ via minimizing the asymptotic variance of $\tauadj$.  The following theorem solves this problem.
\begin{theorem}\label{coro sampling-adj}
The asymptotic variance of $\tauadj$ under the SRSRR experiment (or the stratified randomized survey experiment) is minimized if (i) for given treated proportion $e_{[k]1}$ and total sampling proportion $f$,  we have
$$f_{[k]} / f =  B_k\Big/ \Big(\sum_{k = 1}^{\K} \Pi_{[k]} B_k \Big),$$
where 
{\begin{align*}
B_k=&\Big\{\ekt^{-1}S_{[k]1}^2+\ekc^{-1}S_{[k]0}^2-2 S_{[k]\tau E}V_{EE}^{-1}V_{E\tau }+ V_{\tau E}V_{EE}^{-1}S_{[k]E}^2V_{EE}^{-1}V_{E\tau}\\
&-2(\ekt^{-1}S_{[k]1,C}+\ekc^{-1}S_{[k]0,C})V_{CC}^{-1}V_{C\tau}+(\ekt\ekc)^{-1}V_{\tau C}V_{CC}^{-1}S_{[k]C}^2V_{CC}^{-1}V_{C\tau}\Big\}^{1/2},
\end{align*}}
and (ii) for given sampling proportion $\fk$, we have
$$e_{[k]1} = b_{[k]1}/ (  b_{[k]1} + b_{[k]0} ),$$
where $b_{[k]t}= (|S_{[k]t}^2+V_{\tau C}V_{CC}^{-1}S_{[k]C}^2V_{CC}^{-1}V_{C\tau}-2S_{[k]t,C}V_{CC}^{-1}V_{C\tau}|)^{1/2}$ for $t=0,1$.
\end{theorem}

Similar to Theorem \ref{coro sampling-srrse}, the computation of $V_{\tau\tau}$, $V_{\tau E}$, and $V_{EE}$ relies on the sampling proportion $f_{[k]}$, while $V_{\tau C}$ and $V_{CC}$ depends on the treated proportion $\ekt$.
By Theorem \ref{coro sampling-adj}, we can update the values of $\fk$ and $\ekt$ as follows: $\fk^{(m+1)}=fB_k^{(m)}/(\sumk  \Pik B_k^{(m)})$ and $e_{[k]t}^{(m+1)}=b_{[k]t}^{(m)}/(b_{[k]1}^{(m)}+b_{[k]0}^{(m)})$, for $t=0,1$. Subsequently, we update $B_k^{(m+1)}$ and $b_{[k]1}^{(m+1)}, b_{[k]0}^{(m+1)}$ using the updated values of $\fk^{(m+1)}$, $\ekt^{(m+1)}$, and $\ekc^{(m+1)}$. This process continues until convergence is achieved. We recommend using Theorem~\ref{coro sampling} to obtain initial values for $\fk^{(0)}$ and $\ekt^{(0)}$. 

\section{Covariates used in the Cooperative Congressional Election Study}
\label{C}
The race variable is used to form four strata and is not included in the
covariate vectors below. The covariates used in different stages are as follows:
\begin{enumerate}
\item Sampling stage, $W_i$: age, gender, whether the highest level of education is college or higher;
    \item Treatment assignment stage, additional covariates in $X_i$: whether family annual income is less than $60000$, whether the individual believes the economy has gotten worse last year;
    \item Analysis stage, additional covariates in $C_i$: whether the individual is liberal or moderate, whether party identification is democrat, whether the individual follows the news and public affairs most of the time.
\end{enumerate}
We present the squared multiple correlations of the three generated datasets in Table \ref{table dataset} below.

\begin{table}[H]
    \centering
    \caption{\label{table dataset} Squared multiple correlations of three generated datasets}
    \begin{tabular}{cccc}
     \hline
       shrinkage& $R_W^2=R_E^2$ & $R_X^2$ & $R_C^2$  \\
         \hline
        0.1 & 0.384 & 0.254 & 0.368\\
        0.2& 0.399 &0.174 &0.342\\
        0.3 &0.331 &0.195&0.455\\
         \hline
    \end{tabular}
\end{table}

\vspace{-5cm}

\section{Tables for numerical results in the main text}\label{sec:tables}

\begin{table}[H]
    \centering
\caption{\label{tab:simu_case1}Simulation results for Case 1}
    \renewcommand\arraystretch{1}
    \resizebox{0.85\linewidth}{!}{
    \begin{threeparttable}
    \begin{tabular}{cccccccccc}
    \hline
     Design & Str & Rej-Sam & ReR & Estimator & Bias($\times 10^2$) & SD & RMSE & Length & CP(\%) \\
    \hline
    SRSE   & \cmark &  &  & $\tauunadj$ & -0.317 & 0.137 & 0.137 & 0.547 & 95.4 \\
    SRSE-S & \cmark & \cmark &  & $\tauunadj$ & -0.544 & 0.132 & 0.132 & 0.522 & 95.2 \\
    SRSE-R & \cmark &  & \cmark & $\tauunadj$ & -0.224 & 0.107 & 0.107 & 0.443 & 96.1 \\
    SRSRR  & \cmark & \cmark & \cmark & $\tauunadj$ & 0.306  & 0.096 & 0.096 & 0.411 & 97.1 \\
    SRSE   & \cmark &  &  & $\tauadj$   & -0.010 & 0.034 & 0.034 & 0.169 & 98.5 \\
    SRSE-S & \cmark & \cmark &  & $\tauadj$   & -0.050 & 0.033 & 0.033 & 0.168 & 98.2 \\
    SRSE-R & \cmark &  & \cmark & $\tauadj$   & -0.062 & 0.034 & 0.034 & 0.169 & 98.0 \\
    SRSRR  & \cmark & \cmark & \cmark & $\tauadj$   & 0.092  & 0.034 & 0.034 & 0.170 & 98.1 \\
    CRSE   &  &  &  & $\tauunadj$ & -0.483 & 0.150 & 0.150 & 0.592 & 95.2 \\
    CRSE-S &  & \cmark &  & $\tauunadj$ & -0.437 & 0.138 & 0.138 & 0.570 & 95.6 \\
    CRSE-R &  &  & \cmark & $\tauunadj$ & -0.359 & 0.129 & 0.129 & 0.497 & 94.0 \\
    RRSE   &  & \cmark & \cmark & $\tauunadj$ & 0.179  & 0.123 & 0.123 & 0.470 & 94.1 \\
    CRSE   &  &  &  & $\tauadj$   & -0.134 & 0.074 & 0.074 & 0.283 & 94.8 \\
    CRSE-S &  & \cmark &  & $\tauadj$   & 0.041  & 0.069 & 0.069 & 0.284 & 95.5 \\
    CRSE-R &  &  & \cmark & $\tauadj$   & -0.273 & 0.073 & 0.073 & 0.284 & 94.8 \\
    RRSE   &  & \cmark & \cmark & $\tauadj$   & 0.093  & 0.072 & 0.072 & 0.283 & 95.3 \\
    \hline
    \end{tabular}
    \begin{tablenotes}
    \item {\small Str, stratification; Rej-Sam, rejective sampling; ReR, rerandomization; SRSE, stratified randomized survey experiment; SRSE-S, stratified randomized survey experiment with stratified rejective sampling; SRSE-R, stratified randomized survey experiment with rerandomization; SRSRR, stratified rejective sampling and rerandomized survey experiment; CRSE, completely randomized survey experiment; CRSE-S, completely randomized survey experiment with rejective sampling; CRSE-R, completely randomized survey experiment with rerandomization; RRSE, rejective sampling and rerandomized survey experiment.}
    \end{tablenotes}
    \end{threeparttable}}
\end{table}

\begin{table}
    \centering
\caption{\label{tab:simu_case3}Simulation results for Case 3}
    \renewcommand\arraystretch{1}
    \resizebox{0.85\linewidth}{!}{
    \begin{threeparttable}
    \begin{tabular}{cccccccccc}
    \hline
     Design & Str & Rej-Sam & ReR & Estimator & Bias($\times 10^2$) & SD & RMSE & Length & CP(\%) \\
    \hline
    SRSE   & \cmark &  &  & $\tauunadj$ & 0.250  & 0.146 & 0.146 & 0.608 & 96.1 \\
    SRSE-S & \cmark & \cmark &  & $\tauunadj$ & 0.401  & 0.136 & 0.136 & 0.567 & 96.3 \\
    SRSE-R & \cmark &  & \cmark & $\tauunadj$ & -0.107 & 0.135 & 0.135 & 0.570 & 96.4 \\
    SRSRR  & \cmark & \cmark & \cmark & $\tauunadj$ & -0.153 & 0.128 & 0.128 & 0.529 & 94.9 \\
    SRSE   & \cmark &  &  & $\tauadj$   & -0.258 & 0.110 & 0.110 & 0.464 & 97.1 \\
    SRSE-S & \cmark & \cmark &  & $\tauadj$   & 0.484  & 0.111 & 0.111 & 0.464 & 96.9 \\
    SRSE-R & \cmark &  & \cmark & $\tauadj$   & 0.349  & 0.107 & 0.107 & 0.464 & 96.8 \\
    SRSRR  & \cmark & \cmark & \cmark & $\tauadj$   & 0.077  & 0.111 & 0.111 & 0.464 & 95.9 \\
    CRSE   &  &  &  & $\tauunadj$ & -0.302 & 0.183 & 0.183 & 0.727 & 94.9 \\
    CRSE-S &  & \cmark &  & $\tauunadj$ & -0.467 & 0.185 & 0.185 & 0.731 & 94.7 \\
    CRSE-R &  &  & \cmark & $\tauunadj$ & 0.268  & 0.143 & 0.143 & 0.568 & 95.2 \\
    RRSE   &  & \cmark & \cmark & $\tauunadj$ & -0.422 & 0.147 & 0.147 & 0.566 & 94.6 \\
    CRSE   &  &  &  & $\tauadj$   & -0.072 & 0.136 & 0.136 & 0.544 & 95.8 \\
    CRSE-S &  & \cmark &  & $\tauadj$   & -0.222 & 0.140 & 0.140 & 0.543 & 95.3 \\
    CRSE-R &  &  & \cmark & $\tauadj$   & 0.236  & 0.138 & 0.138 & 0.543 & 94.9 \\
    RRSE   &  & \cmark & \cmark & $\tauadj$   & -0.516 & 0.142 & 0.142 & 0.543 & 94.5 \\
    \hline
    \end{tabular}
    \begin{tablenotes}
    \item {\small Str, stratification; Rej-Sam, rejective sampling; ReR, rerandomization; SRSE, stratified randomized survey experiment; SRSE-S, stratified randomized survey experiment with stratified rejective sampling; SRSE-R, stratified randomized survey experiment with rerandomization; SRSRR, stratified rejective sampling and rerandomized survey experiment; CRSE, completely randomized survey experiment; CRSE-S, completely randomized survey experiment with rejective sampling; CRSE-R, completely randomized survey experiment with rerandomization; RRSE, rejective sampling and rerandomized survey experiment.}
    \end{tablenotes}
    \end{threeparttable}}
\end{table}

\begin{table}[H]
    \centering
\caption{\label{tab:cces_scenario1}Results for the CCES data under Scenario 1 (the first dataset)}
    \renewcommand\arraystretch{1}
    \resizebox{\linewidth}{!}{
    \begin{threeparttable}
    \begin{tabular}{cccccccccc}
    \hline
     Design & Str & Rej-Sam & ReR & Estimator & Bias($\times 10^3$) & SD($\times 10^3$) & RMSE($\times 10^3$) & Length($\times 10^3$) & CP(\%) \\
    \hline
    SRSE   & \cmark &  &  & $\tauunadj$ & -0.12 & 5.22 & 5.22 & 20.46 & 95.6 \\
    SRSE-S & \cmark & \cmark &  & $\tauunadj$ & 0.00  & 3.98 & 3.98 & 16.08 & 96.0 \\
    SRSE-R & \cmark &  & \cmark & $\tauunadj$ & -0.06 & 4.69 & 4.68 & 17.80 & 94.1 \\
    SRSRR  & \cmark & \cmark & \cmark & $\tauunadj$ & -0.00 & 3.12 & 3.12 & 12.55 & 95.0 \\
    SRSE   & \cmark &  &  & $\tauadj$   & 0.15  & 2.61 & 2.61 & 10.23 & 94.7 \\
    SRSE-S & \cmark & \cmark &  & $\tauadj$   & 0.06  & 2.64 & 2.64 & 10.24 & 94.3 \\
    SRSE-R & \cmark &  & \cmark & $\tauadj$   & 0.00  & 2.67 & 2.67 & 10.26 & 94.9 \\
    SRSRR  & \cmark & \cmark & \cmark & $\tauadj$   & 0.02  & 2.59 & 2.59 & 10.21 & 95.2 \\
    CRSE   &  &  &  & $\tauunadj$ & 0.16  & 5.80 & 5.79 & 22.77 & 94.1 \\
    CRSE-S &  & \cmark &  & $\tauunadj$ & 0.35  & 4.63 & 4.65 & 18.85 & 96.5 \\
    CRSE-R &  &  & \cmark & $\tauunadj$ & 0.13  & 5.21 & 5.21 & 20.38 & 94.7 \\
    RRSE   &  & \cmark & \cmark & $\tauunadj$ & 0.09  & 4.08 & 4.08 & 15.93 & 95.9 \\
    CRSE   &  &  &  & $\tauadj$   & 0.14  & 3.52 & 3.52 & 13.88 & 94.4 \\
    CRSE-S &  & \cmark &  & $\tauadj$   & 0.20  & 3.48 & 3.48 & 13.84 & 95.2 \\
    CRSE-R &  &  & \cmark & $\tauadj$   & 0.16  & 3.61 & 3.61 & 13.90 & 95.1 \\
    RRSE   &  & \cmark & \cmark & $\tauadj$   & 0.09  & 3.52 & 3.52 & 13.88 & 96.3 \\
    \hline
    \end{tabular}
    \begin{tablenotes}
    \item {\small Str, stratification; Rej-Sam, rejective sampling; ReR, rerandomization; SRSE, stratified randomized survey experiment; SRSE-S, stratified randomized survey experiment with stratified rejective sampling; SRSE-R, stratified randomized survey experiment with rerandomization; SRSRR, stratified rejective sampling and rerandomized survey experiment; CRSE, completely randomized survey experiment; CRSE-S, completely randomized survey experiment with rejective sampling; CRSE-R, completely randomized survey experiment with rerandomization; RRSE, rejective sampling and rerandomized survey experiment.}
    \end{tablenotes}
    \end{threeparttable}}
\end{table}

\begin{table}[H]
    \centering
\caption{\label{tab:cces_scenario2}Results for the CCES data under Scenario 2 (the second dataset)}
    \renewcommand\arraystretch{1}
    \resizebox{\linewidth}{!}{
    \begin{threeparttable}
    \begin{tabular}{cccccccccc}
    \hline
     Design & Str & Rej-Sam & ReR & Estimator & Bias($\times 10^3$) & SD($\times 10^3$) & RMSE($\times 10^3$) & Length($\times 10^3$) & CP(\%) \\
    \hline
    SRSE   & \cmark &  &  & $\tauunadj$ & -0.16 & 5.18 & 5.18 & 20.06 & 95.4 \\
    SRSE-S & \cmark & \cmark &  & $\tauunadj$ & 0.04  & 3.90 & 3.90 & 15.58 & 94.7 \\
    SRSE-R & \cmark &  & \cmark & $\tauunadj$ & -0.11 & 4.81 & 4.81 & 18.33 & 94.6 \\
    SRSRR  & \cmark & \cmark & \cmark & $\tauunadj$ & -0.10 & 3.36 & 3.36 & 13.24 & 95.3 \\
    SRSE   & \cmark &  &  & $\tauadj$   & 0.15  & 2.61 & 2.61 & 10.23 & 94.4 \\
    SRSE-S & \cmark & \cmark &  & $\tauadj$   & 0.06  & 2.64 & 2.64 & 10.24 & 94.6 \\
    SRSE-R & \cmark &  & \cmark & $\tauadj$   & -0.00 & 2.66 & 2.66 & 10.26 & 94.8 \\
    SRSRR  & \cmark & \cmark & \cmark & $\tauadj$   & 0.02  & 2.59 & 2.59 & 10.21 & 95.1 \\
    CRSE   &  &  &  & $\tauunadj$ & 0.05  & 5.65 & 5.64 & 22.30 & 94.9 \\
    CRSE-S &  & \cmark &  & $\tauunadj$ & 0.32  & 4.43 & 4.44 & 18.28 & 96.8 \\
    CRSE-R &  &  & \cmark & $\tauunadj$ & 0.12  & 5.28 & 5.27 & 20.65 & 94.2 \\
    RRSE   &  & \cmark & \cmark & $\tauunadj$ & 0.04  & 4.13 & 4.12 & 16.24 & 95.7 \\
    CRSE   &  &  &  & $\tauadj$   & 0.15  & 3.46 & 3.46 & 13.65 & 94.2 \\
    CRSE-S &  & \cmark &  & $\tauadj$   & 0.19  & 3.42 & 3.43 & 13.61 & 95.3 \\
    CRSE-R &  &  & \cmark & $\tauadj$   & 0.15  & 3.56 & 3.57 & 13.67 & 94.9 \\
    RRSE   &  & \cmark & \cmark & $\tauadj$   & 0.10  & 3.46 & 3.46 & 13.65 & 95.9 \\
    \hline
    \end{tabular}
    \begin{tablenotes}
    \item {\small Str, stratification; Rej-Sam, rejective sampling; ReR, rerandomization; SRSE, stratified randomized survey experiment; SRSE-S, stratified randomized survey experiment with stratified rejective sampling; SRSE-R, stratified randomized survey experiment with rerandomization; SRSRR, stratified rejective sampling and rerandomized survey experiment; CRSE, completely randomized survey experiment; CRSE-S, completely randomized survey experiment with rejective sampling; CRSE-R, completely randomized survey experiment with rerandomization; RRSE, rejective sampling and rerandomized survey experiment.}
    \end{tablenotes}
    \end{threeparttable}}
\end{table}

\begin{table}[H]
    \centering
\caption{\label{tab:cces_scenario3}Results for the CCES data under Scenario 3 (the third dataset)}
    \renewcommand\arraystretch{1}
    \resizebox{\linewidth}{!}{
    \begin{threeparttable}
    \begin{tabular}{cccccccccc}
    \hline
     Design & Str & Rej-Sam & ReR & Estimator & Bias($\times 10^3$) & SD($\times 10^3$) & RMSE($\times 10^3$) & Length($\times 10^3$) & CP(\%) \\
    \hline
    SRSE   & \cmark &  &  & $\tauunadj$ & -0.21 & 5.74 & 5.74 & 22.01 & 94.7 \\
    SRSE-S & \cmark & \cmark &  & $\tauunadj$ & 0.09  & 4.61 & 4.61 & 18.00 & 94.2 \\
    SRSE-R & \cmark &  & \cmark & $\tauunadj$ & 0.17  & 5.17 & 5.17 & 19.87 & 94.0 \\
    SRSRR  & \cmark & \cmark & \cmark & $\tauunadj$ & -0.19 & 3.97 & 3.97 & 15.29 & 94.7 \\
    SRSE   & \cmark &  &  & $\tauadj$   & 0.15  & 2.61 & 2.61 & 10.21 & 94.2 \\
    SRSE-S & \cmark & \cmark &  & $\tauadj$   & 0.06  & 2.64 & 2.64 & 10.22 & 94.3 \\
    SRSE-R & \cmark &  & \cmark & $\tauadj$   & -0.01 & 2.66 & 2.66 & 10.24 & 94.9 \\
    SRSRR  & \cmark & \cmark & \cmark & $\tauadj$   & 0.02  & 2.59 & 2.59 & 10.19 & 95.5 \\
    CRSE   &  &  &  & $\tauunadj$ & -0.05 & 6.06 & 6.06 & 23.97 & 95.3 \\
    CRSE-S &  & \cmark &  & $\tauunadj$ & 0.30  & 4.98 & 4.99 & 20.28 & 96.5 \\
    CRSE-R &  &  & \cmark & $\tauunadj$ & 0.11  & 5.58 & 5.58 & 21.88 & 94.9 \\
    RRSE   &  & \cmark & \cmark & $\tauunadj$ & -0.16 & 4.48 & 4.48 & 17.76 & 95.6 \\
    CRSE   &  &  &  & $\tauadj$   & 0.15  & 3.42 & 3.42 & 13.49 & 94.0 \\
    CRSE-S &  & \cmark &  & $\tauadj$   & 0.17  & 3.39 & 3.39 & 13.44 & 95.0 \\
    CRSE-R &  &  & \cmark & $\tauadj$   & 0.15  & 3.54 & 3.54 & 13.51 & 94.9 \\
    RRSE   &  & \cmark & \cmark & $\tauadj$   & 0.11  & 3.42 & 3.42 & 13.48 & 95.7 \\
    \hline
    \end{tabular}
    \begin{tablenotes}
    \item {\small Str, stratification; Rej-Sam, rejective sampling; ReR, rerandomization; SRSE, stratified randomized survey experiment; SRSE-S, stratified randomized survey experiment with stratified rejective sampling; SRSE-R, stratified randomized survey experiment with rerandomization; SRSRR, stratified rejective sampling and rerandomized survey experiment; CRSE, completely randomized survey experiment; CRSE-S, completely randomized survey experiment with rejective sampling; CRSE-R, completely randomized survey experiment with rerandomization; RRSE, rejective sampling and rerandomized survey experiment.}
    \end{tablenotes}
    \end{threeparttable}}
\end{table}

\section{Additional simulation results}
\label{sec::Simulation}

Although $\tauadj$ has the same asymptotic distribution under SRSE and SRSRR, its performance for a small sample size could be different. In this section, we conduct a simulation to compare the performance of $\tauadj$ under SRSE and SRSRR for a sample size of $n=80$. Specifically, the potential outcomes and covariates are generated in the same way as those in Case 1 in the main text, except that we set $\K = 4$ with $\Nk=40$, $f=0.5$, $\nk=20$, $\nkt=10$ for $k=1,\ldots,\K$.  Table \ref{tab:finite compare} shows the results. We can see that the SD and RMSE of $\tauadj$ under SRSRR are approximately $17\%$ smaller than those under SRSE. Therefore, rejective sampling and rerandomization have advantages when the sample size is small. Furthermore, we report the SD and RMSE of the oracle-adjusted estimator $\tauora=\hat\tau- \beta_{\opt}^\mathrm{T}\tauunadj_C- \gamma_{\opt}^\mathrm{T}\hat \delta_E$, whose performances are similar under two designs. A more balanced design usually decreases $\hat\tau_C$ and further decreases the gap between $\tauora$ and $\tauadj$.

\begin{table}[H]
\caption{\label{tab:finite compare}Simulation results for $\tauadj$ when $n=80$}
\centering
\resizebox{0.7\linewidth}{!}{\begin{threeparttable}

\begin{tabular}{ccccccc}
 \hline
{  Design} & {  Estimator} & {  Bias($\times 10^2$)}   & {  SD}     & {  RMSE}   & {  CI-length} & CP(\%) \\
\hline
{  SRSE}   & {  $\tauadj$} & {  -2.000} & {  0.350} & {  0.350} & {  2.091}   & 95.7 \\
   & {  $\tauora$} & {  -0.173}  & {  0.203} & {  0.203} & {  2.092}   & 97.5 \\
{  SRSRR}  & {  $\tauadj$} & {  -0.183}  & {  0.291} & {  0.291} & {  2.076}   & 96.5 \\
  & {  $\tauora$} & {  0.061}   & {  0.201} & {  0.201} & {  2.076}   & 97.4\\
 \hline
\end{tabular}
\begin{tablenotes}
\item SRSE, stratified randomized survey experiment; SRSRR, stratified rejective sampling and rerandomized survey experiment.  
\end{tablenotes}
\end{threeparttable}}
\end{table}

\section{Asymptotic equivalence between single-stage and two-stage rerandomized survey experiments}
\label{sec:E}
In this section, we extend the result of \citet[B2]{yang2021rejective} to a stratified version. Single-stage stratified rejective sampling and rerandomized experiment refers to stratified survey experiments conducted only when sampling and treatment assignment satisfying $M_S\le a_S, M_T\le a_T$ simultaneously, which is denoted by $\ssrsrr$. The two-stage version refers to the one discussed in the main text, which means that we first conduct the sampling if $M_S\le a_S$ is satisfied and then conduct the treatment assignment if $M_T\le a_T$ is satisfied. Two-stage stratified rejective sampling and rerandomized survey experiment is denoted as $\srsrr$. 

Denote the acceptance probabilities for the sampling stage and assignment stage as $p_S$ and $p_T$. The expected numbers of trials until acceptance for single-stage and two-stage SRSRR are $p_S^{-1}p_T^{-1}$ and $p_S^{-1}+p_T^{-1}$, respectively. Two-stage $\srsrr$ decreases the computational cost significantly compared to single-stage. Denote $\pr(\cdot \mid \srsrr)$ and $\pr(\cdot \mid \ssrsrr)$ as the conditional probability under $\srsrr$ and $\ssrsrr$, respectively. We will prove the asymptotic equivalence between $\ssrsrr$ and $\srsrr$. We first introduce the total variation distance between the two designs,
\begin{align*}\label{eq:total_variation}
    & \quad \text{d}_{\TV}(\ssrsrr, \srsrr) 
    \nonumber
    \\
    & = \sup_{\mathcal{A}\subset \{0,1\}^N \times \{0,1\}^n}
    \left|
    \pr\big\{ (Z,T_{\mathcal{S}} )\in \mathcal{A} \mid \ssrsrr \big\} - \pr\big\{ (Z,T_{\mathcal{S}} )\in \mathcal{A} \mid \srsrr \big\}
    \right|.
\end{align*}

Total variation distance measures the difference between the two designs and bounds many quantities of perturbation due to the two designs. Denote
$$\mathcal{M}=\{(Z,T_{\mathcal S}): \text{~satisfying~}  M_S\le a_S, M_T\le a_T \text{~under~} \text{SRSE} \}$$ 
as the set of all possible sampling and treatment assignment $(Z, T_{\mathcal{S}})$ under $\ssrsrr$.

\begin{theorem}\label{thm:TV_bound}
   The total variation distance between the probability distributions of $(Z,  T_{\mathcal{S}})$ under $\ssrsrr$ and $\srsrr$, is bounded by
\begin{align*}
    & \quad \text{d}_{\TV}(\ssrsrr, \srsrr) 
    \\
    & =
    \sup_{\mathcal{A}\subset \{0,1\}^N \times \{0,1\}^n}
    \left|
    \pr\big\{ (Z, T_{\mathcal{S}} )\in \mathcal{A} \mid \ssrsrr \big\}
    - 
    \pr\big\{ (Z, T_{\mathcal{S}} )\in \mathcal{A} \mid \srsrr \big\}
    \right|
    \nonumber
    \\
    & \le 
    1(\mathcal{M} = \emptyset) 
    + 
    1(\mathcal{M} \neq \emptyset) \cdot
    \frac{E \left| \pr(M_T \le a_T \mid  {Z} ) - 
    \pr(M_T\le a_T \mid M_S\le a_S) \right|}{ \pr(M_T\le a_T, M_S\le a_S)}.
    \nonumber
\end{align*}
\end{theorem}

Theorem \ref{thm:TV_bound} bounds the total variation distance between the two designs. We can prove that the total variation distance tends to zero as shown in Corollary \ref{coro:d_TV to 0} below.

\begin{coro}\label{coro:d_TV to 0}
    Under Condition \ref{cond srse}, $\text{d}_{\TV}(\ssrsrr, \srsrr) \rightarrow 0$ as $n\rightarrow \infty$.
\end{coro}

 The difference of conditional distributions of $\sqrt{n}(\tauunadj-\tau)$ or $\sqrt{n}(\tauadj-\tau)$ under  $\srsrr$ and $\ssrsrr$ can be bounded by the total variation distance between $\srsrr$ and $\ssrsrr$. By Corollary \ref{coro:d_TV to 0}, the difference tends to zero as $n\rightarrow \infty$. Therefore, the asymptotic distributions of $\sqrt{n}(\tauunadj-\tau)$ or $\sqrt{n}(\tauadj-\tau)$ are the same under the two designs. 

\section{Proof of main results}\label{B}

\subsection{Proof of Theorem~\ref{theorem clt-srse}}

\cite{Li2016} have proved a non-stratified finite-population central limit theory (CLT) to establish the design-based theory for the difference-in-means estimator in the one-stage completely randomized experiment. \cite{liu2022straclt} extended the CLT to the one-stage stratified randomized experiment, assuming that the sampling proportion is one and the treated proportion within each stratum is bounded away from zero and one. 
Nevertheless, this requirement is unrealistic in the two-stage survey experiment. It is sometimes reasonable to allow the sampling proportion in the first stage close to zero. Hence, we need to handle the zero-limit situation for the sampling proportion in our proof. For this purpose, we will establish a new joint CLT for $\hat\tau=\barYt-\barYc$, allowing $Y_i$ to be a vector; see Lemma~\ref{lemma joint clt} below.

To provide a general result for both scalar and vector potential outcomes, we introduce $R_i(1)\in \mathbb R^d$ and $ R_i(0)\in \mathbb R^d$ as $d$-dimensional potential outcomes. For example, we can directly set $R_i(1)=Y_i(1)$ and  $R_i(0) = Y_i(0)$; we can also set $R_i(1) = (Y_i(1), W_i^\T, X_i^\T)^\T$ and $R_i(0) = (Y_i(0),  W_i^\T, X_i^\T)^\T$. For $t=0,1$, define $\bar{R}_{[k]t}$,  $\bar{R}_{t}$, $S_{[k]R(t)}^2$, $\tau_R$, $\hat \tau_R$, and 
$S_{[k]\tau(R)}^2$ similarly to $\bar{Y}_{[k]t}$,  $\bar{Y}_{t}$, $S_{[k]t}^2$, $\tau$, $\hat \tau$, and 
$S_{[k]\tau}^2$, with $Y_i$ replaced by $R_i$. Similar to the proof in Proposition~\ref{prop Cov}, we can obtain 
$$
V_R = \cov (\sqrt{n}\hat \tau_R)=\cov\{\sqrt{n}(\barRt-\barRc)\} = \sumk\Pik^2 \pik^{-1} V_{[k]R},
$$
where
$
V_{[k]R} = \ekt^{-1}S_{[k]R(1)}^2+ \ekc^{-1}S_{[k]R(0)}^2-\fk S_{[k]\tau_R}^2.$

Next, for the general $d$-dimensional potential outcomes $R_i(1)$ and $R_i(0)$,   we need Condition~\ref{cond joint clt} below to derive the asymptotic normality of $\sqrt{n}\hat\tau_R$. 

\begin{condition}\label{cond joint clt}
As $n\rightarrow \infty$, $V_R$ has a finite limit and there exists a constant $L>0$ independent of $N$, such that $\{R_i(1)\}_{i=1}^{N}$, $\{R_i(0)\}_{i=1}^{N}\in \mathcal{M}_{L}$. 
\end{condition}
This is a condition for the general vector potential outcomes $R_i(1)$ and $R_i(0)$. Actually, when 
$R_i(1)=Y_i(1)$ and  $R_i(0) = Y_i(0)$, this condition is implied by Condition~\ref{cond srse}(iii)-(iv).
We still use $V_R$ to denote its limit when no confusion would arise. Lemma \ref{lemma joint clt} below establishes the asymptotic normality of $\sqrt{n}\hat\tau_R$.

\begin{lemma}\label{lemma joint clt}
    Under Condition~\ref{cond srse}(i)--(ii) and Condition~\ref{cond joint clt}, we have $$\sqrt {n} (\hat \tau_R-\tau_R) \cd \mathcal N(0,V_R).$$
\end{lemma}
\begin{proof}[Proof of Lemma~\ref{lemma joint clt}]

We prove Lemma~\ref{lemma joint clt} by using H{\'a}jek's coupling technique for obtaining the vector-form of the Wald--Wolfowitz--Hoeffding theorem for a linear (or bi-linear) rank statistic \citep{hajek1961some}. 

It suffices for Lemma~\ref{lemma joint clt} to show that for any fixed vector $ u \in \mathbb R^d$, 
$$  \sqrt {n} u^\mathrm{T}\{\barRt-\bar R(1)  \} - \sqrt {n} u^\mathrm{T}\{\barRc-\bar R(0)  \} \stackrel{d}{\rightarrow} \mathcal N(0,u^\mathrm{T} V_R u).$$


If $u^\mathrm{T} V_R u=0$, the conclusion holds trivially.  We then consider the case  $ u ^\mathrm{T} V_R u>0$. 
Let $u_1 = u$ and $u_0 = -u$. Denote  $ H_n= \sum_{q=0}^1\sqrt{n} u_q^\mathrm{T}\{{\bar {R}}_q -\bar{R}(q) \}$, which further equals  to 
\begin{align*}
    & \sum_{q=0}^1\sqrt {n}\sumk \Pik \frac{1}{n_{[k]q}}\sumik Z_i I(T_i=q)u_q^\mathrm{T}\{R_i(q) -\bar{R}_{[k]}(q) \}\\
    =& \sumk \sumik \sum_{q=0}^1 Z_i I(T_i=q) \cdot \frac{1}{\sqrt {n}}\cdot \frac{n\Nk}{Nn_{[k]q}}A_i(q),
\end{align*}
where $A_i(q)=u_q^\mathrm{T}\{R_i(q) -\bar{R}_{[k]}(q)\}$, for $i\in[k]$.
Without loss of generality, suppose that the units are ordered strata-by-strata. Then, $i\in [k]$ means $\sum_{k'=1}^{k-1}N_{[k']} <i\le \sum_{k'=1}^k N_{[k']}$. Let
\begin{align*}
    b_i(q)=\left\{
    \begin{array} {ll}
    \frac{1}{\sqrt {n}}\cdot \frac{n\Nk}{Nn_{[k]q}},&\sum_{k'=1}^{k-1}N_{[k']}+\sum _{q'=0}^{q-1}n_{[k]q'}<i\le \sum_{k'=1}^{k-1}N_{[k']}+\sum _{q'=0}^{q}n_{[k]q'},  \\
    0,& \text{otherwise}.
    \end{array} \right.
\end{align*}
Here, $\sum_{k'=1}^{0}N_{[k']}=0$ and $\sum _{q'=0}^{-1}n_{[k]q'}=0$.
Let $(G_i,i\in [k])$ be the random partition of $[k]$, which means $(G_i, i \in [k])$ takes any permutation of $[k]$ with probability $1/\Nk!$. 
Hence $H_n$ has the same distribution as $$H_n'=\sumk \sumik \big\{A_i(1)b_{G_i}(1)+A_i(0)b_{G_i}(0)\big\}.$$
Next,  we borrow ideas from \cite{hajek1961some} to derive the asymptotic normality of $H_n'$. Let $U_i \ \textnormal{i.i.d.} \sim  U(0,1)$, the uniform distribution on $(0,1)$. Within block $k$, denote $b(\lambda, q)$ as the quantile function of $\{b_i(q),\ i\in [k]\}$. 

\begin{definition}[\cite{hajek1961some}]
The quantile function of $N$ real numbers, $\{c_i,i=1,\ldots,N\}$, is defined as 
$c(\lambda)=c_{[i]}$ for  $(i-1)/N<\lambda \le i/N$ and $1\le i\le N,$
where $c_{[1]} \leq \cdots \leq c_{[N]}$ are the order statistics of  $\{c_i,i=1,\ldots,N\}$.
\end{definition}

Define
$$T_n=\sumk \sumik \sum_{q=0}^1[\{A_i(q)-\bar A_{[k]}(q)\} b(U_i,q)+\bar A_{[k]}(q)\cdot b_i(q)],$$
where $\bar A_{[k]}(q)$ is the stratum-specific population mean of $A_i(q)$.

\textbf{Step 1.}  In this step, we examine the scenario where $\K/n \rightarrow 0$. To streamline the proof of Lemma~\ref{lemma joint clt}, we introduce Lemma \ref{lemma:bound order} and  \ref{lemma:hajek} below, with their proofs following immediately after the proof of Lemma~\ref{lemma joint clt}.

    \begin{lemma}\label{lemma:bound order}
    Under Condition~\ref{cond srse}(i)--(ii) and Condition~\ref{cond joint clt}, there exists a constant $L'$ independent of $N$, such that
        $$\max_{q=1,0} \max_{i\in [k]}|b_i(q)|\le \frac{L'}{\sqrt{n}},\quad \max_{q=1,0}\sum_{i\in [k]}\{b_i(q)-\bar b_{[k]}(q)\}^2\le L'\Pik .$$
    \end{lemma}

\begin{lemma}\label{lemma:hajek}
If $u^\mathrm{T}V_R u >0$, $\K/n \rightarrow 0$, Condition~\ref{cond srse}(i)--(ii) and Condition~\ref{cond joint clt} hold, then
$$\frac{E(H'_n-T_n)^2}{\var(H'_n)}\rightarrow 0.$$
\end{lemma}

By Lemma~\ref{lemma:hajek}, $H_n'$ and $T_n$ are asymptotically equivalent in the mean, when the number of blocks $\K$ goes to infinity at a rate much slower than $n$. Considering the independence of the random variables in the summation of $T_n$, we only need to verify the Lindeberg--Feller condition to derive the asymptotic normality of $T_n$.

Recall that
$T_n-E(T_n)=\sumk \sumik \sum_{q=0}^1\{A_i(q)-\bar A_{[k]}(q)\}\cdot b(U_i,q).$ 
By Lemma \ref{lemma:hajek},  we have $\var(T_n)-\var(H_n')\rightarrow 0.$  Therefore, $\var(T_n)\rightarrow u^T V_R u>0$.
It is clear that
\begin{align*}
&\max_{k=1,\ldots,\K }\max_{i\in[k]}|\sum_{q=0}^1\{A_i(q)-\bar A_{[k]}(q)\}\cdot b(U_i,q)|\\
\le & 2 \cdot \max_{k=1,\ldots,\K }\max_{i\in[k]}\max_{q=1,0}|A_i(q)-\bar A_{[k]}(q)|\cdot |b(U_i,q)|.
\end{align*}

By Lemma~\ref{lemma:bound order}, H{\"o}lder inequality, and Condition~\ref{cond joint clt},  we have $$\max_{k=1,\ldots,\K }\max_{i\in[k]}\max_{q=1,0}|A_i(q)-\bar A_{[k]}(q)|\cdot |b(U_i,q)|$$ can be upper bounded by
   \begin{align*}
&\max_{k=1,\ldots,\K }\max_{i\in[k]}\max_{q=1,0} |A_i(q)-\bar A_{[k]}(q)|\cdot \frac{L'}{\sqrt{n}}\\
\le &\max_{q=1,0}\|u_q\|_1\cdot \max_{k=1,\ldots,\K }\max_{i\in[k]}\max_{q=1,0} \|R_i(q)-\bar R_{[k]}(q)\|_{\infty}\cdot \frac{L'}{\sqrt{n}}
\rightarrow  0.
   \end{align*}
Hence, $\forall \varepsilon>0$, when $n$ is sufficient large, we have
$$\max_{k=1,\ldots,\K }\max_{i\in[k]}|\sum_{q=0}^1\{A_i(q)-\bar A_{[k]}(q)\}\cdot b(U_i,q)|<\varepsilon \sqrt{\var(T_n)}.$$
Together with $\var (T_n) \rightarrow u^\T V_R u>0$, it leads to 
\begin{align*} \lim_{n\rightarrow\infty}\sumk \sumik  E & \Big[\sum_{q=0}^1\{A_i(q)-\bar A_{[k]}(q)\}\cdot b(U_i,q)\Big]^2\cdot \\
    &I\Big(|\sum_{q=0}^1\{A_i(q)-\bar A_{[k]}(q)\}\cdot b(U_i,q)|>\varepsilon\sqrt{\var(T_n)}\Big)=0,
\end{align*}
where $I(\cdot)$ is the indicator function.
Thus, Lindeberg's condition holds and we have
$$\frac{T_n-ET_n}{\sqrt{\var (T_n)}}\cd \mathcal N(0,1).$$
By Lemma~\ref{lemma:hajek},  we then have
$$\frac{H_n-EH_n}{\sqrt{\var (H_n)}}\cd \mathcal N(0,1).$$

\textbf{Step 2.} We extend the results of Step 1 to the case of a general $\K$.  Let 
\begin{eqnarray}
  B_1&=&\Big\{k: \nk \ge n^{1/4}/\max_{k=1,...,\K }\max_{i\in [k]}\max_{q=1,0} |A_i(q)-\bar A_{[k]}(q)|^{1/2}\Big\}  \nonumber\\
  B_2&=&\Big\{k: \nk < n^{1/4}/\max_{k=1,...,\K }\max_{i\in [k]}\max_{q=1,0} |A_i(q)-\bar A_{[k]}(q)|^{1/2}\Big\}, \nonumber
\end{eqnarray}
which represent the index set for ``large" and  ``small"  strata, respectively. We have
$$H_n'=\sum_{k\in B_1}\sum_{i\in [k]}\sum_{q=0}^1A_i(q)b_{G_i}(q)+\sum_{k\in B_2}\sum_{i\in [k]}\sum_{q=0}^1A_i(q)b_{G_i}(q):=H_{n1}'+H_{n2}',$$
where $H_{n1}'=\sum_{k\in B_1}\sum_{i\in [k]}\sum_{q=0}^1A_i(q)b_{G_i}(q)$ and    $H_{n2}'=\sum_{k\in B_2}\sum_{i\in [k]}\sum_{q=0}^1A_i(q)b_{G_i}(q).$
Note that $H_{n1}'$ and $H_{n2}'$ are independent. We can use the result of Step 1 if $|B_1|/\sum_{k\in B_1}\nk \rightarrow 0$, where $|B_1|$ denote the cardinality of $B_1$. In fact, $$|B_1|/\sum_{k\in B_1}\nk \le (\max_{k=1,..., \K }\max_{i\in [k]}\max_{q=1,0} n^{-1}|A_i(q)-\bar A_{[k]}(q)|^2)^{\frac 1 4}\rightarrow 0.$$ Hence, according to Step 1, we can induce that $H_{n1}'$ is asymptotically normal if its variance converges to a finite and positive limit. Next,  we consider the remaining term $H_{n2}'$ via the following three cases:
\begin{itemize}
    \item[(i)] $\var (H_{n2}')\rightarrow 0$. This implies $H_{n2}'$ converges to $0$ in probability. Hence, $\var(H_{n1}')=\var(H_{n}')-\var(H_{n2}')\rightarrow u^T V_R u>0$. Then,   the conclusion of the theorem holds.
    \item[(ii)] $\var (H_{n2}')\rightarrow c>0$. It leads to  $\var (H_{n1}')=\var (H_{n}')-\var (H_{n2}')\rightarrow u^\T V_R u-c$.

First, we show that $|B_2|\rightarrow \infty$.
 For $k\in B_2$ and $q=0,1$, by Condition~\ref{cond srse}(i),  there is only $n_{[k]q}$ nonzero values of $b_{G_i}(q)$.  We then have $|\sum_{i\in [k]}\sum_{q=0}^1 \{ A_i(q)-\bar A_{[k]}(q)\}b_{G_i}(q)|$ can be bounded by
\begin{align*}
  &\sum_{q=0}^1 n_{[k]q}\max_{i\in [k]}| A_i(q)-\bar A_{[k]}(q)|\frac{L'}{\sqrt n}\\
  =& L' n_{[k]}(\max_{k=1,..., \K }\max_{i\in [k]}\max_{q=1,0} n^{-1}|A_i(q)-\bar A_{[k]}(q)|^2)^{1/2}\\
 \le & L' (\max_{k=1,..., \K }\max_{i\in [k]}\max_{q=1,0} n^{-1}|A_i(q)-\bar A_{[k]}(q)|^2)^{1/4}
 \rightarrow 0.
\end{align*}
Hence, $\var \{\sum_{i\in [k]}\sum_{q=0}^1 A_i(q)b_{G_i}(q)\}$ uniformly converges to zero. If $|B_2|\not \rightarrow \infty$, then 
$$\var(H_{n_2}')=\sum_{k\in B_2}\var\Big\{\sum_{i\in [k]}\sum_{q=0}^1 A_i(q)b_{G_i}(q)\Big\}\rightarrow 0.$$
This is in contradiction with $\var(H_{n2}')\rightarrow c>0$. Hence, we conclude that $|B_2|\rightarrow \infty$.

Next, we verify Lindeberg's condition for $H_{n2}'$. 
Since $$\max_{k\in B_2}|\sum_{i\in[k]}\sum_{q=0}^1\{A_i(q)-\bar A_{[k]}(q)\}b_{G_i}(q)|\rightarrow 0,$$
then, $\forall \varepsilon >0$, when $n$ is large enough,
$$\max_{k\in B_2}|\sum_{i\in[k]}\sum_{q=0}^1\{A_i(q)-\bar A_{[k]}(q)\}b_{G_i}(q)|<\varepsilon\sqrt{\var(H_{n2}')}.$$

Hence, with $\var(H_{n2'})\rightarrow c>0,$
\begin{align*}
\lim_{n\rightarrow \infty}\sum_{k\in B_2} E&\Big[\sum_{i\in[k]}\sum_{q=0}^1\{A_i(q)-\bar A_{[k]}(q)\}b_{G_i}(q)\Big]^2\cdot  \\
    &I(|\sum_{i\in[k]}\sum_{q=0}^1\{A_i(q)-\bar A_{[k]}(q)\}b_{G_i}(q)|>\varepsilon \sqrt{\var(H_{n2}')})=0,
\end{align*}
   which implies the Lindeberg's condition for $H_{n2}'$ holds.
    \item[(iii)] $\var (H_{n2}')$ does not converge. For any subsequence $H_{n_k}'$ of $H_n'$, there exists a further subsequence $H_{n_{k_l}}'$ such that $\var (H'_{n_{k_l}2}) \rightarrow c'\in [0,u^\T V_R u]$. Then, 
    $$\var (H'_{n_{k_l}1})=\var (H'_{n_{k_l}})-\var (H'_{n_{k_l}2})\rightarrow u^\T V_R u-c'.$$
    Hence, according to previous conclusion, $H'_{n_{k_l}}\cd \mathcal N(0,u^\T V_R u),$ which means for any subsequence $H_{n_k}'$ of $H_n'$, there exists a further subsequence $H_{n_{k_l}}'$ such that $H'_{n_{k_l}}\cd \mathcal N(0,u^\T V_R u).$ Applying \citet[Theorem 2.2.3]{durrett2010probability} for $\pr(H_{n}'\le t)$, $t\in \mathbb R$, $H'_n\cd \mathcal N(0,u^\T V_R u)$ and we finish the proof for  Lemma~\ref{lemma joint clt}. 
\end{itemize}

\end{proof}

\begin{proof}[Proof of Lemma~\ref{lemma:bound order}]
By Condition~\ref{cond srse}(i), we have 
$$\max_{i\in [k]}|b_i(q)|=\frac{1}{\sqrt {n}}\cdot \frac{n \Nk }{Nn_{[k]q}}\le \frac{c_3}{c_1} \frac{1}{\sqrt{n}}.$$
Then, we have
 $$\bar b_{[k]}(q)=\frac{n_{[k]q}}{\Nk}\frac{1}{\sqrt {n}}\cdot \frac{n \Nk}{Nn_{[k]q}}=\frac{\sqrt{n}}{N}.$$
 Hence,  $\sum_{i\in [k]}\{b_i(q)-\bar b_{[k]}(q)\}^2$ can be upper bounded by
 \[  (\Nk -n_{[k]q})\frac{n}{N^2}+n_{[k]q}\frac{n}{N^2}\Big(\frac{\Nk}{n_{[k]q}}-1\Big)^2=\frac{\Nk}{N}\frac{ n}{n_{[k]q}}\frac{N_{[k]}-n_{[k]q}}{N}
     \le \frac{c_3}{c_1} \Pik. \]  
\end{proof}

\begin{proof}[Proof of Lemma~\ref{lemma:hajek}]
    Since $\var(H_n')\rightarrow u^\T V_R u > 0$, it suffices to show that $$E(H_n'-T_n)^2 \rightarrow 0.$$

    Let $\bar{b}_{[k]}(q)$ be the stratum-specific population mean of $b_i(q)$. By  the proof of Theorem 3.1 in \cite{hajek1961some} (see 3.9--3.11), we have
\begin{align*}
    &E(H_n'-T_n)^2\\
    =&\sumk E\Big[\sumik \sum_{q=0}^1A_i(q)b_{G_i}(q)-\{A_i(q)-\bar A_{[k]}(q)\}b(U_i,q)-\bar A_{[k]}(q) b_i(q)\Big]^2\\
    \le &\sumk 2\sum_{q=0}^1 E\Big[\sum_{i\in [k]}A_i(q)b_{G_i}(q)-\{A_i(q)-\bar A_{[k]}(q)\}b(U_i,q)-\bar A_{[m]}(q) b_i(q)\Big]^2\\
    \le &2\sum_{q=0}^1\sumk \frac{1}{\Nk -1}\sum_{i\in [k]}\{A_i(q)-\bar A_{[k]}(q)\}^2\\
    &\times 2\sqrt 2 \times \max_{i\in [k]}|b_i(q)-\bar b_{[k]}(q)|\times \Big[\sum_{i \in [k]}\{b_i(q)-\bar b_{[k]}(q)\}^2\Big]^{1/2}.
\end{align*}
By Condition \ref{cond srse}(i), Condition~\ref{cond joint clt}, and Lemma~\ref{lemma:bound order}, 
\begin{align*}
    &E(H_n'-T_n)^2\\
    \le& 8 \sqrt{2} L (L')^{3/2} \Big(\sumk \frac{1}{\sqrt n}\sqrt{\Pik}\Big)\\
    \le & 8 \sqrt{2} L (L')^{3/2} \Big\{\frac{1}{\sqrt n}  \Big(\sumk \Pik\sumk 1 \Big)^{1/2} \Big\}=  8 \sqrt{2} L (L')^{3/2}\Big(\frac{\K}{n}\Big)^{1/2}\rightarrow 0.
\end{align*}

\end{proof}




Taking suitable values of $R_i(1)$ and $R_i(0)$, we can prove Theorem~\ref{theorem clt-srse}.

\begin{proof}[Proof of Theorem~\ref{theorem clt-srse}]
It is easy to show that $E\hat \tau = \tau$, $E\hat\tau_X = 0$, and $E\hat\delta_W = 0$. We define pseudo potential outcome vectors for unit $i \in [k]$ as 
\begin{align*}
     R_i(1)=  \left(
    \begin{array}{c}
         Y_i(1) - \bar Y(1) \\
           X_i\\
        \ekt ( W_i- {\bar W}_{[k]})
    \end{array}\right), \quad 
    R_i(0)= \left(
    \begin{array}{c}
        Y_i(0) - \bar Y(0)\\
          X_i\\
          -\ekc (W_i- {\bar W}_{[k]})
    \end{array}\right).
\end{align*}
Then,  we have 
\begin{align*}
\barRt - \barRc &=\sumk\Pik \frac{1}{\nkt }\sum_{i\in [k]} Z_i T_i R_i(1)-\sumk\Pik \frac{1}{\nkc }\sum_{i \in [k]} Z_i (1-T_i) R_i(0)\\
&=\sumk\Pik \left(\begin{array}{c}
      \hat\tau_{[k]}\\
      \hat\tau_{[k]X}\\
      \hat\delta_{[k]W}
\end{array}\right)=\left(\begin{array}{c}
      \hat\tau - \tau \\
      \hat\tau_{X}\\
      \hat\delta_{W}
\end{array}\right).
\end{align*}

By Proposition~\ref{prop Cov},
$E\{\sqrt{n}( \barRt - \barRc )\} = 0$ and $\cov \{\sqrt{n}( \barRt - \barRc )\}=  V$. Condition~\ref{cond srse} implies that $V$ has a finite limit and  $\{R_i(1)\}_{i=1}^N$, $\{R_i(0)\}_{i=1}^N\in \mathcal M_L$. By Lemma~\ref{lemma joint clt}, the conclusion of Theorem~\ref{theorem clt-srse} holds. 

\end{proof}

\subsection{Proof of Corollary~\ref{coro:efficiency gain S}}
\begin{proof}
 By setting $\ekt=e_1$ and $\fk=f$ for all $k=1,\ldots,\K $, we have $\Pik=\pik$ and then can express $\Vsr$ as follows:
\begin{align*}
\Vsr &= \sum_{k=1}^{\K}\frac{\Pik^2}{\pik}(\ekt^{-1}S_{[k]1}^2 + \ekc^{-1}S_{[k]0}^2 - \fk S_{[k]\tau}^2) \\
&= \sum_{k=1}^{\K} \Pik(e_1^{-1}S_{[k]1}^2 + e_0^{-1}S_{[k]0}^2 - f S_{[k]\tau}^2) = \sum_{k=1}^{\K} \Pik V_{[k]\tau\tau}.
\end{align*}
Let $S_t^2=(N-1)^{-1}\sum_{i=1}^N \{Y_i(t)-\bar Y(t)\}^2$, $t=0,1$, and $ S_\tau^2=(N-1)^{-1}\sum_{i=1}^N (\tau_i-\tau)^2$ represent the variances of $\{Y_i(t)\}$ and $\{\tau_i\}$, respectively. Then, by \citet[][Theorem 1]{yang2021rejective},
$$
\Vcr = e_1^{-1}S_1^2 + e_0^{-1}S_0^2 - fS_\tau^2.
$$

We can decompose $S_t^2$ as follows:
\begin{align*}
    S_t^2=&(N-1)^{-1}\sumk\sumik \big[ \{ Y_i(t)-\bar Y_{[k]}(t) \} + \{ \bar Y_{[k]}(t)-\bar Y(t) \} \big]^2\\
    =& (N-1)^{-1} \sumk \big[ (\Nk - 1) S_{[k]t}^2+ \Nk \{\bar Y_{[k]}(t)-\bar Y(t)\}^2\} \big].
\end{align*}
Similarly, 
\begin{align*}
S_\tau^2 &= (N-1)^{-1} \sumk  \big\{ (\Nk - 1 ) S_{[k]\tau}^2 + \Nk (\tau_{[k]}-\tau)^2\big\}.
\end{align*}
Hence, we can decompose $\Vcr$ as follows:
\begin{align*}
\Vcr &= (N-1)^{-1} \sumk  (\Nk - 1) V_{[k]\tau\tau} + (N-1)^{-1} \sumk  \Nk d_{[k]},
\end{align*}
where
\begin{eqnarray*}
d_{[k]} &=& e_1^{-1}\{ \bar Y_{[k]}(1)-\bar Y(1) \}^2 + e_0^{-1}\{\bar Y_{[k]}(0)-\bar Y(0)\}^2 - f(\tauk-\tau)^2 \\
& = & [(e_0/e_1)^{1/2}\{\bar Y_{[k]}(1)-\bar Y(1)\}+ (e_1/e_0)^{1/2}\{\bar Y_{[k]}(0)-\bar Y(0)\}]^2\\
& &+(1-f)(\tauk-\tau)^2 \ge 0.
\end{eqnarray*}


Hence, the difference of the asymptotic variances is
$$\Vcr-\Vsr = (N-1)^{-1} \sumk  \Nk d_{[k]} -  (N-1)^{-1} \sumk  (1- \Pik) V_{[k]\tau\tau}.$$


\end{proof}

\subsection{Proof of Theorem~\ref{coro sampling}}
\begin{proof}
    For (i), note that $\Pik/\pik=f/\fk$ and $f=\sumk \Pik\fk$, then we have 
    \begin{align*}
          V_{\tau\tau}&=f \Big[\sumk\Big\{\frac{\Pik}{\fk}\Big(\frac{S_{[k]1}^2}{\ekt }+\frac{S_{[k]0}^2}{\ekc}\Big)-\Pik S_{[k]\tau}^2\Big\} \Big]\\
          &=f \Big[\sumk\Big\{\frac{\Pik^2}{\fk \Pik}\Big(\frac{S_{[k]1}^2}{\ekt }+\frac{S_{[k]0}^2}{\ekc}\Big)-\Pik S_{[k]\tau}^2\Big\} \Big]\\
          &\ge f\Big\{\frac{\sumk \big(\Pik\sqrt{S_{[k]1}^2/\ekt+S_{[k]0}^2/\ekc}\big)^2}{\sumk \fk \Pik}-\sumk \Pik S_{[k]\tau}^2\Big\}\\     &=\sumk\big(\Pik\sqrt{S_{[k]1}^2/\ekt+S_{[k]0}^2/\ekc}\big)^2-f\sumk \Pik S_{[k]\tau}^2,
    \end{align*}
    where the inequality is due to the {Cauchy--Schwarz} inequality with equality holds if and only if $\fk$ is proportion to $\sqrt{S_{[k]1}^2/\ekt+S_{[k]0}^2/\ekc}$. Since $f=\sumk \Pik\fk$, then 
    $$f_{[k]} / f =   \sqrt{e_{[k]1}^{-1}S_{[k]1}^2+e_{[k]0}^{-1}S_{[k]0}^2} \Big/ \Big(\sum_{k' = 1}^{\K} \Pi_{[k']} \sqrt{e_{[k']1}^{-1}S_{[k']1}^2+e_{[k']0}^{-1}S_{[k']0}^2} \Big)$$
    minimizes $V_{\tau\tau}$, the asymptotic variance of $\hat \tau$.

    For (ii), by Cauchy--Schwarz inequality, we have $$\frac{S_{[k]1}^2}{e_{[k]1}}+\frac{S_{[k]0}^2}{e_{[k]0}}\ge \frac{(S_{[k]1}+S_{[k]0})^2}{e_{[k]1}+e_{[k]0}}=(S_{[k]1}+S_{[k]0})^2,$$
    and the equality holds if and only if $e_{[k]1}/e_{[k]0}=\sqrt{S^2_{[k]1}/S^2_{[k]0}}$, or equivalently,
    $$e_{[k]1} = \sqrt{S_{[k]1}^2} \Big/ \Big( \sqrt{S_{[k]1}^2} + \sqrt{S_{[k]0}^2 } \Big).$$   
\end{proof}

\subsection{Proof of Theorem~\ref{theorem clt-SRSRR}}

Before proving Theorem~\ref{theorem clt-SRSRR}, we establish a lemma on the convergence of sample covariance to the population covariance.



\begin{lemma}\label{lemma o_1}
Under Condition~\ref{cond srse} and the stratified randomized survey experiment or the SRSRR experiment, for $t\in\{0,1\}$, we have
$$
\sumk\Pik^2 \pik^{-1}e_{[k]t}^{-1}s^2_{[k]t}-\sumk\Pik^2 \pik^{-1}e_{[k]t}^{-1}S^2_{[k]t}=\op,
$$
$$
\sumk\Pik^2 \pik^{-1} (\ekt \ekc )^{-1} S^2_{[k]X\mid \mathcal S}-\sumk\Pik^2 \pik^{-1}(\ekt \ekc )^{-1} S^2_{[k]X}=\op,
$$
$$\sumk\Pik^2 \pik^{-1}e_{[k]t}^{-1}s_{[k]X,t}-\sumk\Pik^2 \pik^{-1}e_{[k]t}^{-1}S_{[k]X,t}=\op,
$$
$$\sumk\Pik^2 \pik^{-1}(1-\fk )s_{[k]W,t}-\sumk\Pik^2 \pik^{-1}(1-\fk )S_{[k]W,t}=\op.$$
\end{lemma}

\begin{proof}[Proof of Lemma~\ref{lemma o_1}]

We will only prove the first statement, as the proof for the other statements is similar. To prove the first statement, we compute the expectation and upper bound of the variance, $\var(\sumk\Pik^2 \pik^{-1}e_{[k]t}^{-1}s^2_{[k]t})$. 

First, under the stratified randomized survey experiment, note that $s_{[k]t}^2$ is an unbiased estimator of $S_{[k]t}^2$ for $t=0,1$. Thus, 
$$
E\Big(\sumk\Pik^2 \pik^{-1}e_{[k]t}^{-1}s^2_{[k]t}\Big) = \sumk\Pik^2 \pik^{-1}e_{[k]t}^{-1}S^2_{[k]t}.
$$

Second, we show that under the stratified randomized survey experiment, $$\var\Big(\sumk\Pik^2 \pik^{-1}e_{[k]t}^{-1}s^2_{[k]t}\Big) \rightarrow 0.$$
The stratum-specific sample variance $s^2_{[k]t}$ can be decomposed as
$$s_{[k]t}^2=\frac{n_{[k]t}}{n_{[k]t}-1} \Big[\frac{1}{n_{[k]t}}\sum_{i\in[k],\ T_i=t} Z_i \{Y_i(t)-\bar Y_{[k]t}\}^2-  \{\bar Y_{[k]\mid \mathcal S}(t)-\bar Y_{[k]t}\}^2\Big],$$
where  $\bar Y_{[k]\mid \mathcal S}(t)=n_{[k]}^{-1}\sum_{i\in[k]} Z_i Y_i(t)$.
Then 
\begin{align*}
    &\var\Big(\sumk\Pik^2 \pik^{-1}e_{[k]t}^{-1}s^2_{[k]t}\Big)=\sumk \Pik^4 \pik^{-2}e_{[k]t}^{-2} \var(s_{[k]t}^2)\\
    =&\sumk \Pik^4 \pik^{-2}e_{[k]t}^{-2}\frac{n_{[k]t}^2}{(n_{[k]t}-1)^2} \var\Big[\frac{1}{n_{[k]t}}\sum_{i\in[k],\ T_i=t}Z_i\{Y_i(t)-\bar Y_{[k]}(t)\}^2-\{\bar Y_{[k]\mid \mathcal S}(t)-\bar Y_{[k]}(t)\}^2\Big]\\
    \le &8 \sumk\Pik^4 \pik^{-2}e_{[k]t}^{-2} \Big(\var\Big[\frac{1}{n_{[k]t}}\sum_{i\in[k],\ T_i=t}Z_i\{ Y_i(t)-\bar Y_{[k]}(t)\}^2\Big]+\var\big[\{\bar Y_{[k]\mid \mathcal S}(t)-\bar Y_{[k]}(t)\}^2\big]\Big),
\end{align*}
where the last inequality holds due to ${n_{[k]t}^2}/{(n_{[k]t}-1)^2}\le 4$ and $\var(X+Y)\le 2 \{\var(X)+\var(Y)\}$. We consider the two terms separately. The first term  
\begin{align*}
    &8 \sumk \Pik^4 \pik^{-2}e_{[k]t}^{-2} \var\Big[\frac{1}{n_{[k]t}}\sum_{i\in[k],\ T_i=t}\{Y_i(t)-\bar Y_{[k]}(t)\}^2\Big]\\
    \le&  8 \sumk \Pik^4 \pik^{-2}e_{[k]t}^{-2} \Big(\frac{1}{n_{[k]t}}-\frac{1}{\Nk }\Big)\frac{1}{\Nk -1}\sum_{i\in[k]}\{Y_i(t)-\bar Y_{[k]}(t)\}^4\\
    \le& 8\Big(\frac{1}{n}\max_{k=1,...,\K }\max_{i\in[k]}|Y_i(t)-\bar Y_{[k]}(t)|^2\Big)\sumk\Big[ \frac{1}{\Nk -1}\sum_{i\in[k]}\{Y_i(t)-\bar Y_{[k]}(t)\}^2\Big]\frac{\Pik^4 \pik^{-2}e_{[k]t}^{-2} n}{n_{[k]t}}\\
    \le& 8\frac{c_3^2}{c_1^3}\cdot\Big\{\frac{1}{n}\max_{k=1,...,\K }\max_{i\in[k]}|Y_i(t)-\bar Y_{[k]}(t)|^2\Big\}\cdot\Big(\sumk \Pik^2 \pik^{-1} S_{[k]t}^2\Big)
    \rightarrow  0,
\end{align*}
where the last inequality is due to Condition~\ref{cond srse}(i) and 
$$\Pik^2 \pik^{-1}e_{[k]t}^{-1}\frac{n}{n_{[k]t}}=\Big(\frac{f }{\fk e_{[k]t}}\Big)^2\le \frac{c_3^2}{c_1^2}.$$

The second term is  
$8\sumk \Pik^4 \pik^{-2}e_{[k]t}^{-2}\var [\{\bar Y_{[k]\mid \mathcal S}(t)-\bar Y_{[k]}(t)\}^2]$, which is bounded by 
\begin{align*}
     &8\sumk \Pik^4 \pik^{-2}e_{[k]t}^{-2}\max_{i\in[k]}\{ Y_i(t)-\bar Y_{[k]}(t)\}^2\var [\bar Y_{[k]\mid \mathcal S}(t)-\bar Y_{[k]}(t)]\\
    \le &8\sumk \Pik^4 \pik^{-2}e_{[k]t}^{-2}\max_{i\in[k]}\{ Y_i(t)-\bar Y_{[k]}(t)\}^2\Big(\frac{1}{n_{[k]t}}-\frac{1}{\Nk }\Big)\frac{1}{\Nk -1}\sum_{i\in[k]}\{Y_i(t)-\bar Y_{[k]}(t)\}^2\\
    \le &8\Big[\frac{1}{n}\max_{k=1,...,\K }\max_{i\in[k]}\big\{ Y_i(t)-\bar Y_{[k]}(t)\big\}^2\Big]\sumk S_{[k]t}^2\frac{\Pik^4 \pik^{-2}e_{[k]t}^{-2}n}{n_{[k]t}}\\
    \le& 8\frac{c_3^2}{c_1^3}\cdot\Big\{\frac{1}{n}\max_{k=1,...,\K }\max_{i\in[k]}|Y_i(t)-\bar Y_{[k]}(t)|^2\Big\}\cdot\Big(\sumk \Pik^2 \pik^{-1} S_{[k]t}^2\Big)
    \rightarrow  0.
\end{align*}
Thus, we have
$\var(\sumk\Pik^2 \pik^{-1}e_{[k]t}^{-1}s^2_{[k]t})\rightarrow 0$
and the first statement holds under the stratified randomized survey experiment.

Additionally, under the SRSRR experiment, it is equivalent to conditioning on $(\MS\le a_S,\ \MT\le a_T)$. By Theorem~\ref{theorem clt-srse},  we have 
$$
\pr( \MS\le a_S,\ \MT\le a_T ) \rightarrow \pr( \chi^2_{J_1} \leq a_S, \ \chi^2_{J_2} \leq a_T ) > 0.
$$
Thus, for any $\epsilon > 0$,
\begin{eqnarray*}
&&\pr \Big( \Big| \sumk\Pik^2 \pik^{-1}e_{[k]t}^{-1}s^2_{[k]t}-\sumk\Pik^2 \pik^{-1}e_{[k]t}^{-1}S^2_{[k]t} \Big| > \epsilon  \mid \MS\le a_S,\ \MT\le a_T \Big) \\
& \leq & \pr \Big( \Big| \sumk\Pik^2 \pik^{-1}e_{[k]t}^{-1}s^2_{[k]t}-\sumk\Pik^2 \pik^{-1}e_{[k]t}^{-1}S^2_{[k]t} \Big| > \epsilon  \Big) / \pr( \MS\le a_S,\ \MT\le a_T ) \\
&\rightarrow &  0.
\end{eqnarray*}
That is, the first statement holds under the SRSRR experiment.

\end{proof}

Next, we prove Theorem~\ref{theorem clt-SRSRR}.
\begin{proof}[Proof of Theorem~\ref{theorem clt-SRSRR}]
By Theorem~\ref{theorem clt-srse}, we have $\sqrt n(\hat\tau-\tau,\hat\tau_X^\T,\hat\delta_W^\T)^\T \cd N(0,V)$. Let $(A,B^\T,C^\T)^\T \sim N(0,V)$ denote the limiting distribution of $\sqrt n(\hat\tau-\tau,\hat\tau_X^\T,\hat\delta_W^\T)^\T$. Let $V_1 \asyequ V_2$ denote two random vectors $V_1$ and $V_2$ having the same asymptotic distribution. 
The asymptotic distribution of $\sqrt n(\hat\tau-\tau)$ under the SRSRR experiment is the asymptotic distribution of $\sqrt n(\hat\tau-\tau)$ conditional on $(\MS\le a_S,\ \MT\le a_T)$, i.e.,
$$
\sqrt n(\hat\tau-\tau) \mid (\MS\le a_S,\ \MT\le a_T).
$$
Note that 
$
\MS=\hat \delta_W ^\mathrm{T} \cov (\hat \delta_W)^{-1}\hat \delta_W$ and $\MT =\tauunadj_X^\mathrm{T} \cov (\tauunadj_X \mid \mathcal S)^{-1}\tauunadj_X.
$
By Lemma~\ref{lemma o_1}, we can obtain $\cov(\sqrt{n} \hat \tau_X \mid \mathcal S) - \cov(\sqrt{n} \hat \tau_X) = \op$. Then,
\begin{align*}
\sqrt n (\hat\tau-\tau)\mid \textnormal{SRSRR} \asyequ &\sqrt n (\hat\tau-\tau)\mid\sqrt n \hat\tau_X^\mathrm{T}V_{XX}^{-1}\sqrt n\hat\tau_X\le a_T,\sqrt n\hat \delta_W^\mathrm{T} V_{WW}^{-1}\sqrt n\hat\delta_W\le a_S\\
\asyequ  &A\mid B^\mathrm{T}V_{XX}^{-1}B\le a_T,C^\mathrm{T}V_{WW}^{-1}C\le a_S.\end{align*}
Consider $\xi=A-V_{\tau X}V_{XX}^{-1}B-V_{\tau W}V_{WW}^{-1}C$. Accordingly,  
$$\cov (\xi,B)=V_{\tau X}-(V_{\tau X}V_{XX}^{-1})V_{XX}-0=0,$$
$$\cov (\xi,C)=V_{\tau W}-0-(V_{\tau W}V_{WW}^{-1})V_{WW}=0.$$
Note that $\cov (B,C)=0$. Hence, normal random variables $\xi,B,C$ are independent. It implies that  $A\mid B^\mathrm{T}V_{XX}^{-1}B\le a_T,C^\mathrm{T}V_{WW}^{-1}C\le a_S$ further equals to
\begin{align*}
&\xi+V_{\tau X}V_{XX}^{-1}B+V_{\tau W}V_{WW}^{-1}C\mid B^\mathrm{T}V_{XX}^{-1}B\le a_T,C^\mathrm{T}V_{WW}^{-1}C\le a_S \\
=& \xi + (V_{\tau X}V_{XX}^{-1}B \mid B^\mathrm{T}V_{XX}^{-1}B\le a_T) + (V_{\tau W}V_{WW}^{-1}C\mid C^\mathrm{T}V_{WW}^{-1}C\le a_S),
\end{align*}
where the above three terms are independent. The first term $\xi$ is a normal variable with zero mean and variance
$\var(\xi)=\cov (\xi,A)=V_{\tau\tau}-V_{\tau X}V_{XX}^{-1}V_{X\tau}-V_{\tau W}V_{WW}^{-1}V_{W\tau}.$
Note that $V_{XX}^{-1/2}B\sim \mathcal N(0,I_{J_2})$, then, 
$$V_{\tau X}V_{XX}^{-1}B\mid B^\mathrm{T}V_{XX}^{-1}B\le a_T=V_{\tau X}V_{XX}^{-1/2}(V_{XX}^{-1/2}B)\mid(V_{XX}^{-1/2}B)^\mathrm{T}(V_{XX}^{-1/2}B)\le a_T.$$
Thus, the second term follows a truncated normal distribution with scaling $V_{\tau X}V_{XX}^{-1}V_{X\tau}$ (i.e, $V_{\tau X}V_{XX}^{-1}B\mid B^\mathrm{T}V_{XX}^{-1}B\le a_T \sim V_{\tau\tau}^{1/2} \sqrt{R^2_X} L_{J_2, a_T}$; see \cite{Li2018} and \cite{yang2021rejective}). Similarly,  $V_{WW}^{-1/2}C\sim \mathcal N(0,I_{J_1})$,
$$V_{\tau W}V_{WW}^{-1}C\mid C^\mathrm{T}V_{WW}^{-1}C\le a_S=V_{\tau W}V_{WW}^{-1/2}(V_{WW}^{-1/2}C)\mid(V_{WW}^{-1/2}C)^\mathrm{T}(V_{WW}^{-1/2}C)\le a_S.$$
Thus, the third term follows a truncated normal distribution with scaling $V_{\tau W}V_{WW}^{-1}V_{W\tau}$.


\end{proof}

\subsection{Proof of Corollary~\ref{coro avar}}

\begin{proof}
By Theorem~\ref{theorem clt-SRSRR}, the asymptotic variance of $\sqrt n(\hat\tau-\tau)$ under the SRSRR experiment is 
\begin{align*}
&\var\Big\{V_{\tau\tau}^{1/2}\big(\sqrt{1-R_W^2-R_X^2}\cdot \varepsilon +\sqrt{R_W^2}\cdot L_{J_1,a_S}+\sqrt{R_X^2}\cdot L_{J_2,a_T}\big)\Big\} \\
=&V_{\tau\tau}\Big\{ (1-R_W^2-R_X^2)+R_W^2\cdot \nu_{J_1,a_S}+R_X^2\cdot \nu_{J_2,a_T}) \Big\} \\
=&\Big\{1-(1-\nu_{J_1,a_S})R_W^2-(1-\nu_{J_2,a_T})R_X^2\Big\} V_{\tau\tau}.
\end{align*}
By Theorem~\ref{theorem clt-srse}, the asymptotic variance of $\sqrt n(\hat\tau-\tau)$ under the stratified randomized survey experiment is $V_{\tau\tau}$. Thus, the percentage reduction in the asymptotic variance is
$ [ 100\{(1-\nu_{J_1,a_S})R_W^2+(1-\nu_{J_2,a_T})R_X^2\}] \%$.

Since the truncated normal distribution is more concentrated at zero than the normal (or truncated normal) distribution with the same or larger variance \citep[][Lemma A3]{Li2018},  the asymptotic distribution of $\tauunadj$ under the SRSRR experiment is more concentrated at $\tau$ than that under the stratified randomized survey experiment. 

\end{proof}

\subsection{Proof of Theorem~\ref{theorem conserv}}


\begin{proof}
Let $V_1 \prec V_2$ denote random variable $V_1$ being asymptotically more concentrated at zero than $V_2$. By Proposition~\ref{prop Cov},
$$V_{\tau\tau}=\sumk \Pik^2\pik^{-1}(\ekt^{-1}S_{[k]1}^2+\ekc^{-1}S_{[k]0}^2-\fk S_{[k]\tau}^2).$$
Recall that
$\hat V_{\tau\tau}=\sumk \Pik^2\pik^{-1}(\ekt^{-1}\skt^2+\ekc^{-1}\skc^2).$
By Lemma~\ref{lemma o_1}, {we have
$\hat V_{\tau\tau}-V_{\tau\tau}=\sumk \Pik^2\pik^{-1} \fk S_{[k]\tau}^2+\op.$}

Similarly, by Proposition~\ref{prop Cov} and Lemma~\ref{lemma o_1},  we have
$$\hat V_{X X}-V_{X X}=\op,\quad \hat V_{X \tau}-V_{X \tau}=\op,\quad \hat V_{W \tau }-V_{W \tau}=\op.$$
Moreover, by definition, 
$\hat V_{\tau\tau}\hat R_{X}^2- V_{\tau\tau} R_{X}^2=\op$ and $\hat V_{\tau\tau}\hat R_{W}^2- V_{\tau\tau} R_{W}^2=\op.$
Thus, 
\begin{eqnarray*}
&&V_{\tau\tau}^{1/2}\Big \{\sqrt{1-R_W^2-R_X^2}\cdot \varepsilon +\sqrt{R_W^2}\cdot L_{J_1,a_S}+\sqrt{R_X^2}\cdot L_{J_2,a_T} \Big\}  \\
&\prec& \hat V_{\tau\tau}^{1/2}\Big \{\sqrt{1- \hat R_W^2- \hat R_X^2}\cdot \varepsilon +\sqrt{ \hat R_W^2}\cdot L_{J_1,a_S}+\sqrt{ \hat R_X^2}\cdot L_{J_2,a_T} \Big\}. 
\end{eqnarray*}
That is, 
$$
\hat V_{\tau\tau}^{1/2}\Big \{\sqrt{1- \hat R_W^2- \hat R_X^2}\cdot \varepsilon +\sqrt{ \hat R_W^2}\cdot L_{J_1,a_S}+\sqrt{ \hat R_X^2}\cdot L_{J_2,a_T} \Big\}
$$
is a conservative estimator for the asymptotic distribution of $\sqrt{n}(\hat\tau - \tau)$ under the SRSRR experiment. By \citet[Lemma~B21]{yang2021rejective},
 $$    \big[\hat\tau- n^{-1/2} \nu_{1-\alpha/2}(\hat V_{\tau\tau}, \hat R_W^2,\hat R_X^2), \ \hat\tau + n^{-1/2}\nu_{1-\alpha/2}(\hat V_{\tau\tau},\hat R_W^2,\hat R_X^2) \big]
    $$
    is an asymptotic conservative $1-\alpha$ confidence interval for $\tau$.
\end{proof}

\subsection{Proof of Theorem~\ref{theorem S-optimal}}
\begin{proof}
We try to adjust potential outcomes to apply Theorems \ref{theorem clt-srse} and \ref{theorem clt-SRSRR} and derive the asymptotic property. We will transform the target estimator into an easily handled one.

Let 
$
\hat \tau_{\opt}=\hat \tau-\beta_{\opt}^\T \hat \tau_C-\gamma_{\opt}^\T  \hat \delta_E
$ be the optimal projection estimator. First, we show that $\tauadj$ has the same asymptotic distribution as $\hat \tau_{\opt}$. Applying Lemma~\ref{lemma o_1} to the covariates $C_i$ and $E_i$, we have
$\hat\beta-\beta_{\opt}=\op$ and  $\hat\gamma- \gamma_{\opt}=\op.$
Applying Theorem \ref{theorem clt-srse} to the potential outcomes $Y_i(t)$ and covariates $C_i$ and $E_i$, we have 
$$\hat\tau_C=O_p(n^{-1/2}), \quad \hat\delta_E=O_p(n^{-1/2}).
$$ 
Thus, it leads to 
$$\tauadj- \hat\tau_{\opt}=(\beta_{\opt}-\hat\beta)^\mathrm{T}\hat\tau_C+(\gamma_{\opt}-\hat\gamma)^\mathrm{T}\hat\delta_E=o_{p}(n^{-1/2}).$$
Hence,  we have
$$\sqrt{n}(\tauadj-\tau)=\sqrt{n}(\hat\tau_{{\opt}}-\tau)+\op.$$
It implies that $\tauadj$ has the same asymptotic distribution as $\hat\tau_{{\opt}}$.

Second, we show that, under the stratified randomized survey experiment, $\sqrt{n}(\hat\tau_{{\opt}} - \tau)$ converges in distribution to $\mathcal N(0, (1-R_E^2-R_C^2) V_{\tau\tau}  )$.
For $i\in [k]$, considering  the linear adjusted potential outcomes $Y_{i}^\dagger(1)$and $Y_{i}^\dagger(0)$ 
defined by $$Y_{i}^\dagger(1)=Y_i(1)-\beta_{\opt} ^\T C_i-\ekt \gamma_{\opt}^\T (E_i-E_{[k]}),$$
$$Y_{i}^\dagger(0)=Y_i(0)-\beta_{\opt}^\T  C_i+\ekc \gamma_{\opt} ^\T (E_i-E_{[k]}).$$
Then $\{Y_i^\dagger(1)\},\{Y_i^\dagger(0)\}\in \mathcal M_L.$ Define $\hat\tau_{[k]}^\dagger$, $\tau_{[k]}^\dagger$,  $\hat\tau^\dagger$, $\tau^\dagger$, $(R_X^\dagger)^2$, $(R_W^\dagger)^2$, and $V_{\tau\tau}^\dagger$ similarly to $\hat\tau_{[k]}$, $\tau_{[k]}$, $\hat\tau$, $\tau$, $R_X^2$, $R_W^2$, and $V_{\tau\tau}$ with $\{Y_i(1)\},\{Y_i(0)\}$ replaced by  $\{Y_i(1)^\dagger\},\{Y_i(0)^\dagger\}$. Then, for $k=1,...,\K $,
$$\tau^\dagger_{[k]}=\frac{1}{\nk }\sum_{i\in [k]}\big\{Y_i^\dagger(1)-Y_i^\dagger(0)\big\}=\tau_{[k]}, \quad \hat\tau^\dagger_{[k]}=\tauunadj_{[k]}-\beta_{\opt}^\mathrm{T}\hat\tau_{[k]C}-\gamma_{\opt}^\mathrm{T} \hat\delta_{[k]E},$$
$$\tau^\dagger=\sumk\Pik \tau_{[k]}^\dagger=\sumk\Pik \tau_{[k]}=\tau,$$
and
$$\hat\tau^\dagger =\sumk\Pik \hat\tau_{[k]}^\dagger =\sumk\Pik (\tauunadj_{[k]}-\beta_{\opt}^\mathrm{T}\hat\tau_{[k]C}-\gamma_{\opt}^\mathrm{T} \hat\delta_{[k]E})=\hat\tau-\beta_{\opt}^\mathrm{T}\tauunadj_C-\gamma_{\opt}^\mathrm{T}\hat \delta_E=\hat\tau_{{\opt}}.$$
Applying Theorem~\ref{theorem clt-srse} to $Y_{i}^\dagger(1),Y_{i}^\dagger(0)$, we have
$$
\sqrt{n}(\hat\tau_{{\opt}} - \tau) = \sqrt{n}(\hat\tau^\dagger - \tau^\dagger) \cd \mathcal N(0,  V_{\tau\tau}^\dagger  ).
$$ 
Moreover,  we have $V_{\tau\tau}^\dagger= n \var (\hat \tau_{{\opt}})
    = (1-R_E^2-R_C^2)V_{\tau\tau}.$

Third, we show that the conclusion holds under the SRSRR experiment. Applying Theorem \ref{theorem clt-SRSRR} to $Y_{i}^\dagger(1)$and $Y_{i}^\dagger(0)$, we have, under the  SRSRR experiment,
\begin{align*}
    &\sqrt{n}(\hat \tau_{{\opt}}-\tau)\\
    =&\sqrt{n} (\hat \tau ^\dagger-\tau^\dagger) \\
    \cd & (V_{\tau\tau}^\dagger)^{1/2}\Big \{\sqrt{1-(R^\dagger_W)^2-(R_X^\dagger)^2}\cdot \varepsilon +\sqrt{(R_W^\dagger)^2}\cdot L_{J_1,a_S}+\sqrt{(R_X^\dagger)^2}\cdot L_{J_2,a_T} \Big\},
\end{align*}
where $\varepsilon,L_{J_1,a_S},$ and $L_{J_2,a_T}$ are independent. 

By definition,
\begin{align*}
    V_{\tau\tau}^\dagger(R_W^\dagger)^2=&n \cov(\hat\tau ^\dagger,\hat \delta_W)\cov(\hat\delta_W)^{-1}\cov(\hat \delta_W,\hat\tau ^\dagger).
\end{align*}
Since $\cov(\hat\tau_C,\hat\delta_E)=0$, then,
$$\cov(\hat\tau ^\dagger,\hat \delta_W)=\cov(\hat\tau -\beta_{\opt}^\T \hat\tau_C-\gamma_{\opt} ^\T \hat\delta_E,\hat \delta_W)=\cov(\hat\tau -\gamma_{\opt} ^\T \hat\delta_E,\hat \delta_W)=0,$$
where the last equation holds because of the property of linear projection and $W_i \subset E_i$. Thus,
$(R_W^\dagger) ^2=0.$
Similarly,  we can show that
$(R_X^\dagger)^2=0.$

Therefore, under the SRSRR experiment, we have
$$ \sqrt{n} ( \tauadj -\tau ) \cd \mathcal N(0, (1-R_E^2-R_C^2) V_{\tau\tau}  ).$$

Finally, we prove that $\tauadj$ (or equivalently $\hat \tau_{{\opt}}$) has the smallest asymptotic variance among the class of linearly adjusted estimators $\{ \hat\tau-\beta^\mathrm{T}\tauunadj_C-\gamma^\mathrm{T}\hat \delta_E, \ \beta \in \mathbb R^{J_4}, \ \gamma \in \mathbb R^{J_3} \}.$

Since $\cov(\hat \tau_C,\hat \delta_E)=0$ and by the property of linear projection, we have
$$\cov(\hat \tau_{{\opt}},\hat \tau_C)=\cov(\hat \tau-\beta_{\opt}^\T \hat \tau_C-\gamma_{\opt}^\T \hat \delta_E,\hat \tau_C)=\cov(\hat \tau-\beta_{\opt}^\T \hat \tau_C,\hat \tau_C)=0.$$
Similarly, $\cov(\hat \tau_{{\opt}},\hat \delta_E)=0.$
Thus, $n\var(\hat \tau-\beta^\T \hat \tau_C-\gamma ^\T \hat \delta_E) $ equals to 
\begin{align*}
&n\var\Big\{\hat\tau_{{\opt}}-(\beta-\beta_{\opt})^\T \hat \tau_C-(\gamma-\gamma_{\opt})^\T \hat\delta_E\Big\}\\
=& (1-R_E^2-R_C^2)V_{\tau\tau}+(\beta-\beta_{\opt})^\T V_{CC}(\beta-\beta_{\opt})+(\gamma-\gamma_{\opt})^\T V_{EE}(\gamma-\gamma_{\opt})\\
\ge & (1-R_E^2-R_C^2)V_{\tau\tau}.
\end{align*}
Thus, $\hat\tau_{\opt}$ and also $\tauadj$ have the smallest asymptotic variance among the class of linearly adjusted estimators $\{ \hat\tau-\beta^\mathrm{T}\tauunadj_C-\gamma^\mathrm{T}\hat \delta_E, \ \beta \in \mathbb R^{J_4}, \ \gamma \in \mathbb R^{J_3} \}.$
\end{proof}

\subsection{Proof of Theorem \ref{theorem ci-adj}}

\begin{proof}

By Lemma~\ref{lemma o_1}, we have shown that 
$$
\hat V_{\tau\tau}-V_{\tau\tau}=\sumk \Pik^2\pik^{-1}\fk S_{[k]\tau}^2+\op.
$$
Moreover, applying Lemma~\ref{lemma o_1} to the covariates $E_i$ and $C_i$, we have
$$
\hat V_{\tau\tau} \hat R_E^2 - V_{\tau\tau} R_E^2 = \op, \quad \hat V_{\tau\tau} \hat R_C^2 - V_{\tau\tau} R_C^2 = \op.
$$
Therefore, we have
$
\hat V_{\tau\tau} (1 - \hat R_E^2 - \hat R_C^2) -  V_{\tau\tau} (1 - R_E^2 - R_C^2) = \sumk \Pik^2\pik^{-1} \fk S_{[k]\tau}^2+\op.
$
Accordingly, the confidence interval
$$
\Big[ \tauadj - n^{-1/2}q_{1-\alpha/2}\hat V_{\tau\tau}^{1/2}\sqrt{1-\hat R_E^2-\hat R_C^2},\ \tauadj + n^{-1/2}q_{1-\alpha/2}\hat V_{\tau\tau}^{1/2}\sqrt{1-\hat R_E^2-\hat R_C^2} \Big]
$$
has an asymptotic coverage rate greater than or equal to $1 - \alpha$. 

For the second part of this theorem, we utilize Lemma~\ref{lemma concen} below, which has been established by \cite{yang2021rejective}. We omit its proof here.
\begin{lemma}[\cite{yang2021rejective}, Lemma B15]\label{lemma concen}
    For any positive integers $K_1$, $K_2$ and constants $a_1$, $a_2$, suppose that  $\varepsilon,L_{K_1,a_1},L_{K_2,a_2}$ are mutually independent. Then, for any nonnegative constants $b_0\le \bar b_0$, $b_1\le \bar b_1$, $b_2\le \bar b_2$, and any constant $c>0$,
$$\pr\Big(|b_0\varepsilon+b_1L_{K_1,a_1}+b_2L_{K_2,a_2}|\le c\Big)\le \pr\Big(|\bar b_0\varepsilon+\bar b_1L_{K_1,a_1}+\bar b_2L_{K_2,a_2}|\le c\Big).$$
\end{lemma}

Since $W_i \subset E_i$ and $X_i \subset C_i$, then, $R_E^2 \geq R_W^2$ and $R_C^2 \geq R_X^2$. Taking $b_0=V_{\tau\tau}^{1/2}\sqrt{1-R_E^2-R_C^2}$, $\bar b_0=V_{\tau\tau}^{1/2}\sqrt{1-R_W^2-R_X^2}$, $b_1=b_2=0$, $\bar b_1= V_{\tau\tau}^{1/2}\sqrt{R_W^2}$ and $\bar b_2=V_{\tau\tau}^{1/2}\sqrt{R_X^2}$, the asymptotic distribution of $\tauadj$ is more concentrated around $\tau$ than $\hat\tau$ under the SRSRR experiment. Moreover, by the convergence of $\hat V_{\tau\tau}$, $\hat V_{\tau\tau} \hat R_E^2$, $\hat V_{\tau\tau} \hat R_C^2$, $\hat V_{\tau\tau} \hat R_W^2$, and $\hat V_{\tau\tau} \hat R_X^2$, the proposed confidence interval in Theorem~\ref{theorem ci-adj} is asymptotically shorter than, at least as short as, the confidence interval based on $\tauunadj$ in Theorem~\ref{theorem conserv}.

\end{proof}

\subsection{Proof of Theorem~\ref{coro sampling-srrse}}

\begin{proof}
By Theorem \ref{theorem clt-SRSRR}, under the SRSRR experiment,  we have $$\var(\sqrt n\hat \tau)=V_{\tau\tau}-(1-\nu_{J_1,a_S})V_{\tau W}V_{WW}^{-1}V_{W\tau}-(1-\nu_{J_2,a_T})V_{\tau X}V_{XX}^{-1}V_{X\tau}.$$
To prove  (i), by Proposition \ref{prop Cov}, we have 
$$\frac{\partial}{\partial(1/\fk)}V_{\tau\tau}=\Pik f (\ekt^{-1}S_{[k]1}^2+\ekc^{-1} S_{[k]0}^2),$$
$$\frac{\partial}{\partial(1/\fk)}V_{\tau W}=\Pik f S_{[k]\tau W},\quad\frac{\partial}{\partial(1/\fk)}V_{WW}=\Pik f S_{[k] W}^2,$$
$$\frac{\partial}{\partial(1/\fk)}V_{\tau X}=\Pik f (\ekt^{-1}S_{[k]1,X}+\ekc^{-1}S_{[k]0,X}),\quad\frac{\partial}{\partial(1/\fk)}V_{XX}=\Pik f (\ekt\ekc)^{-1}S_{[k] X}^2.$$

Since $\fk$ satisfies 
$\sumk\Pik \fk=f,$
then, considering the Lagrange function:
$$\var(\sqrt n\hat \tau)+\lambda \Big(\sumk\Pik \fk-f \Big).$$
For the optimal $\{ \fk\}_{k=1}^{\K}$, we have
\begin{align*}
    0=&\frac{\partial}{\partial (1/f_{[k]})}\Big\{\var(\sqrt n\hat \tau)+\lambda \Big(\sumk\Pik \fk-f \Big)\Big\}\\
    =&\Pik f (\ekt^{-1}S_{[k]1}^2+\ekc^{-1} S_{[k]0}^2)-2(1-\nu_{J_1,a_S})\Pik f S_{[k]\tau W}V_{WW}^{-1}V_{W\tau}\\
    +&(1-\nu_{J_1,a_S})\Pik f V_{\tau W}V_{WW}^{-1}S_{[k]W}^2V_{WW}^{-1}V_{W\tau}\\
    -&2(1-\nu_{J_2,a_T})\Pik f (\ekt^{-1}S_{[k]1,X}+\ekc^{-1}S_{[k]0,X})V_{XX}^{-1}V_{X\tau}\\
+&(1-\nu_{J_2,a_T})\Pik f(\ekt\ekc)^{-1} V_{\tau X}V_{XX}^{-1}S_{[k]X}^2V_{XX}^{-1}V_{X\tau}-\lambda\Pik \fk^2.
\end{align*}
Hence, 
$$\Big(\frac{|\lambda|} f\Big)^{1/2}=\frac{A_k}{\fk}=\frac{\Pik A_k}{\Pik \fk}=\frac{\sumk\Pik A_k}{\sumk \Pik \fk}=\frac{\sumk\Pik A_k}{f}.$$

To prove (ii) of the theorem, we have
$$\frac{\partial}{\partial(e_{[k]t})}V_{\tau\tau}=-\Pik^2\pik^{-1}  e_{[k]t}^{-2}S_{[k]t}^2,$$
$$\frac{\partial}{\partial(e_{[k]t})}V_{\tau X}=-\Pik^2\pik^{-1}  e_{[k]t}^{-2}S_{[k]t,X},\quad \frac{\partial}{\partial(e_{[k]t})}V_{XX}=-\Pik^2\pik^{-1}  e_{[k]t}^{-2}S_{[k]X}^2, \quad t=0,1.$$
Since $\ekt,\ekc$ satisfies $\ekt+\ekc=1$,  considering the Lagrange function
$$\var(\sqrt n \tauunadj)+\sumk\lambda_k(\ekt+\ekc-1).$$

For the optimal $\{e_{[k]t}\}_{k=1}^{\K}$, $t=0,1$, we have 
\begin{align*}
    0=&\frac{\partial}{\partial e_{[k]t}}\Big\{\var(\sqrt n \tauunadj)+\sumk\lambda_k(\ekt+\ekc-1) \Big\}\\
    =&\frac{\Pik^2}{\pik} \Big(-\frac{S_{[k]t}^2}{e_{[k]t}^2}\Big)+2(1-\nu_{J_2,a_T})\frac{\Pik^2}{\pik} \frac{S_{[k]t,X}}{e_{[k]t}^2}V_{XX}^{-1}V_{X\tau}\\
    &-(1-\nu_{J_2,a_T})V_{\tau X}V_{XX}^{-1}\frac{\Pik^2}{\pik}\frac{S_{[k]t}^2}{e_{[k]t}^2}V_{XX}^{-1}V_{X\tau}+\lambda_k.
\end{align*}
Hence, we have
\begin{align*}
    \frac{\ekt}{\ekc}=&\Big\{ \Big|\frac{S_{[k]1}^2+(1-\nu_{J_2,a_T})V_{\tau X}V_{XX}^{-1}S_{[k]X}^2V_{XX}^{-1}V_{X\tau }-2(1-\nu_{J_2,a_T})S_{[k]1,X}V_{XX}^{-1}V_{X\tau}}{S_{[k]0}^2+(1-\nu_{J_2,a_T})V_{\tau X}V_{XX}^{-1}S_{[k]X}^2V_{XX}^{-1}V_{X\tau }-2(1-\nu_{J_2,a_T})S_{[k]0,X}V_{XX}^{-1}V_{X\tau}}\Big|\Big\}^{1/2} \\
    =& \frac{a_{[k]1}}{a_{[k]0}}.
\end{align*}
Thus,
$$e_{[k]1} = ( a_{[k]1} ) / ( a_{[k]1} + a_{[k]0} ).$$
\end{proof}

\subsection{Proof of Proposition~\ref{prop Cov}}
\begin{proof}
First, we show that $\hat \tau$ is an unbiased estimator of $\tau$. Since $E\hat\tau_{[k]}=\tau_{[k]}$, then 
$$E\hat\tau=\sumk\Pik E\hat\tau_{[k]}=\sumk\Pik \tau_{[k]}=\tau.$$

Next, we compute the covariance of $\hat \tau$. Simple calculation gives 
\begin{align*}
      &\cov (\barWS-\bar W)=\cov \Big(\sumk \Pik \barWSk\Big)
    = \cov \Big\{ \sumk\Pik \frac{1}{\nk } \sum_{i\in [k]}Z_i(W_i-\bar W_{[k]}) \Big\} \\
    =&\sumk\frac{\Pik ^2}{\nk ^2} \cov \Big\{ \sum_{i\in [k]}Z_i(W_i-\bar W_{[k]})\Big\} \\
    &
 +  \sum_{ k_1\neq k_2 }\frac{\Pi_{[k_1]}}{n_{[k_1]}}\frac{\Pi_{[k_2]}}{n_{[k_2]}} \cov \Big\{\sum_{i \in [k_1]}Z_i(W_i-\bar W_{[k_1]}), \sum_{j\in [k_2]}Z_j(W_j-\bar W_{[k_2]})\Big\}.
\end{align*}
For $1\le k\le \K$,  
$ \cov \{ \sum_{i\in[k]} Z_i (W_i-\bar W_{[k]})\}$ 
can be further expressed as 
\begin{align*}
   &E\Big[\Big\{ \sum_{i \in [k]}Z_i(W_i-\bar W_{[k]})\Big\}\Big\{\sum_{i \in [k]}Z_i(W_i-\bar W_{[k]})^\mathrm{T}\Big\}\Big]\\
=&\sum_{i\in [k]} E\big\{ Z_i(W_i-\bar W_{[k]})(W_i-\bar W_{[k]})^\mathrm{T} \big\}  +\sum_{i\neq j;i,j\in [k]}E\big\{Z_iZ_j(W_i-\bar W_{[k]})(W_j-\bar W_{[k]})^\mathrm{T}\big\}\\
=&\sum_{i\in [k]} \frac{\nk }{\Nk}(W_i-\bar W_{[k]})(W_i-\bar W_{[k]})^\mathrm{T} +\sum_{i\neq j;i,j\in [k]}\frac{\nk }{\Nk}\frac{\nk -1}{\Nk-1}(W_i-\bar W_{[k]})(W_j-\bar W_{[k]})^\mathrm{T}\\
=&\sum_{i\in [k]} \Big(\frac{\nk }{\Nk}-\frac{\nk }{\Nk}\frac{\nk -1}{\Nk-1}\Big)(W_i-\bar W_{[k]})(W_i-\bar W_{[k]})^\mathrm{T}=\nk \Big(1-\frac{\nk }{\Nk}\Big)S_{[k]W}^2,
\end{align*}
where the second to last equality is due to
\begin{align*}
0=&\sum_{i,j\in [k]}(W_i-\bar W_{[k]})(W_j-\bar W_{[k]})^\mathrm{T} \\
=&\sum_{i\in [k]} (W_i-\bar W_{[k]})(W_i-\bar W_{[k]})^\mathrm{T}
+\sum_{i\neq j;i,j\in [k]}(W_i-\bar W_{[k]})(W_j-\bar W_{[k]})^\mathrm{T}.
\end{align*}
For $k_1 \neq k_2$, because the sampling is independent across strata, we have
\begin{align*}
    &\cov \Big\{ \sum_{i \in [k_1]}Z_i(W_i-\bar W_{[k_1]}),\sum_{j\in [k_2]}Z_j(W_j-\bar W_{[k_2]}) \Big\} = 0.
\end{align*}
Hence,
\begin{equation}\label{eqn:prop1}
\cov(\hat\delta_W) = \cov (\barWS-\barW)=\sumk \Pik ^2 \Big(\frac{1}{\nk }-\frac{1}{\Nk} \Big)S_{[k]W}^2. \nonumber
\end{equation}
Moreover, $ \cov ( \hat\tau_X \mid\mathcal S )
      =\cov \{\sumk \Pik ( \bar X_{[k]1}-\bar X_{[k]0}\mid\mathcal S) \}$, which further equals to
\begin{align*}
    &\sumk\Pik ^2 \cov \Big\{ \frac{1}{\nkt }\sum_{i \in [k],i \in \mathcal S}X_i T_i-\frac{1}{\nkc }\sum_{i\in [k],i \in \mathcal S}X_i(1-T_i)\mid\mathcal S \Big\} \\
    =& \sumk\Pik ^2 \cov \Big\{\sum_{i\in[k], i\in \mathcal S} \Big(\frac{1}{\nkt }+\frac{1}{\nkc }\Big)T_i(X_i-\bar X_{[k]})\mid\mathcal S \Big\}\\
    =&\sumk \Pik ^2 \frac{\nk ^2}{\nkt ^2\nkc ^2}\cov \Big\{\sum_{i \in [k],  i \in \mathcal S} T_i(X_i-\bar X_{[k]})\mid\mathcal S\Big\}\\
    =& \sumk\Pik ^2 \frac{\nk ^2}{\nkt ^2\nkc ^2}\nkt \frac{\nk-\nkt }{\nk }S_{[k]X \mid\mathcal S}^2=\sumk \Pik ^2\frac{\nk }{\nkt \nkc }S_{[k]X \mid\mathcal S}^2.
\end{align*}
Thus, we have
$$
 \cov ( \hat\tau_X) = \sumk \Pik ^2\frac{\nk }{\nkt \nkc }S_{[k]X}^2.
$$ 
Similarly, we can compute the covariances between $\hat \tau - \tau$, $\hat\tau_X$, and $\hat\delta_W$, and obtain 
\begin{align*}
&\cov (\sqrt n(\hat\tau -\tau,\hat\tau_X^\mathrm{T},\hat\delta_W^\mathrm{T})^\mathrm{T}) = \sumk\Pik^2 \pik^{-1} \cdot \\
&\left( \begin{matrix}\ekt^{-1}S_{[k]1}^2+\ekc ^{-1}S_{[k]0}^2-\fk S_{[k]\tau}^2&&\ekt ^{-1}S_{[k]1,X}+\ekc ^{-1}S_{[k]0,X}&&(1-\fk )S_{[k]\tau,W}\\\ekt ^{-1}S_{[k]X,1}+\ekc ^{-1}S_{[k]X,0}&&(\ekt \ekc )^{-1}S_{[k]X}^2&& 0\\(1-\fk )S_{[k]W,\tau}&& 0&&(1-\fk )S_{[k]W}^2\end{matrix}\right).
\end{align*}
\end{proof}

\subsection{Proof of Theorem~\ref{coro sampling-adj}}
\begin{proof}
By Theorem \ref{theorem S-optimal}, under the SRSRR experiment (or the stratified randomized survey experiment), $$\var(\sqrt n \tauadj)=V_{\tau\tau}-V_{\tau E}V_{EE}^{-1}V_{E\tau}-V_{\tau C}V_{CC}^{-1}V_{C\tau}.$$
For (i), by Proposition~\ref{prop Cov}, we have 
{$$\frac{\partial}{\partial(1/\fk)}V_{\tau\tau}=\Pik f (\ekt^{-1}S_{[k]1}^2+\ekc^{-1} S_{[k]0}^2)$$
$$\frac{\partial}{\partial(1/\fk)}V_{\tau E}=\Pik f S_{[k]\tau E},\quad\frac{\partial}{\partial(1/\fk)}V_{EE}=\Pik f S_{[k] E}^2.$$
$$\frac{\partial}{\partial(1/\fk)}V_{\tau C}=\Pik f (\ekt^{-1}S_{[k]1,C}+\ekc^{-1}S_{[k]0,C}),\quad\frac{\partial}{\partial(1/\fk)}V_{CC}=\Pik f (\ekt\ekc)^{-1}S_{[k] C}^2.$$}
Since $\sumk\Pik \fk=f$, we consider the Lagrange function:
$$\var(\sqrt n\tauadj)+\lambda\Big(\sumk\Pik \fk-f\Big).$$

For the optimal $\{ \fk\}_{k=1}^{\K}$, we have
{\begin{align*}
    0=&\frac{\partial}{\partial (1/f_{[k]})}\Big\{\var(\sqrt n\tauadj)+\lambda(\sumk\Pik \fk-f)\Big\}\\
  =&\Pik f (\ekt^{-1}S_{[k]1}^2+\ekc^{-1} S_{[k]0}^2)-2\Pik f S_{[k]\tau E}V_{EE}^{-1}V_{E\tau}+\Pik f V_{\tau E}V_{EE}^{-1}S_{[k]E}^2V_{EE}^{-1}V_{E\tau}\\
    &-2\Pik f (\ekt^{-1}S_{[k]1,C}+\ekc^{-1}S_{[k]0,C})V_{CC}^{-1}V_{C\tau} \\
    &+\Pik f(\ekt\ekc)^{-1} V_{\tau C}V_{CC}^{-1}S_{[k]C}^2V_{CC}^{-1}V_{C\tau}-\lambda\Pik \fk^2.
\end{align*}}
Hence,  it leads to 
$$\Big(\frac{|\lambda|} f\Big)^{1/2}=\frac{B_k}{\fk}=\frac{\Pik B_k}{\Pik \fk}=\frac{\sumk\Pik B_k}{\sumk \Pik \fk}=\frac{\sumk\Pik B_k}{f}.$$

For (ii), we have
$$\frac{\partial}{\partial(e_{[k]t})}V_{\tau\tau}=-\Pik^2\pik^{-1}  e_{[k]t}^{-2}S_{[k]t}^2$$
$$\frac{\partial}{\partial(e_{[k]t})}V_{\tau C}=-\Pik^2\pik^{-1}  e_{[k]t}^{-2}S_{[k]t,C},\quad \frac{\partial}{\partial(e_{[k]t})}V_{CC}=-\Pik^2\pik^{-1}  e_{[k]t}^{-2}S_{[k]C}^2.$$
Since $\ekt+\ekc=1,$ we consider the Lagrange function
$$\var(\sqrt n \tauadj)+\sumk\lambda_k(\ekt+\ekc-1).$$

For the optimal $\{e_{[k]t}\}_{k=1}^{\K}$, $t=0,1$, we have 
\begin{align*}
    0=&\frac{\partial}{\partial e_{[k]t}}\Big\{\var(\sqrt n \tauadj)+\sumk\lambda_k(\ekt+\ekc-1) \Big\}\\
    =&\frac{\Pik^2}{\pik}(-\frac{S_{[k]t}^2}{e_{[k]t}^2})+2\frac{\Pik^2}{\pik}\frac{S_{[k]t,C}}{e_{[k]t}^2}V_{CC}^{-1}V_{C\tau}-V_{\tau C}V_{CC}^{-1}\frac{\Pik^2}{\pik}\frac{S_{[k]t}^2}{e_{[k]t}^2}V_{CC}^{-1}V_{C\tau}+\lambda_k.
\end{align*}
We have
$$\frac{\ekt}{\ekc}=\Big(\Big|\frac{S_{[k]1}^2+V_{\tau C}V_{CC}^{-1}S_{[k]C}^2V_{CC}^{-1}V_{ C\tau}-2S_{[k]1,C}V_{CC}^{-1}V_{C\tau}}{S_{[k]0}^2+V_{\tau C}V_{CC}^{-1}S_{[k]C}^2V_{CC}^{-1}V_{ C\tau}-2S_{[k]0,C}V_{CC}^{-1}V_{C\tau}}\Big|\Big)^{1/2} = \frac{b_{[k]1}}{b_{[k]0}}.$$
Thus,
$$e_{[k]1} = b_{[k]1}/ (  b_{[k]1} + b_{[k]0} ).$$

\end{proof}

\subsection{Proof of Theorem \ref{thm:TV_bound}}

\begin{proof}
Recall that $$\mathcal{M}=\{(Z,T_{\mathcal S}): \text{~satisfying~}  M_S\le a_S, M_T\le a_T \text{~under~} \text{SRSE} \}$$ 
is all possible sampling and treatment assignment vectors under $\ssrsrr$. Define the acceptable treatment assignment vector when the sampling vector takes value $ {z}$ as
 $$\mathcal{M}_2(z) =\{t: (z, t)\in \mathcal{M}\},$$
 and the set of acceptable sampling under $\ssrsrr$ as
 $$\mathcal{M}_1' =\big\{ {z}: ( {z},  {t})\in \mathcal{M} \text{ for some }  {t} \in \{0,1\}^n \text{ and } \sumikS t_i = \nkt, k=1,\ldots,\K \big\}.$$
Furthermore, if we only consider the sampling stage, we can define $\mathcal{M}_1$ as the sampling indicators set such that the corresponding $M_S\le a_S$.

When $\mathcal{M} = \emptyset$, Theorem \ref{thm:TV_bound} holds.  
Below we consider only the case when $\mathcal{M}\ne \emptyset$. We first consider the difference between $\ssrsrr$ and $\srsrr$ for sampling and assignment vectors in $\mathcal{M}$. 
Applying the property of total variation distance for discrete measures, we need to bound $|\pr(Z=z,T_{\mathcal S}=t)\mid \ssrsrr)-\pr(Z=z,T_{\mathcal S}=t)\mid \srsrr)|$ under $(z,t)\in \mathcal M$ and $(z,t)\notin \mathcal M$, respectively.

For convenience, we denote the combinatorial number $$\binom{\Nk}{\nkt, \nkc}=\frac{\Nk !}{  (\Nk - \nk)!  \nkt !  \nkc !}.$$

For any $(z, t)\in \mathcal{M}$, we can express the conditional probability under $\ssrsrr$
\begin{align*}
    \pr(Z= z, T_{\mathcal{S}} =  {t} \mid \ssrsrr)
    & = 
    \frac{1}{|\mathcal{M}|}
    = 
    \frac{1}{ \prodk \binom{\Nk}{\nkt, \nkc} \cdot \pr(M_T\le a_T, M_S \le a_S)}, 
\end{align*}
and the conditional probability under $\srsrr$
\begin{align*}
    & \quad \ \pr( {Z} =  {z},  {T}_{\mathcal{S}} =  {t} \mid \srsrr)
    \\
    & = \pr( {Z} =  {z} \mid \srsrr) 
    \cdot
    \pr( {T}_{\mathcal{S}} =  {t} \mid  {Z} =  {z}, \srsrr)
    = 
    \frac{1}{|\mathcal{M}_1|} \cdot  \frac{1}{|\mathcal{M}_2( {z})|}
    \\
    & = 
    \frac{1}{ \prodk \binom{\Nk}{\nk} \cdot\pr(M_S\le a_S)} \cdot
    \frac{1}{ \prodk  \binom{\nk}{\nkt} \cdot \pr(M_T \le a_T \mid  {Z} =  {z})}
    \\
    & = 
    \frac{1}{ \prodk \binom{\Nk}{\nkt, \nkc} \cdot \pr(M_S\le a_S) \cdot \pr(M_T \le a_T \mid  {Z} =  {z})}. 
\end{align*}
Therefore, for any $(z,  t)\in \mathcal{M}$, the difference of conditional probability under the two designs can be expressed as  
\begin{align*}
    & \quad \ 
    \pr(Z =  z,  T_{\mathcal{S}} =  t \mid \ssrsrr)
    - 
    \pr(Z =  z,  T_{\mathcal{S}} =  t \mid \srsrr)
    \\
    & = 
    \frac{\pr(M_T \le a_T \mid  Z =  z) - 
    \pr(M_T\le a_T \mid M_S\le a_S)
    }{ \prodk  \binom{\Nk}{\nkt, \nkc} \cdot \pr(M_T\le a_T, M_S \le a_S) \cdot \pr(M_T \le a_T \mid  Z =  z)}.
\end{align*}
Combine the result of all $(z,t)\in \mathcal M$, then we have
\begin{align*}
    & \quad \ 
    \sum_{( {z},  {t})\in \mathcal{M}}
    \left|
    \pr( {Z} =  {z},  {T}_{\mathcal{S}} =  {t} \mid \ssrsrr)
    - 
    \pr( {Z} =  {z},  {T}_{\mathcal{S}} =  {t} \mid \srsrr)
    \right|
    \\
    & = 
    \sum_{ {z}\in \mathcal{M}_1'} |\mathcal{M}_2( {z})| \cdot
    \frac{|\pr(M_T \le a_T \mid  {Z} =  {z}) - 
    \pr(M_T\le a_T \mid M_S\le a_S)|}{ \prodk \binom{\Nk}{\nkt, \nkc} \cdot \pr(M_T\le a_T, M_S \le a_S) \cdot \pr(M_T \le a_T \mid  {Z} =  {z})}
    \\
    & = 
    \sum_{ {z}\in \mathcal{M}_1'}  \prodk  \binom{\nk}{\nkt} \cdot \pr(M_T \le a_T \mid  {Z} =  z) \\
    &~~~\cdot
    \frac{|\pr(M_T \le a_T \mid  {Z} =  {z}) - 
    \pr(M_T\le a_T \mid M_S\le a_S)|}{ \prodk \binom{\Nk}{\nkt, \nkc} \cdot \pr(M_T\le a_T, M_S \le a_S) \cdot \pr(M_T \le a_T \mid  {Z} =  {z})}
    \\
  \end{align*}
  \begin{align*}
    & = 
    \frac{1}{ \pr(M_T\le a_T, M_S\le a_S)} \cdot
    \frac{1}{ \prodk  \binom{\Nk}{\nk}}
    \sum_{ {z}\in \mathcal{M}_1'} 
    |\pr(M_T \le a_T \mid  {Z} =  {z}) \\
    & \quad - 
    \pr(M_T\le a_T \mid M_S\le a_S)|.
    \end{align*}
Because $\mathcal M_1'\subset \{z\in\{0,1\}^N: \sumik z_i=\nk \}$, we have
\begin{align*}
 &\quad \ 
    \sum_{( {z},  {t})\in \mathcal{M}}
    \left|
    \pr( {Z} =  {z},  {T}_{\mathcal{S}} =  {t} \mid \ssrsrr)
    - 
    \pr( {Z} =  {z},  {T}_{\mathcal{S}} =  {t} \mid \srsrr)
    \right|\\
    & \le 
    \frac{1}{ \pr(M_T\le a_T, M_S\le a_S)} \cdot
    \frac{1}{ \prodk  \binom{\Nk}{\nk}}
    \\
    &~~~~\sum_{ {z}\in \{0,1\}^N: \sumik z_i = \nk} 
    \left|\pr(M_T \le a_T \mid  {Z} =  {z}) - 
    \pr(M_T\le a_T \mid M_S\le a_S)\right|
    \\
    & = \frac{1}{ \pr(M_T\le a_T, M_S\le a_S)} \cdot
    E \left| \pr(M_T \le a_T \mid  {Z} ) - 
    \pr(M_T\le a_T \mid M_S\le a_S) \right|. 
\end{align*}

For the difference between $\ssrsrr$ and $\srsrr$ for sampling and assignment vectors not in $\mathcal{M}$, we have
\begin{align*}
    & \quad \ 
    \sum_{( {z},  {t})\notin \mathcal{M}}
    \left|
    \pr( {Z} =  {z},  {T}_{\mathcal{S}} =  {t} \mid \ssrsrr)
    - 
    \pr( {Z} =  {z},  {T}_{\mathcal{S}} =  {t} \mid \srsrr)
    \right|
    \\
    & = 
    \sum_{( {z},  {t})\notin \mathcal{M}}
    \pr( {Z} =  {z},  {T}_{\mathcal{S}} =  {t} \mid \srsrr)
    = 
    1 - \sum_{( {z},  {t})\in \mathcal{M}}
    \pr( {Z} =  {z},  {T}_{\mathcal{S}} =  {t} \mid \srsrr)
    \\
    & = 
    \sum_{( {z},  {t})\in \mathcal{M}}
    \pr( {Z} =  {z},  {T}_{\mathcal{S}} =  {t} \mid \ssrsrr) - \sum_{( {z},  {t})\in \mathcal{M}}
    \pr( {Z} =  {z},  {T}_{\mathcal{S}} =  {t} \mid \srsrr)
    \\
    & \le 
    \sum_{( {z},  {t})\in \mathcal{M}}
    \left|
    \pr( {Z} =  {z},  {T}_{\mathcal{S}} =  {t} \mid \ssrsrr)
    - 
    \pr( {Z} =  {z},  {T}_{\mathcal{S}} =  {t} \mid \srsrr)
    \right|. 
\end{align*}

Combining the result for $(z,t)\in \mathcal M$ and $(z,t)\notin \mathcal M$, we can bound the total variation distance between the two designs as follows:
\begin{align*}
    & \quad \text{d}_{\TV}(\ssrsrr, \srsrr) 
    \\
    & = \frac{1}{2} \sum_{( {z}, {t})\in \mathcal{M} } 
    \left|
    \pr( {Z} =  {z},  {T}_{\mathcal{S}} =  {t} \mid \ssrsrr)
    - 
    \pr( {Z} =  {z},  {T}_{\mathcal{S}} =  {t} \mid \srsrr)
    \right|
    \\
    & \quad \ + 
    \frac{1}{2} \sum_{( {z}, {t})\notin \mathcal{M} } 
    \left|
    \pr( {Z} =  {z},  {T}_{\mathcal{S}} =  {t} \mid \ssrsrr)
    - 
    \pr( {Z} =  {z},  {T}_{\mathcal{S}} =  {t} \mid \srsrr)
    \right|
    \\
    & 
    \le \sum_{( {z}, {t})\in \mathcal{M} } 
    \left|
    \pr( {Z} =  {z},  {T}_{\mathcal{S}} =  {t} \mid \ssrsrr)
    - 
    \pr( {Z} =  {z},  {T}_{\mathcal{S}} =  {t} \mid \srsrr)
    \right|
    \\
    & \le 
    \frac{E | \pr(M_T \le a_T \mid  {Z} ) - 
    \pr(M_T\le a_T \mid M_S\le a_S) |}{ \pr(M_T\le a_T, M_S\le a_S)}.
\end{align*}
Therefore, Theorem \ref{thm:TV_bound} holds. 
\end{proof}

\subsection{Proof of Corollary \ref{coro:d_TV to 0}}
\begin{proof}
    By Theorem \ref{theorem clt-srse},
as $n\rightarrow \infty$, we have 
$\pr(M_S\le a_S, M_T\le a_T) \rightarrow \pr(\chi^2_{J_1}\le a_S) \pr(\chi^2_{J_2} \le a_T) > 0$ 
and 
$\pr(M_S\le a_S) \rightarrow \pr(\chi^2_{J_1}\le a_S)$. 
These imply that 
$\pr(M_T\le a_T \mid M_S\le a_S) \rightarrow \pr(\chi^2_{J_2} \le a_T)$, 
and $\mathcal{M}$ is not an empty set when $N$ is sufficient large. 
Consequently, $1(\mathcal{M} = \emptyset) \rightarrow 0$ as $n\rightarrow \infty$.

When $Z$ is drawn from SRSE, given $\mathcal S$, applying CLT with $f=1$ and considering the whole population as $\mathcal S$, we have $\pr(M_T\le a_T\mid  Z) \xrightarrow{p} \pr(\chi^2_{J_2} \le a_T)$ as $n\rightarrow \infty$. 
Then, we have 
$|\pr(M_T \le a_T \mid Z) - \pr(M_T\le a_T \mid M_S\le a_S) | \xrightarrow{p} 0$.
Because $|\pr(M_T \le a_T \mid  Z) - \pr(M_T\le a_T \mid M_S\le a_S) |$ is upper bounded by 2, by Lebesgue's dominated convergence theorem, as $n\rightarrow \infty$, we have 
\begin{align*}
    E| \pr(M_T \le a_T \mid  {Z} ) - \pr(M_T\le a_T \mid M_S\le a_S) | \rightarrow 0.
\end{align*}

By Theorem \ref{thm:TV_bound}, as $n\rightarrow \infty$,  
\begin{align*}
    \text{d}_{\TV}(\ssrsrr, \srsrr) 
    &\le 
    1(\mathcal{M} = \emptyset) 
    + 
    \frac{E \left| \pr(M_T \le a_T \mid  {Z} ) - 
    \pr(M_T\le a_T \mid M_S\le a_S) \right|}{ \pr(M_T\le a_T, M_S\le a_S)} \\
    &
    \rightarrow 0. 
\end{align*}
Therefore, Corollary \ref{coro:d_TV to 0} holds.
\end{proof}

\end{document}